\documentclass[a4paper,oneside,reqno]{amsart}

\usepackage[english]{babel}
\usepackage[utf8]{inputenc}		
\usepackage{lmodern}
\usepackage[scaled]{helvet}		
\usepackage{courier}			
\usepackage{eulervm}			
\usepackage[bb=boondox,bbscaled=1.05,scr=dutchcal]{mathalfa}	%
\usepackage{dsfont}
\normalfont

\usepackage{indentfirst}				
\usepackage{tabularx,longtable,booktabs,multirow}			
\usepackage{caption,subcaption}			
\usepackage{fullpage}				    
\usepackage[english]{varioref}			
\usepackage[dvipsnames]{xcolor}			
\usepackage{hyperref}		
\usepackage{url}
\usepackage{bookmark}					
\usepackage{textcomp}
\usepackage{rotating,afterpage,pdflscape}
\usepackage{pbox}

\usepackage{graphicx}		
\usepackage{float}
\usepackage{tikz}			
\usepackage{tikz-cd}		
\usetikzlibrary{
  arrows,decorations.pathreplacing,decorations.markings,decorations.pathmorphing,shapes.geometric,positioning,calc,fadings,angles
	}
\tikzfading[name=inner fade, inner color=transparent!0, outer color=transparent!100]
\graphicspath{{Images/}}
\usepackage[export]{adjustbox}

\usepackage{amsmath,amssymb,amsthm,mathtools}
\usepackage{bm,braket,faktor,esint,stmaryrd,bbm}
\usepackage{tcolorbox}
\usepackage[export]{adjustbox}

\definecolor{webbrown}{rgb}{0.65, 0.16, 0.16}

\hypersetup{
colorlinks=true, linktocpage=true, pdfstartpage=1, pdfstartview=FitV,breaklinks=true, pdfpagemode=UseNone, pageanchor=true, pdfpagemode=UseOutlines,plainpages=false, bookmarksnumbered, bookmarksopen=true, bookmarksopenlevel=1,hypertexnames=true, pdfhighlight=/O,urlcolor=webbrown, linkcolor=RoyalBlue, citecolor=ForestGreen}

\usepackage{enumerate,enumitem}
\numberwithin{equation}{section}

\makeatletter
\def\l@subsection{\@tocline{2}{0pt}{2.5pc}{5pc}{}}
\makeatother

\usepackage[style=alphabetic,maxalphanames=5,giveninits,sorting=nyt,%
maxbibnames=99,backend=biber]{biblatex}
\renewbibmacro{in:}{}
\usepackage{csquotes}
\usepackage[noabbrev]{cleveref}
\expandafter\def\csname ver@etex.sty\endcsname{3000/12/31}

\crefname{lemma}{lemma}{lemmata}
\Crefname{lemma}{Lemma}{Lemmata}
\crefname{subsection}{subsection}{subsections}
\Crefname{subsection}{Subsection}{Subsections}
\crefname{conjecture}{conjecture}{conjectures}
\Crefname{conjecture}{Conjecture}{Conjectures}
\crefname{question}{question}{questions}
\Crefname{question}{Question}{Questions}
\crefname{warning}{warning}{warnings}
\Crefname{warning}{Warning}{Warnings}

\newcommand{\de}{\partial}
\newcommand{\N}{\mathbb{N}}
\newcommand{\Z}{\mathbb{Z}}
\newcommand{\Q}{\mathbb{Q}}
\newcommand{\R}{\mathbb{R}}
\newcommand{\C}{\mathbb{C}}
\renewcommand{\P}{\mathbb{P}}
\newcommand{\ch}{\mathrm{ch}}
\newcommand{\bbraket}[1]{\llbracket #1 \rrbracket}
\newcommand{\mf}[1]{\mathfrak{#1}}
\newcommand{\mc}[1]{\mathcal{#1}}

\newcommand{\into}{\hookrightarrow}
\renewcommand{\theta}{\vartheta}
\newcommand*{\iso}{
   \mathrel{\vcenter{\offinterlineskip
   \hbox{\(\hspace{1pt} \sim\)}\vskip-.35ex\hbox{\( \rightarrow \)}}}}
\newcommand{\<}{\langle}
\renewcommand{\>}{\rangle}
\renewcommand{\epsilon}{\varepsilon}
\renewcommand{\phi}{\varphi}
\newcommand{\plh} {\mathord{\cdot}}
\newcommand{\del}{\partial}

\newcommand{\Id}{{\mathord{Id}}}
\newcommand{\normord}[1]{{\vcentcolon}\!\mathrel{#1}\!{\vcentcolon}}
\newcommand{\M}{\overline{\mc{M}}}

\DeclareSymbolFont{greekletters}{OML}{cmr}{m}{it}
\DeclareMathSymbol{\varsigma}{\mathalpha}{greekletters}{"26}

\newcommand{\corr}[1]{\bigg{\langle} \, #1 \,  \bigg{\rangle}}

\DeclareMathOperator{\ad}{ad}

\DeclareMathOperator\SP{\mc{SP}}
\DeclareMathOperator\OP{\mc{OP}}
\newcommand{\Aut}{\mathord{\mathrm{Aut}}}

\DeclareMathOperator{\End}{End}
\DeclareMathOperator{\Spec}{Spec}
\DeclareMathOperator{\Pic}{Pic}

\DeclareMathOperator{\Pf}{Pf}
\DeclareMathOperator*{\Res}{Res}
\DeclareMathOperator{\Ch}{Ch}
\DeclareMathOperator{\Cl}{\mc{C}\ell}
\DeclarePairedDelimiter{\ceil}{\lceil}{\rceil}
\newcommand{\iu}{\mathrm{i}}

\newcommand{\CP}{\mathord{\mathrm{CP}}}
\newcommand{\WC}{\mathord{\mathrm{WC}}}

\theoremstyle{plain}
\newtheorem{theorem}{Theorem}[section]

\newtheorem{proposition}[theorem]{Proposition}
\newtheorem{lemma}[theorem]{Lemma}
\newtheorem{corollary}[theorem]{Corollary}
\newtheorem{conjecture}[theorem]{Conjecture}
\newtheorem*{question*}{\textcolor{BrickRed}{Question}}

\theoremstyle{definition}
\newtheorem{definition}[theorem]{Definition}
\theoremstyle{remark}
\newtheorem{remark}[theorem]{Remark}
\newtheorem{example}[theorem]{Example}

\newcommand{\bvs}[8]{\varsigma ' \left(\begin{smallmatrix} #1 & #2 & #3 & #4 \\ #5 & #6 & #7 & #8 \end{smallmatrix}\right)}

\usepackage{pdftexcmds}

\DeclareFontFamily{U}{cbgreek}{}
\DeclareFontShape{U}{cbgreek}{m}{n}{
        <-6>    grmn0500
        <6-7>   grmn0600
        <7-8>   grmn0700
        <8-9>   grmn0800
        <9-10>  grmn0900
        <10-12> grmn1000
        <12-17> grmn1200
        <17->   grmn1728
      }{}
\DeclareFontShape{U}{cbgreek}{bx}{n}{
        <-6>    grxn0500
        <6-7>   grxn0600
        <7-8>   grxn0700
        <8-9>   grxn0800
        <9-10>  grxn0900
        <10-12> grxn1000
        <12-17> grxn1200
        <17->   grxn1728
      }{}

\DeclareRobustCommand{\qoppa}{%
  \text{\usefont{U}{cbgreek}{\normalorbold}{n}\symbol{19}}%
}

\makeatletter
\newcommand{\normalorbold}{%
  \ifnum\pdf@strcmp{\math@version}{bold}=\z@ bx\else m\fi
}
\makeatother

\title{
	A new spin on Hurwitz theory and ELSV via theta characteristics
}

\author[A.~Giacchetto]{Alessandro Giacchetto}
\address[A.G.]{Max-Planck-Institut f\"ur Mathematik, Vivatsgasse 7, 53111 Bonn, Germany}
\curraddr{Departement Mathematik, ETH Zürich, Rämisstrasse 101, Zürich 8044, Switzerland}
\email{alessandro.giacchetto@math.ethz.ch}
	
\author[R.~Kramer]{Reinier Kramer}
\address[R.K.]{Max-Planck-Institut f\"ur Mathematik, Vivatsgasse 7, 53111 Bonn, Germany}
\curraddr{Department of Mathematical {\&} Statistical Sciences, University of Alberta, 632 CAB, Edmonton, Alberta, Canada, T6G 2G1}
\email{reinier@ualberta.ca}

\author[D.~Lewański]{Danilo Lewański}
\address[D.L.]{%
Inst. Hautes Études Scientifiques (IHES), 35 Route de Chartres, 91440 Bures-sur-Yvette, Paris, France \& Inst. de Physique Théorique (IPhT), Commissariat à l'énergie atomique, Orme des Merisiers Bât. 774, 91191 Gif-sur-Yvette, Paris, France}
\curraddr{Dipartimento di Matematica, Informatica e Geoscienze, Università degli Studi di Trieste, Via Weiss 2, 34128 Trieste, Italy}
\email{danilo.lewanski@units.it}

\subjclass[2010]{14N10, 14H10, 05A15, 14N35, 81R10, 05E10, 20C35}

\makeatletter
\let\runauthor\@author
\let\runtitle\@title
\makeatother

\begin{document}

\begin{abstract}
	We study spin Hurwitz numbers, which count ramified covers of the Riemann sphere with a sign coming from a theta characteristic. These numbers are known to be related to Gromov--Witten theory of Kähler surfaces and to representation theory of the Sergeev group, and are generated by BKP tau-functions. We use the latter interpretation to give polynomiality properties of these numbers and we derive a spectral curve which we conjecture computes spin Hurwitz numbers via a new type of topological recursion. We prove that this conjectural topological recursion is equivalent to an ELSV-type formula, expressing spin Hurwitz numbers in terms of the Chiodo class twisted by the $2$-spin Witten class.
\end{abstract}

\maketitle
\vspace{-.5cm}
\tableofcontents

\section{Introduction}

Over the last decades, a deep interaction between the mathematical physics of integrable systems, the algebraic geometry of moduli spaces of curves, and models arising from theoretical physics has been established. Hurwitz theory provides a rich and insightful variety of examples of this interplay. 

\smallskip

Hurwitz numbers, the counts of covers of a Riemann surface which satisfy given conditions on their ramifications, were defined by Hurwitz~\cite{Hur91,Hur01} and have been studied since then. Through the monodromy representation, they can be seen to count decompositions of the identity in the symmetric group (the Frobenius presentation~\cite{Fro00}) and hence can be given in terms of Schur functions~\cite{Sch01}. In more recent years, Hurwitz numbers have again become an object of interest, due to strong ties with the integrable hierarchies of the Kyoto school (see e.g.~\cite{JM83,MJD00,Oko00}), the intersection theory of the moduli spaces of curves via Ekedahl--Lando--Shapiro--Vainshtein type formulae~\cite{ELSV01}, and topological recursion~\cite{EO07,BM08,BEMS11}.

\smallskip

These results have been generalised in many directions, but the type of Hurwitz numbers considered is essentially the following: the covers counted have target $\P^1$, and specified variable ramifications over one or two points, while all other ramification profiles are described by some uniform condition.

\smallskip

In this paper, we consider a type of Hurwitz numbers called \emph{spin Hurwitz numbers}, introduced by Eskin--Okounkov--Pandharipande~\cite{EOP08}. The defining feature of these numbers is the presence of a spin structure (or theta characteristic) on the source, and the count is weighted by the parity of this theta characteristic. We denote any object related to spin Hurwitz numbers with a superscript $\theta$, to emphasise the role of the theta characteristic.

\smallskip

We clarify that the term `spin Hurwitz numbers' is also used (in e.g.~\cite{MSS13,SSZ15,BKLPS21,KLPS19,DKPS23}) for what we call `Hurwitz numbers with completed cycles' (in fact we consider more properly `spin Hurwitz numbers with completed cycles'). We hope this does not lead to confusion.

\subsubsection*{Results concerning Hurwitz theory}

One of the purposes of this paper is to analyse various features of spin Hurwitz numbers with completed cycles, in analogy with their non-spin counterparts. We summarise here their similarities and their differences.

\smallskip

The double Hurwitz numbers with $(r+1)$-completed cycles are defined for a fixed genus $g$ and partitions $\mu, \nu \vdash d$ as 
\begin{equation}
	h_{g;\mu,\nu}^r = \sum_{[f]} \frac{1}{|\Aut(f)|},
\end{equation}
where the finite summation is taken over isomorphism classes of degree $d$ branched covers $f \colon C \to \P^1$ from a Riemann surface of genus $g$ to the Riemann sphere with ramification profile $\mu$ over zero, $\nu$ over infinity, and all other ramifications are of $(r+1)$-completed cycle type, in number $b = (2g - 2 + \ell(\mu) + \ell(\nu))/r$. The completed cycles can be thought as ramifications of order $r+1$, plus a particular linear combination of lower order ramifications that take into account degenerations of coverings (for a detailed definition, see \cite{OP06,SSZ12}). For example, the case $r=1$ recovers simple ramifications, i.e. represented by a transposition of the cover sheets.

\smallskip

Spin Hurwitz numbers, on the other hand, are defined for a fixed genus $g$ and odd partitions $\mu$, $\nu$ as 
\begin{equation}
	h_{g;\mu,\nu}^{r,\theta}
	=
	\sum_{[f]} \frac{(-1)^{p(f)}}{|\Aut(f)|} ,
\end{equation}
with a similar $f$ summation, although each contribution gets a sign from the parity of a twisted pullback of $ \mc{O}(-1)$ along $f$ (for a precise definition, see \cref{defn:general:spin:HNs,defn:spin:HNs}), and the completion of the $(r+1)$-cycles is defined slightly differently (here, $r$ is required to be even, so that $r+1$ is also odd).

\smallskip

Ordinary double Hurwitz numbers with completed cycles $h_{g;\mu,\nu}^r$ have a formulation as vacuum expectations of bosonic operators formed from charged fermions and acting on the Fock space \cite{Oko00}:
\begin{equation}
	h_{g;\mu,\nu}^r
	=
	\frac{1}{b!} \left\langle{ \,
		\prod_{i=1}^{\ell(\mu)} \frac{\alpha_{\mu_i}}{\mu_i}
		\prod_{k=1}^b \frac{\mc{F}_{r+1}}{r+1}
		\prod_{j=1}^{\ell(\nu)} \frac{\alpha_{-\nu_j}}{\nu_j}
	}\right\rangle,
\end{equation}
which immediately implies that their partition function is a tau-function of the KP hierarchy. Both $\alpha_m$ and $\mc{F}_{r+1}$ are specialisations of some operators $\mc{E}_m(z)$, defined by Okounkov and Pandharipande \cite{OP06}, whose commutation relations are given by
\begin{equation}
	\bigl[ \mc{E}_m(z), \mc{E}_n(w) \bigr] = 2\sinh\bigl( \tfrac{mw - nz}{2} \bigr) \mc{E}_{m+n}(z + w).
\end{equation}
The fact that the algebra of coefficients of the operators $\mc{E}_m(z)$ is closed under commutations is quite a non-trivial fact and it allows explicit computations of the vacuum expectations, showing several peculiar properties of these numbers such as their quasi-polynomiality, their chamber polynomiality, and the computation of their wall-crossing formulae.

\smallskip

Spin Hurwitz numbers $h_{g;\mu,\nu}^{r,\theta}$ are also known to have a vacuum expectation formulation (see e.g. \cite{Lee19})
\begin{equation}
	h_{g;\mu,\nu}^{r,\theta}
	=
  \frac{2^{1-g}}{b!}
	 \left\langle{ \,
		\prod_{i=1}^{\ell(\mu)} \frac{\alpha_{\mu_i}^{B}}{\mu_i}
		\prod_{k=1}^b \frac{\mc{F}_{r+1}^{B}}{r+1}
		\prod_{j=1}^{\ell(\nu)} \frac{\alpha_{-\nu_j}^{B}}{\nu_j}
	}\right\rangle \, ,
\end{equation}
where the bosonic operators are this time given in terms of uncharged fermions and, as a result, the spin Hurwitz numbers partition function is a tau-function of the BKP integrable hierarchy. Analogously to the non-spin case, we define an algebra of operators $\mc{E}_m^{B}(z)$ such that
\begin{equation}
  \mc{F}_{r+1}^{\textup{B}} = (r+1)![z^{r+1}]. \hat{\mc{E}}_0^{\textup{B}}(z),
	\qquad
	\alpha^{B}_n = \mc{E}^{B}_n(0),
\end{equation}
and, most importantly, such that it is closed under commutation:
\begin{equation}
	\bigl[ \mc{E}^{B}_m(z),\mc{E}^{B}_n(w) \bigr]
	=
	\sinh\bigl( \tfrac{mw - nz}{2} \bigr) \mc{E}^{B}_{m+n}(z + w)
	+ (-1)^n
	\sinh\bigl( \tfrac{mw + nz}{2} \bigr) \mc{E}^{B}_{m+n}(z - w).
\end{equation}

 This allows us to unveil similar polynomiality properties.

\begin{theorem}[\Cref{cor:quasi-poly,thm:piecewise,thm:wall:crossing}]
	Single spin Hurwitz numbers are quasi-po\-ly\-no\-mi\-al; spin double Hurwitz numbers are piecewise polynomials with explicit wall-crossing formulae.
\end{theorem}	 

\subsubsection*{Results concerning topological recursion}

The method of topological recursion \cite{EO07} can be thought of as an algorithm which recursively generates solutions to enumerative problems from the initial data of a spectral curve (see \cref{sec:TR:CohFTs}). It produces a standard output of multidifferentials $\omega_{g,n}$ that encode the numbers of interest as expansion coefficients near a particular point on the curve.

\smallskip

For ordinary single Hurwitz numbers with completed cycles, the following result is known.

\begin{theorem}[\cite{SSZ15,DKPS23}]
	The spectral curve $(\Sigma,x,y,B)$ given by $\Sigma = \P^1$
	\begin{equation}
		x(z) = \log(z) - z^{r}, 
		\qquad\quad 
		y(z) = z \,,
		\qquad\quad
		B(z_1, z_2) 
		=
		\frac{dz_1 dz_2}{(z_1 - z_2)^2} 
	\end{equation}
	generates via topological recursion the numbers $h_{g;\mu}^{r}$ as
	\begin{equation}
		\omega_{g,n}(z_1,\dots,z_n)
		=
		\sum_{\mu_1,\dots,\mu_n} h_{g;\mu}^{r} \prod_{i=1}^n \mu_i e^{\mu_i x(z_i)} dx(z_i) \, .
	\end{equation}
\end{theorem}

In this paper, we find a conjectural spectral curve for spin single Hurwitz numbers (\cref{conj:spin:HNs:TR}).

\begin{conjecture}\label{intro:conj:spin:HNs:TR}
	Let $r$ be a positive even integer. The spectral curve $(\Sigma,x,y,B)$ given by $\Sigma = \P^1$
	\begin{equation}
		x(z) = \log(z) - z^{r}, 
		\qquad\quad 
		y(z) = z \,,
		\qquad\quad
		B(z_1, z_2) 
		=
		\frac{1}{2} \left( \frac{1}{(z_1 - z_2)^2} + \frac{1}{(z_1 + z_2)^2} \right) dz_1 dz_2
	\end{equation}
	generates via topological recursion the numbers $h_{g;\mu}^{r,\theta}$ as
	\begin{equation}
		\omega_{g,n}(z_1,\dots,z_n)
		=
		\sum_{\substack{\mu_1,\dots,\mu_n \\ \textup{odd}}} h_{g;\mu}^{r,\theta} \prod_{i=1}^n \mu_i e^{\mu_i x(z_i)} dx(z_i) \, .
	\end{equation}
\end{conjecture}

After we communicated this conjecture to them, Alexandrov and Shadrin soon proved it.

\begin{theorem}[{\cite[Theorem~7.1]{AS23}}]\label{intro:thm:AS}
	\Cref{intro:conj:spin:HNs:TR} holds.
\end{theorem}

This conjectural spectral curve does not satisfy one canonical requirement imposed by the theory of topological recursion in the classical sense: the $(0,2)$-correlator $B(z_1, z_2)$ has poles when $z_1$ and $z_2$ approach two distinct ramification points of $x$. However, the way these extra poles arise encode in a beautiful way some extra structure of this curve. First of all, these new poles are of order two and with the same biresidue, as in the case when both variables approach the same critical point. More importantly, the number of critical points is even and they come in pairs because of an underlying $\Z/2\Z$-action: the double poles only arise for $z_2$ approaching either the same critical point $z_1$, or its conjugate with respect to the group action. We then realise the quotient of the conjectural spectral curve modulo the group action, reducing by half the number of ramification points. Surprisingly, we find that the correlators $\omega_{g,n}$ of the quotient spectral curve differ by the initial correlators just by some simple prefactor.

\smallskip

We export this principle to the more general setting of spectral curves with a finite group acting on them.

\subsubsection*{Results concerning algebraic geometry of moduli spaces of curves}

There is a general statement of topological recursion at the service of algebraic geometry (\cref{thm:Eyn:DOSS}): any enumerative problem generated by topological recursion has a representation in terms of the intersection theory on the moduli spaces of curves. In general, the intersection class is described quite abstractly as the action of an element of the Givental group on a base point given by $1$, but in particular cases, it turns out to also have a more geometric interpretation related to a specific moduli problem.

\smallskip

In the case of ordinary Hurwitz numbers with completed cycles, a cohomological representation motivated by Gromov--Witten theory was proposed by Zvonkine \cite{Zvo06}.

\begin{theorem}[\cite{SSZ15,DKPS23}]\label{thm:Zvonkine:conj}
	Single Hurwitz numbers with $(r+1)$-completed cycles are given by
	\begin{equation}\label{eqn:Zvonkine:conj}
		h_{g;\mu}^{r}
		=
		r^{2g-2+n+b}
		\left( \prod_{i=1}^n \frac{\left( \frac{\mu_i}{r} \right)^{[\mu_i]}}{[\mu_i]!} \right)
		\int_{\overline{\mathcal{M}}_{g,n}}
			\frac{C(r,1;\braket{\mu})}{\prod_{i=1}^n(1 - \frac{\mu_i}{r}\psi_i)} \, ,
	\end{equation}
	where $n = \ell(\mu)$, $b = (2g-2+n+|\mu|)/r$, $\mu_i = r[\mu_i] + r- \langle \mu_i \rangle$, and most importantly $C$ is the Chiodo cohomological field theory \cite{Chi08}.
\end{theorem}

\begin{remark}
  For $r=1$ the Chiodo class becomes the total Chern class of the dual Hodge bundle of abelian differentials: $C(1,1;\braket{\mu}) = \Lambda(-1) = 1 - \lambda_1 + \lambda_2 - \dots + (-1)^g \lambda_g$. In this case, \cref{eqn:Zvonkine:conj} reduces to the ELSV formula~\cite{ELSV01}, the first formula relating Hurwitz theory with intersections on moduli of curves.
\end{remark}

We exploit the relation between Givental group action and topological recursion to compute the cohomological field theory representing spin Hurwitz numbers.

\begin{theorem}\label{intro:spin:Zvonkine:conj}
	\Cref{intro:conj:spin:HNs:TR} is equivalent to
	\begin{equation}
		h_{g;\mu}^{r,\theta}
		=
		2^{1-g} r^{2g-2+n+b}
		\left( \prod_{i=1}^n \frac{\left( \frac{\mu_i}{r} \right)^{[\mu_i]}}{[\mu_i]!} \right)
		\int_{\overline{\mathcal{M}}_{g,n}}
		\frac{C^{\theta}(r,1;\braket{\mu})}{\prod_{i=1}^n(1 - \frac{\mu_i}{r}\psi_i)} \, ,
	\end{equation}
	where $n = \ell(\mu)$, $b = (2g-2+n+|\mu|)/r$, $\mu_i = r[\mu_i] + r - (2\langle \mu_i \rangle +1)$, and most importantly $C^{\theta}$ is the Chiodo cohomological field theory twisted by the Witten $2$-spin class $c_{\textup{W},2}$ (see \cref{cor:spin:ELSV:with:Chiodo} for the precise definition).\par
	By the result of \cite{AS23}, \cref{intro:thm:AS}, the formula holds unconditionally.
\end{theorem}
		
It is worth remarking that the product of the Chiodo class and the Witten $2$-spin class has a particularly clean expression as a sum over stable graphs, and we use this expression in the proof.

\subsubsection*{Results concerning Gromov--Witten theory of Kähler surfaces}

Hurwitz numbers are intrinsically related to topological string theory. Simple Hurwitz numbers arise as a framing limit of the topological vertex and, for general $r$, completed cycles Hurwitz numbers encode the Gromov--Witten theory of $\P^1$ via the celebrated Gromov--Witten/Hurwitz (GW/H) correspondence of Okounkov--Pandharipande \cite{OP06}.

\smallskip

For instance, the statement of Zvonkine's conjecture is equivalent to the following relation with connected descendant GW invariants of equivariant $\P^1$ relative to a partition $\mu$:
\begin{equation}
	\corr{
    \frac{(r! \, \tau_{r}(\omega ))^b}{b!},\mu
  	}^{\P^1}_{|\mu|} 
	=
	r^{2g-2+n+b}
	\left( \prod_{i=1}^n \frac{\left( \frac{\mu_i}{r} \right)^{[\mu_i]}}{[\mu_i]!}\right)
	\int_{\overline{\mathcal{M}}_{g,n}}
	\frac{C_{g,n}(r,1;\braket{\mu})}{\prod_{i=1}^n(1 - \frac{\mu_i}{r}\psi_i)} \, ,
\end{equation}
where $ \omega$ is the class dual to a point.

\smallskip

On the other hand, the GW theory of spin curves, i.e. curves together with a theta characteristic, is believed to correspond to spin Hurwitz numbers via a spin analogue of the GW/H correspondence \cite{MP08,Lee19}. Base cases of this correspondence have been proved for degree $1$ and $2$, and conjectured for $d\geq 3$.

\smallskip

The motivation to understand this correspondence finds its roots in the GW theory of Kähler surfaces with a smooth canonical divisor, which has been proved to be entirely determined by the GW theory of spin curves, see \cite{LP07,LP13}. Assuming the spin GW/H correspondence, our results on the moduli spaces of curves imply that the diagonal connected descendent GW invariants of $(\P^1,\mc{O}(-1))$ are given by
\begin{equation}\label{eqn:spin:GW:H}
	\corr{
   	 \frac{( (-1)^{r/2} 2^r (r-1)!! \, \tau_{r/2}(\omega) )^b}{b!},\mu
 	 }^{\P^1,+}_{|\mu|}
	=
	2^{1-g} r^{2g-2+n+b}
	\left( \prod_{i=1}^n \frac{\left( \frac{\mu_i}{r} \right)^{[\mu_i]}}{[\mu_i]!}\right)
	\int_{\overline{\mathcal{M}}_{g,n}}
	\frac{C_{g,n}^{\theta}(r,1;\braket{\mu})}{\prod_{i=1}^n(1 - \frac{\mu_i}{r}\psi_i)} \, .
\end{equation}
The twist obtained by intersecting with the Witten class in the GW potential for Kähler targets already appeared e.g. in \cite{JKV01,CZ09}. However, both these appearances in a sense take place on the Gromov--Witten side of the GW/H correspondence, whereas our result reveals the Witten class on the Hurwitz theory side.

\subsection{Related developments since prepublication}

As stated before, Alexandrov--Shadrin \cite{AS23} proved \cref{intro:conj:spin:HNs:TR} soon after we stated it. The techniques they use are very different from the considerations we give in this paper that led to our conjecture.

\smallskip

The main strength of this article relies instead in the analysis of one particular enumerative problem, the spin Hurwitz numbers, and its realisation and novelty in different fields: the appearance of an equivariant spectral curve in topological recursion theory, the appearance of a spin analogous of the Okounkov--Pandharipande algebra of operators, and the cohomological representation (ELSV-type formula) of these numbers by a peculiar class $C^{\vartheta}(r,k;A) \in H^{\bullet}(\M_{g,n})$, obtained from theta characteristics weighted with their parity on the space of spin curves.

\smallskip

In fact,  the class $C^{\vartheta}(r,k;A)$ is related to another recent work.  For a fixed value of $k$ and $A$, the Chiodo class $C(r,k;A)$ (after multiplying by a certain power of $r$) is a polynomial in $r$ for larges values of $r$. The degree zero of this polynomial is denoted by $DR_g(k,A)$. The restriction to $\mc{M}_{g,n}$ of the Poincaré dual of this class is the locus
$
\mc{M}_{g}(k,A) = \{ (C,x_1,\dotsc,x_n) \; {|} \; \omega_{\log}^{\otimes k} \cong \mathcal{O}(\sum_i a_i x_i ) \} . 
$

It is expected that the classes $C(r,k;A)^\theta$ satisfy the same polynomiality in $r$, thus allowing one to define a class $DR_{g}(k,A)^\theta$. In \cite{CSS21},  the authors propose some conjectural properties of this class allowing them to compute the intersection numbers $DR_{g}(k,A)^\theta \cdot \psi_{1}^{2g-3+n}$. In particular, if $k$ and the parts of $A$ are odd,  then the space $\mathcal{M}_{g}(k,A)$ splits into components with a constant parity of the spin structure $\mathcal{O}(\sum_i \frac{(a_i-1)}{2} x_i ) \otimes \omega_C^{(1-k)/2}$. Costantini--Sauvaget--Schmitt conjecture that the restriction of $DR_{g}(k,A)^\theta$ to $\mathcal{M}_{g,n}$ is a sum of the classes of components of $\mc{M}_{g}(k,A)$ with a sign determined by the parity.

\smallskip

In \cite{GKLS22}, Sauvaget and the authors of this paper proved the spin Gromov--Witten/Hurwitz correspondence for $ \P^1$. \Cref{intro:spin:Zvonkine:conj} is an essential step in the proof, required for interpreting the localised intersection theory involved in spin Gromov--Witten theory in terms of spin Hurwitz numbers. The correspondence for positive-genus curves is still open.

\subsection{Open questions}

This paper gives a conjectural example of topological recursion related to a B-type problem. We do not know whether other B-type enumerative problems satisfy a form of topological recursion similar to our case. This is a possible direction for future research.

\smallskip

Lee \cite[Theorem~1.1]{Lee19} has found that a generating series for $3$-completed spin double Hurwitz numbers squares to the generating series of $3$-completed non-spin double Hurwitz numbers (after some tuning of the weights). A natural question would be whether this generalises to higher $r$.

\smallskip

A more outlandish question is inspired by \cref{intro:spin:Zvonkine:conj}, cf. \cref{p-SpinChiodo?}: the intersection class is the usual Chiodo class twisted by the Witten $2$-spin class. In this particular case, Witten's class is zero-dimensional, and gives the parity of the associated theta characteristic. What if we twisted the Chiodo class by the Witten $p$-spin class for $ p >2$? This would require $ p \mid r$, and furthermore $r \nmid \mu_i $. But because this class is no longer zero-dimensional, how should we define the associated Hurwitz theory? Does it satisfy some integrable hierarchy?

\smallskip

Witten class, or Witten class twisted by Chiodo class, is shown to play a key role in relation with the Gromov--Witten theory of Kähler targets \cite{CZ09,JKV01}. The theory for Kähler surfaces is determined by the theory of spin curves, which are expected to correspond to spin Hurwitz theory as explained above. This work shows for the first time the appearance of Witten class (or more precisely Chiodo class twisted by Witten class) on the other side of the GW/H correspondence. The open problem is to give a natural and geometric explanation of its appearance.

\addtocontents{toc}{\protect\setcounter{tocdepth}{1}}
\subsection*{Outline of the paper}

The article is divided into three parts.

\smallskip

\emph{Background material.} In \cref{sec:TR:CohFTs} we review the definition of topological recursion, as well as its identification with the Givental group action on cohomological field theories. In \cref{sec:spin:algebra} we review the representation theory of the spin algebras, which allows to represent spin Hurwitz numbers in terms of characters of the Sergeev group in \cref{sec:spin:HNs}.

\smallskip

\emph{Properties of double spin Hurwitz numbers: polynomiality structure and wall-crossing formulae.} In \cref{sec:CJ:OP:operators} we derive the spin analogue of the Okounkov--Pandharipande operators on uncharged fermions, which is then employed in \cref{sec:properties:spin:HNs} for the analysis of the polynomiality structure of spin double Hurwitz numbers and their wall-crossing formulae. 

\smallskip

\emph{Properties of single spin Hurwitz numbers: topological recursion and ELSV formula.} \Cref{sec:conj} contains our main conjecture: single spin Hurwitz numbers are generated by topological recursion for a specific spectral curve. We also give evidence for this conjecture. Since the conjectural spectral curve differs from the usual definition, in \cref{sec:equiv:TR} we define and analyse $G$-quotients of spectral curves, and reduce them to the usual setting of topological recursion. We then employ the correspondence with cohomological field theories in \cref{sec:spin:HNs:ELSV} to derive the representation of spin Hurwitz numbers as intersection numbers on $\M_{g,n}$.

\smallskip

We append two sections supporting the main body of the work containing, respectively, useful computations with uncharged fermions to build the bosonic operators and numerics generated by the implementation of topological recursion.

\subsection*{Notation}

We work over the complex numbers. For a natural number $n$, we write $ \bbraket{n} \coloneqq \{ 1, \dotsc, n\}$. We also define the following functions:
\begin{equation}
  \varsigma (z) = 2\sinh \left(\frac{z}{2}\right),
  \qquad 
  \mc{S}(z) = \frac{\sinh (\frac{z}{2})}{(\frac{z}{2})},
  \qquad 
  \qoppa(z) = \frac{\cosh(\frac{z}{2})}{2},
  \qquad
  \mc{K}(z) = \frac{\cosh(\frac{z}{2})}{(2z)}.
\end{equation}

\subsection*{Acknowledgements}
\addtocontents{toc}{\protect\setcounter{tocdepth}{2}}

We would like to thank Alexander Alexandrov, Gaëtan Borot, Alessandro Chi\-o\-do, Bertrand Eynard, John Harnad, Junho Lee, Kento Osuga, Adrien Sauvaget, Johannes Schmitt, Sergey Shadrin, and Dimitri Zvonkine for useful discussions.
A.~G. and R.~K. were supported by the Max-Planck-Gesell\-schaft. D.~L. was supported by the ERC grant described below and by the grant host institutes IHES and IPhT. D.~L. was moreover supported by the INdAM group GNSAGA for research visits. This paper is partly a result of the ERC-SyG project, Recursive and Exact New Quantum Theory (ReNewQuantum) which received funding from the European Research Council (ERC) under the European Union's Horizon 2020 research and innovation programme under grant agreement No 810573.

\section{Topological recursion and cohomological field theories}
\label{sec:TR:CohFTs}

In this section, we recall the basic theory of topological recursion (TR), first developed by Eynard--Oran\-tin~\cite{EO07}, and in particular its relation to cohomological field theories (CohFTs). This relation, given by Dunin--Barkowski--Orantin--Shadrin--Spitz and Eynard~\cite{DOSS14,Eyn11}, identifies the graph expansion of TR with the Givental--Teleman action~\cite{Giv01,Tel12} on semisimple CohFTs.

\subsection{Topological recursion}

Topological recursion is a universal procedure that associates a collection of symmetric multidifferentials to a spectral curve, a curve with some extra data. What makes TR especially useful is its applications to enumerative geometry: many counting problems are solved by TR, in the sense that the sought numbers are Taylor coefficients of the multidifferentials when expanded around a specific point.

\begin{definition}[\cite{EO07}]\label{defn:spectral:curve}
	A \emph{spectral curve} $\mathcal{S} = (\Sigma,x,y,B)$ consists of
	\begin{itemize}
		\item a Riemann surface $\Sigma$ (not necessarily compact or connected);
		\item a function $x \colon \Sigma \to \C$ such that its differential $dx$ is meromorphic and has finitely many zeros (called ramification points) $a_1,\dots,a_r$ that are simple;
		\item a meromorphic function $y \colon \Sigma \to \C$ such that $dy$ does not vanish at the zeros of $dx$;
		\item a symmetric bidifferential $B$ on $\Sigma \times \Sigma$, with a double pole on the diagonal with biresidue $1$, and no other poles.
	\end{itemize}
	The \emph{topological recursion} produces symmetric multidifferentials (also called \emph{correlators}) $\omega_{g,n}$ on $\Sigma^n$, defined recursively on $2g-2+n > 0$ as
	\begin{equation}\label{eqn:TR}
	\begin{split}
		\omega_{g,n}(z_1,\dots,z_n) \coloneqq \sum_{i=1}^r \Res_{z = a_i} K_i(z_1,z) \bigg( &
			\omega_{g-1,n+2}(z,\sigma_i(z),z_2,\dots,z_n) \\
			& +
			\sum_{\substack{g_1+g_2 = g \\ J_1 \sqcup J_2 = \{2,\dots,n\}}}^{\text{no $(0,1)$}}
				\omega_{g_1,1+|J_1|}(z,z_{J_1})
				\omega_{g_2,1+|J_2|}(\sigma_i(z),z_{J_2})
		\bigg),
	\end{split}
	\end{equation}
	where $K_i$, called the topological recursion kernels, are locally defined in a neighbourhood $U_i$ of $a_i$ as
	\begin{equation}\label{eqn:TR:kernel}
		K_i(z_1,z) \coloneqq \frac{\frac{1}{2} \int_{w = \sigma_i(z)}^z B(z_1,w)}{\bigl( y(z) - y(\sigma_i(z)) \bigr) d x(z)},
	\end{equation}
	and $\sigma_i \colon U_i \to U_i$ is the Galois involution near the ramification point $a_i \in U_i$. It can be shown that $\omega_{g,n}$ is a symmetric meromorphic multidifferential on $\Sigma^n$, with poles only at the ramification points.
\end{definition}


\begin{remark}\label{rem:B:rescaling:local:TR}
	The topological recursion depends only on the local behaviour of the spectral curve around the ramification points. In particular, the bidifferential $B$ is allowed to have more poles, \emph{as long as there are no poles when the two arguments approach different ramification points}. More formally, given a spectral curve as in \cref{defn:spectral:curve}, we get $r$ formal neighbourhoods $ U_i \coloneqq \Spec \C \bbraket{z-a_i}$ of the ramification points, and the topological recursion correlators only depends on the \emph{local spectral curve} $ \mc{S}_{\textup{loc}} = (U = \bigsqcup_{i=1}^r U_i, x |_U, y|_U, B|_U)$. The requirements in \cref{defn:spectral:curve} need only hold for $\mc{S}_{\textup{loc}}$.\par
	Moreover, the bidifferential $B$ is allowed to have double pole on the diagonal with arbitrary non-vanishing biresidue. A rescaling of the form $B \mapsto \lambda B$ corresponds to a rescaling of the correlators: $\omega_{g,n} \mapsto \lambda^{3g-3+2n} \omega_{g,n}$.
\end{remark}

\subsection{Cohomological field theories and Givental action}

Let us recall some facts about the Givental action on CohFTs. More details can be found in \cite{PPZ15}.

\begin{definition}[\cite{KM94}]
	Let $V$ be a finite dimensional $\C$-vector space with a non-degenerate symmetric $2$-form $\eta$. A \emph{cohomological field theory} on $(V,\eta)$ consists of a collection $\Omega = (\Omega_{g,n})_{2g-2+n>0}$ of elements
	\begin{equation}
		\Omega_{g,n} \in H^{\bullet}(\overline{\mathcal{M}}_{g,n}) \otimes (V^{\ast})^{\otimes n}
	\end{equation}
	satisfying the following axioms.
	\begin{enumerate}
		\item[i)] Each $\Omega_{g,n}$ is $\mf{S}_n$-invariant, where the action of the symmetric group $\mf{S}_n$ permutes both the marked points of $\overline{\mathcal{M}}_{g,n}$ and the copies of $(V^{\ast})^{\otimes n}$.
		\item[ii)] Consider the gluing maps
		\begin{equation}
		\begin{aligned}
			q \colon& \overline{\mathcal{M}}_{g-1,n+2} \longrightarrow \overline{\mathcal{M}}_{g,n}, \\
			r_{h,I} \colon& \overline{\mathcal{M}}_{h,|I|+1} \times \overline{\mathcal{M}}_{g-h,|J|+1}  \longrightarrow \overline{\mathcal{M}}_{g,n},
			\qquad
			I \sqcup J = \set{1,\dots,n}.
		\end{aligned}
		\end{equation}
		Then
		\begin{equation}
		\begin{aligned}
			q^{\ast}\Omega_{g,n}(v_1 \otimes \cdots \otimes v_n)
			& =
			\Omega_{g-1,n+2}(v_1 \otimes \cdots \otimes v_n \otimes \eta^{\dag}), \\
			r_{h,I}^{\ast} \Omega_{g,n}(v_1 \otimes \cdots \otimes v_n)
			& =
			(\Omega_{h,|I|+1} \otimes \Omega_{g-h,|J|+1}) \Biggl( \bigotimes_{i \in I} v_i \otimes \eta^{\dag} \otimes \bigotimes_{j \in J} v_j \Biggr)
		\end{aligned}
		\end{equation}
		where $\eta^{\dag} \in V^{\otimes 2}$ is the bivector dual to $\eta$.
	\end{enumerate}
	If the vector space comes with a distinguished element $\mathbb{1} \in V$, we can also ask for a third axiom:
	\begin{enumerate}
		\item[iii)] Consider the forgetful map
		\begin{equation}
			p \colon \overline{\mathcal{M}}_{g,n+1} \longrightarrow \overline{\mathcal{M}}_{g,n}.
		\end{equation}
		Then
		\begin{equation}
		\begin{aligned}
			p^{\ast} \Omega_{g,n}(v_1 \otimes \cdots \otimes v_n)
			& =
			\Omega_{g,n}(v_1 \otimes \cdots \otimes v_n \otimes \mathbb{1}), \\
			\Omega_{0,3}(v_1 \otimes v_2 \otimes \mathbb{1})
			& =
			\eta(v_1,v_2).
		\end{aligned}
		\end{equation}
	\end{enumerate}
	In this case, $\Omega$ is called a \emph{cohomological field theory with flat unit}.
\end{definition}

A cohomological field theory defines a quantum product $\star$ on $V$:
\begin{equation}
	\eta(v_1 \star v_2,v_3) = \Omega_{0,3}(v_1 \otimes v_2 \otimes v_3).
\end{equation}
Commutativity follows from (i), associativity from (ii). If the cohomological field theory has a flat unit, the quantum product is unital, with $\mathbb{1} \in V$ being the identity by (iii).

\smallskip

The degree $0$ part of a CohFT
\begin{equation}
	\deg_0{\Omega_{g,n}} \in H^{0}(\overline{\mathcal{M}}_{g,n}) \otimes (V^{\ast})^{\otimes n} \cong (V^{\ast})^{\otimes n}
\end{equation}
is a $2d$ topological field theory (TFT), and is uniquely determined by the values of $\deg_0{\Omega_{0,3}}$ and by the bilinear form $\eta$ (or equivalently, by the associated quantum product). In particular, we will refer to a TFT as being unital and/or semisimple when the associated algebra is so.

\smallskip

In \cite{Giv01}, Givental defined a certain action on Gromov--Witten potentials, and this action was lifted to CohFTs in the works of Teleman \cite{Tel12} and Shadrin \cite{Sha09}. Here we recall the basic definitions.

\smallskip

{\sc $R$-matrix action.}
Fix a vector space $V$ with a non-degenerate symmetric bilinear form $\eta$. An $R$-matrix is an element $ R(u) = \Id + \sum_{k=1}^\infty R_k u^k \in \Id + u \End (V) \bbraket{u}$ satisfying the symplectic condition:
\begin{equation}
	R(u) R^{\dag}(-u) = \Id.
\end{equation}
Here $R^{\dag}$ is the adjoint with respect to $\eta$. The inverse matrix $R^{-1}(u)$ also satisfies the symplectic condition. In particular, we can consider the $V^{\otimes 2}$-valued power series
\begin{equation}
	E(u,v)
	\coloneqq
	\frac{\Id \otimes \Id - R^{-1}(u) \otimes R^{-1}(v)}{u + v} \eta^{\dag} \in V^{\otimes 2}\bbraket{u,v}.
\end{equation}
We will write $E(u,v) = \sum_{k,\ell \ge 0} E_{k,\ell} u^k v^{\ell}$, with $E_{k,\ell} \in V^{\otimes 2}$.

\smallskip

We can write the above condition in components. Fix a basis $(e_i)$ of $V$, denote by $(e^i)$ the dual basis of $V^{\ast}$ and by $\braket{\,,\,} \colon V^{\ast} \times V \to \C$ the canonical pairing. The symplectic condition and the bivector $E$ can be written as
\begin{equation}
	\sum_{a,b} R^{i}_{a}(u) \eta^{a,b} R^{j}_{b}(-u) = \eta^{i,j},
	\qquad
	E(u,v) = \sum_{i,j} \frac{\eta^{i,j} - \sum_{a,b}(R^{-1})^{i}_a(u) \eta^{a,b} (R^{-1})^{j}_{b}(v)}{u + v} e_{i} \otimes e_j,
\end{equation}
where $R^{k}_{\ell}(u) = \braket{e^k,R(u) e_{\ell}}$ and $\eta^{\dag} = \sum_{i,j} \eta^{i,j} e_i \otimes e_j$. We will denote by $E^{i,j}(u,v) = \sum_{k,\ell \ge 0} E^{i,j}_{k,\ell} u^{k} v^{\ell}$ the formal power series
\begin{equation}
	E^{i,j}(u,v) = \frac{\eta^{i,j} - \sum_{a,b}(R^{-1})^{i}_a(u) \eta^{a,b} (R^{-1})^{j}_{b}(v)}{u + v}
	\in
	\C\bbraket{u,v}.
\end{equation}

\begin{definition}
	Consider a CohFT $\Omega$ on $(V,\eta)$, together with an $R$-matrix. We define a collection of cohomology classes
	\begin{equation}
		R\Omega_{g,n} \in H^{\bullet}(\overline{\mathcal{M}}_{g,n}) \otimes (V^{\ast})^{\otimes n}
	\end{equation}
	as follows. Let $\mathcal{G}_{g,n}$ be the finite set of stable graphs of genus $g$ with $n$ leaves (we refer to \cite{PPZ15} for definition and notation). For each $\Gamma \in \mathcal{G}_{g,n}$, define a contribution	$\mathrm{Cont}_{\Gamma} \in H^{\bullet}(\overline{\mathcal{M}}_{\Gamma}) \otimes (V^{\ast})^{\otimes n}$ by the following construction:
	\begin{itemize}
		\item place $\Omega_{g(v),n(v)}$ at each vertex $v$ of $\Gamma$;
		\item place $R^{-1}(\psi_{\lambda})$ at each leaf $\lambda$ of $\Gamma$;
		\item at every edge $e = (h,h')$ of $\Gamma$, place $E(\psi_{h},\psi_{h'})$;
		\item for each half-edge, contract the contributions from the adjacent vertex and edge or leaf.
	\end{itemize}
	Define $R\Omega_{g,n}$ to be the sum of contributions of all stable graphs:
	\begin{equation}
		R\Omega_{g,n} \coloneqq \sum_{\Gamma \in \mathcal{G}_{g,n}} \frac{1}{|\mathrm{Aut}(\Gamma)|} \, \xi_{\Gamma,\ast}\mathrm{Cont}_{\Gamma}.
	\end{equation}
	Here, $\xi_\Gamma$ denotes the natural composition of gluing maps with image the stratum indexed by $ \Gamma$.
\end{definition}

\begin{proposition}
	The data $R\Omega = (R\Omega_{g,n})_{2g-2+n > 0}$ form a CohFT on $(V,\eta)$. Moreover, the $R$-matrix action on CohFTs is a left group action.
\end{proposition}

{\sc Translations.}
There is also another action on the space of CohFTs: a translation is a $V$-valued power series vanishing in degree $0$ and $1$:
\begin{equation}
	T(u) = \sum_{d \ge 1} T_d u^{d+1},
	\qquad
	T_d \in V.
\end{equation}

\begin{definition}
	Consider a CohFT $\Omega$ on $(V,\eta)$, together with a translation $T$. We define a collection of cohomology classes
	\begin{equation}
		T\Omega_{g,n} \in H^{\bullet}(\overline{\mathcal{M}}_{g,n}) \otimes (V^{\ast})^{\otimes n}
	\end{equation}
	by setting
	\begin{equation}
		T\Omega_{g,n}(v_1 \otimes \cdots \otimes v_n)
		\coloneqq
		\sum_{m \ge 0} \frac{1}{m!} p_{m,\ast} \Omega_{g,n+m} \bigl(
				v_1 \otimes \cdots \otimes v_n \otimes T(\psi_{n+1}) \otimes \cdots \otimes T(\psi_{n+m})
			\bigr).
	\end{equation}
	Here $p_{m} \colon \overline{\mathcal{M}}_{g,n+m} \to \overline{\mathcal{M}}_{g,n}$ is the map forgetting the last $m$ marked points. The vanishing of $T$ in degree $0$ and $1$ ensures that the above sum is actually finite, for dimension reasons.
\end{definition}

\begin{proposition}
	The data $T\Omega = (T\Omega_{g,n})_{2g-2+n > 0}$ form a CohFT on $(V,\eta)$. Moreover, translations form an abelian group action on CohFTs.
\end{proposition}

In general, if we act by translation on a unital TFT on $(V,\eta,\mathbb{1})$, we can express the result as multiplication by an exponential of $\kappa$-classes.

\begin{lemma}\label{lem:translation:hat}
	Let $\varpi$ be a unital TFT on $(V,\eta,\mathbb{1})$. Let $T$ be a translation, and define $\hat{T}(u) = \sum_{m \ge 1} \hat{T}_{m} u^{m} \in u V\bbraket{u}$ by setting
	\begin{equation}
		T(u) \eqqcolon u\left( \mathbb{1} - \exp\bigl(- \hat{T}(u)\bigr) \right).
	\end{equation}
	Here $\exp(-\hat{T}(u)) = \sum_{k \ge 0} \frac{(-1)^k}{k!} \hat{T}(u)^{\star k}$. Then the following relation holds on $H^{\bullet}(\overline{\mathcal{M}}_{g,n}) \otimes (V^{\ast})^{\otimes n}$:
	\begin{equation}
		T\varpi_{g,n} = \varpi_{g,n} \cdot \exp\Biggl( \sum_{m \ge 1} \hat{T}_{m} \kappa_{m} \Biggr).
	\end{equation}
\end{lemma}

{\sc The flat case.}
In general, if we start from a CohFT with unit on $(V,\eta,\mathbb{1})$, acting by an $R$-matrix or by translation does not preserve flatness. However, there is a specific combination of the two which does achieve this.

\begin{proposition}\label{prop:TL:TR}
	Let $\Omega$ be a CohFT on $(V,\eta,\mathbb{1})$ with flat unit. Let $R$ be an $R$-matrix, and consider the $V$-valued power series
	\begin{equation}\label{eqn:TL:TR}
		T_{\textup{L}}(u) \coloneqq u \bigl( R^{-1}(u)\mathbb{1} - \mathbb{1} \bigr),
		\qquad
		T_{\textup{R}}(u) \coloneqq	u \bigl( \mathbb{1} - R^{-1}(u)\mathbb{1} \bigr).
	\end{equation}
	Then $T_{\textup{L}}R\Omega$ and $RT_{\textup{R}}\Omega$ coincide, and form a CohFT with flat unit.
\end{proposition}

\begin{definition}
	Let $\Omega$ be a CohFT on $(V,\eta,\mathbb{1})$ with flat unit. Let $R$ be an $R$-matrix. We define the \emph{unit-preserving action} by $R$ as
	\begin{equation}
		R.\Omega \coloneqq T_{\textup{L}}R\Omega = RT_{\textup{R}}\Omega.
	\end{equation}
\end{definition}

\subsection{Identification of the two theories}
This section mostly follows \cite{LPSZ17}. Consider a local spectral curve $(\Sigma,x,y,B)$ with $r$ simple ramification points. Under some conditions (see \cite{Eyn11,DOSS14,DNOPS18}), we can represent the topological recursion correlators on a basis of auxiliary differentials, with coefficients given by intersection numbers on the moduli space of curves of a CohFT and $\psi$-classes.

\smallskip

To be more precise, fix some constants $C_1,\dots,C_r \in \C^{\times}$ and an additional constant $C \in \C^{\times}$. Choose local coordinates $\zeta_i$ on $U_i$ such that $\zeta_i(a_i) = 0$ and $x = (C_i\zeta_i)^2 + x(a_i)$. Consider the auxiliary functions $\xi_i \colon \Sigma \to \C$ and the associated differentials $d\xi_{k,i}$ defined as
\begin{equation}
	\xi_i(z)
	\coloneqq
	\int^z \left.\frac{B(\zeta_{i},\cdot)}{d\zeta_{i}}\right|_{\zeta_{i} = 0},
	\qquad
	d\xi_{k,i}(z)
	\coloneqq
	d\biggl( \left( - \frac{1}{\zeta_i} \frac{d}{d\zeta_i} \right)^{k} \xi_{i}(z) \biggr).
\end{equation}
We also set $\Delta_i \coloneqq \frac{dy}{d\zeta_i}(0)$ and $t_i \coloneqq -2C_i^2 C \Delta_i$. We define a unital, semisimple TFT on $V \coloneqq \C\braket{e_1,\dots,e_r}$ by setting
\begin{equation}
	\eta(e_i,e_j) \coloneqq \delta_{i,j},
	\qquad
	\mathbb{1}
	\coloneqq
	\sum_{i=1}^r t_i e_i,
	\qquad
	\varpi_{g,n}(e_{i_1} \otimes \cdots \otimes e_{i_n})
	\coloneqq
	\delta_{i_1,\ldots,i_n} t_i^{-2g+2-n}.
\end{equation}
Define the $R$-matrix $R \in \End(V)\bbraket{u}$ and the translation $\hat{T} \in V\bbraket{u}$ (given by \cref{lem:translation:hat}) by setting
\begin{align}
	\label{eqn:R:matrix:TR}
	R^{-1}(u)^j_i
	& \coloneqq
	- \sqrt{\frac{u}{2\pi}} \int_{\R}
 		d\xi_i(\zeta_j)
 		\,
 		e^{-\frac{1}{2u} \zeta_{j}^2}, \\
	\label{eqn:translation:TR}
 	\exp\bigl(-\hat{T}^i(u)\bigr)
 	& \coloneqq \frac{1}{\Delta_i \sqrt{2\pi u}} \int_{\R} dy(\zeta_i) \, e^{-\frac{1}{2u} \zeta_{i}^2}.
\end{align}
The equations are intended as equalities between formal power series in $u$, where on the {\sc rhs} we take the asymptotic expansion as $u \to 0$. Through the Givental action, we can then define a cohomological field theory
\begin{equation}
	\Omega_{g,n} \coloneqq RT\varpi_{g,n} \in H^{\bullet}(\overline{\mathcal{M}}_{g,n}) \otimes (V^{\ast})^{\otimes n}
\end{equation}
from the data $(\varpi,R,T)$, through a sum over stable graphs $\Gamma \in \mathcal{G}_{g,n}$ as explained in the previous section. The link with the topological recursion correlators is given by the following theorem.

\begin{theorem}[\cite{Eyn11,DOSS14}]\label{thm:Eyn:DOSS}
	Suppose we have a compact spectral curve $ (\Sigma, x,y,B)$. Then its topological recursion correlators are given by
	\begin{equation}
		\omega_{g,n}(z_1,\dots,z_n)
		=
		C^{2g-2+n}
		\sum_{i_1,\dots,i_n = 1}^r \int_{\overline{\mathcal{M}}_{g,n}}
			\Omega_{g,n}(e_{i_1} \otimes\cdots\otimes e_{i_n}) \prod_{j=1}^n \sum_{k_j \ge 0}
				\psi_j^{k_j} d\xi_{k_j,i_j}(z_j).
	\end{equation}
\end{theorem}

\begin{remark}
	The compatibility between the translation and the $R$-matrix given by \cref{prop:TL:TR} is equivalent to the following condition (cf. \cite[equation~(17)]{LPSZ17}):
	\begin{equation}\label{eqn:y:B}
		t_j \exp\bigl(- \hat{T}^j(u)\bigr) = \sum_{i=1}^r R^{-1}(u)^j_i \, t_i \,.
	\end{equation}
	If such equation (which can be seen as a compatibility condition between $y$ and $B$) holds, the resulting CohFT coincides with $R.\varpi$, which is flat.
\end{remark}

\section{Preliminaries on spin algebra}
\label{sec:spin:algebra}
\subsection{Spin representations}

It is a classical fact that Hurwitz numbers are related to representation theory of the symmetric group via the monodromy representation. The analogous result in the context of spin Hurwitz numbers can be expressed in terms of representations of the spin-symmetric group. Because of this, we recall here some facts about spin representations, the Sergeev algebra, and supersymmetric functions. It contains no new results, but as the theory of spin representations is not as well-known as that of symmetric group representations, we felt it useful to include it here. This exposition is based on \cite{Joz88,Joz00,Iva04,EOP08,Gun16,McD98}, but beware that many definitions may vary in normalisation between these texts. 

\smallskip

By a result of Schur, $ H^2(\mf{S}_d; \C^{\times}) = \Z/2\Z$ for $d > 3$, which means there is one non-trivial central extension of $\mf{S}_d$, called the \emph{spin-symmetric group}
\begin{equation}
	1 \to \Z/2\Z \to \tilde{\mf{S}}_d \to \mf{S}_d \to 1.
\end{equation}
Explicitly, the spin-symmetric group can be presented as follows.
\begin{equation}
	\tilde{\mf{S}}_d
	=
	\bigl\< t_1, \dotsc, t_{d-1}, \epsilon \bigm|  \epsilon^2 = 1, t_j^2 = \epsilon, (t_jt_{j+1})^3 = \epsilon, (t_jt_k)^2 = \epsilon \text{ for } |j-k|>1 \bigr\>\,.
\end{equation}
It has a $\Z/2\Z$ grading given by $\deg \epsilon = 0$, $\deg t_j = 1$. In order to fix notation, we also give a presentation of the symmetric group:
\begin{equation}
	\mf{S}_d
	=
	\bigl\< \sigma_1, \dotsc, \sigma_{d-1} \bigm| \sigma_j^2 = 1, (\sigma_j\sigma_{j+1})^3 = 1, (\sigma_j\sigma_k)^2 = 1 \text{ for } |j-k| >1 \bigr\>\,.
\end{equation}
The map $ \tilde{\mf{S}}_d \to \mf{S}_d$ is then given by $ \epsilon \mapsto 1, t_j \mapsto \sigma_j$.

\smallskip

The representations of $\tilde{\mf{S}}_d$ that do not factor through $\mf{S}_d$ are called \emph{spin representations}. First we will recall the basics of their theory.

\begin{lemma}
	Spin representations are exactly the representations of the \emph{twisted group algebra}
	\begin{equation}
	\mc{S}_d \coloneqq \C[\tilde{\mf{S}}_d]/(\epsilon +1)\,,
	\end{equation}
	where $\epsilon $ is the added central element. It inherits a $\Z/2\Z$ grading, and hence is a superalgebra.
\end{lemma}

Note that irreducible modules of the twisted group algebra correspond exactly to irreducible representations of the spin-symmetric group where $ \epsilon$ acts as $ -1$.

%
%

\smallskip

For many explicit computations, it is easier not to work with the twisted symmetric group algebra, but with the Sergeev algebra, which we now introduce. 
Before doing so, let us introduce the Clifford algebra. There is a central extension of $ (\Z/2\Z)^d$ called the \emph{Clifford group}, given by
\begin{equation}
	\mathsf{Cl}_d = \bigl\< \xi_1, \dotsc, \xi_d, \epsilon \bigm| \xi_i^2 = \epsilon^2 = 1,  \xi_i \xi_j  = \epsilon \xi_j \xi_i \textup{ for } i \neq j \bigr\>.
\end{equation}
The Clifford group has a $\Z/2\Z$ grading given by $\deg \epsilon = 0$, $\deg \xi_j = 1$.

\begin{definition}\label{CliffAlg}
	The \emph{Clifford algebra} is the twisted group algebra of the Clifford group:
	\begin{equation}
		\Cl_d \coloneqq \C [\mathsf{Cl}_d]/(\epsilon + 1) \, .
	\end{equation}
	It inherits a $\Z/2\Z$ grading from its group, and hence is a superalgebra.
\end{definition}

\begin{proposition}[{\cite[proposition~3.1]{Joz88}}]
	All Clifford superalgebras are simple.
\end{proposition}

Let us introduce the Sergeev algebra. Let $D$ be a set of $2d$ elements and $ \sigma $ a fixed-point free involution on $ D$. We define the \emph{hyperoctahedral group} $\mf{H}_d$ to be the centraliser in $ \mf{S}_{2d} $ of $ \sigma $. Often, $ D $ is chosen to be $ \{ \pm 1, \dotsc, \pm d\} $ and $ \sigma (j) = -j$. Alternatively, $\mf{H}_d = \mf{S}_d \ltimes (\Z/2\Z)^d$, where $ \mf{S}_d $ acts on $ (\Z/2\Z)^d$ by permuting factors. The hyperoctahedral group has a central extension
\begin{equation}
	1 \to \Z/2\Z \to \tilde{\mf{H}}_d \to \mf{H}_d \to 1\,,
\end{equation}
which is the \emph{Sergeev group}. It can alternatively be given as $ \tilde{\mf{H}}_d = \mf{S}_d \ltimes \mathsf{Cl}_d $. The $\Z/2\Z$ grading extends to the Sergeev group by declaring pure permutations to be of degree $0$.

\begin{definition}\label{defn:Sergeev:alg}
	The \emph{Sergeev algebra} is the twisted algebra of the Sergeev group:
	\begin{equation}
		\mc{H}_d \coloneqq \C [\tilde{\mf{H}}_d] /(\epsilon + 1).
	\end{equation}
	It inherits a $\Z/2\Z$ grading, and hence is a superalgebra. Denote the even part of the centre of $ \mc{H}_d $ by $ \mc{Z}_d$.
\end{definition}

Note that again irreducible modules of the Sergeev algebra correspond exactly to irreducible representations of the Sergeev group where $ \epsilon$ acts as $ -1$.

\smallskip

The link between the representation theory of $\mc{S}_d$ and that of $\mc{H}_d$ is realised through the Clifford group. 

\begin{theorem}[{\cite[theorem~3.2]{Yam99}}]
	The map $\gamma_d \colon \mc{S}_d \otimes \Cl_d \to \mc{H}_d $ given by
	\begin{equation}
		\gamma_d \colon \begin{cases} 1 \otimes \xi_j &\!\!\!\! \mapsto \xi_j \\ t_j \otimes 1 &\!\!\!\! \mapsto \frac{1}{\sqrt{2}} ( \xi_j - \xi_{j+1}) \sigma_j  \end{cases}
	\end{equation}
	is an isomorphism. In $\mc{S}_d \otimes \Cl_d$ the tensor product is as superalgebras.
\end{theorem}

\begin{corollary}
	There is a bijection between simple supermodules for $\mc{S}_d$ and simple supermodules for $\mc{H}_d$, given by tensoring with the unique simple supermodule $ \Cl_d$ for $ \Cl_d$.
\end{corollary}

The classification of simple supermodules for $\mc{S}_d$ and $\mc{H}_d$ had both been done before, respectively by Schur~\cite{Sch11} and by Sergeev~\cite{Ser84}, recast in the language of superalgebra by J\'{o}zefiak~\cite{Joz89,Joz90}. To give it, we will need some preliminary definitions.

\begin{definition}
	For $d \in \N$, we write $ \mc{P}(d)$, $ \OP(d)$, and $\SP(d)$ for, respectively, the set of partitions, the set of odd partitions, and the set of strict partitions of $d$. A partition is \emph{odd} if all its parts are odd and \emph{strict} if all its parts are distinct. Also write $ \mc{P} = \bigcup_d \mc{P}(d) $, $ \OP = \bigcup_d \OP(d) $, and $ \SP = \bigcup_d \SP(d) $. We divide $ \SP(d)$ (and $\SP$) into $ \SP^+(d)$ ($\SP^+$) and $\SP^-(d) $ ($\SP^-$) denoting the strict partitions of even, resp. odd length.

	\smallskip

	We write $ m^\mu_k $ for the number of parts of a partition $ \mu$ equal to $k$ and $ z_\mu = \prod_k m^\mu_k! k^{m^\mu_k}$.
\end{definition}

\begin{lemma}[Euler]
	For any $d$, the set of odd partitions and the set of strict partitions of $d$ are of equal size:
	\begin{equation}
		|\OP(d)| = |\SP(d)|.
	\end{equation}
\end{lemma}

\begin{proposition}[{\cite{Ser84}}]
	Let $ \mu \in \OP(d) $. Then the conjugacy class $O_\mu$ in $ \tilde{\mf{H}}_d$ of a pure permutations of cycle type $ \mu$ contains $ \frac{d!}{z_\mu} 2^{d-\ell (\mu)} $ elements: for each element $ \alpha $ of the set
	\begin{equation}
		\Biggl\{ \tau \!\! \prod_{j \in J \subset \bbraket{d}} \!\! \xi_j  \Biggm|  \tau \in C_\mu \subset \mf{S}_d, |J \cap \mc{O}|  \text{ even for all orbits $\mc{O} $ of $ \tau$} \Biggr\}
	\end{equation}
	it contains either $ \alpha $ or $ \epsilon \alpha $ (but not both). In particular, it contains all pure permutations of cycle  type $\mu$.\par
	The set
	\begin{equation}
		\Biggl\{ c_\mu = \sum_{\alpha \in O_\mu} \alpha \Biggm| \mu \in \OP(d) \Biggr\}
	\end{equation}
	forms a linear basis of $ \mc{Z}_d$.
\end{proposition}

We will use this proposition to identify $ \mc{Z}_d$ with $ \C \{ \mc{OP}(d) \}$, the vector space with basis $\{ \mu \in \mc{OP}(d) \} $.

\begin{proposition}[{\cite{Ser84}}]
	There is a bijection between $\SP(d)$ and irreducible supermodules $ V^\lambda $ of $ \mc{H}_d$. 
\end{proposition}

\begin{definition}
	For an $ \mu \in \mc{OP}(d) $ and $ \lambda \in \mc{SP}(d)$, we write $ f_\mu^\lambda $ for the \emph{central character} of $ V^\lambda$ at $c_\mu$, i.e. $ f_\mu^\lambda $ is the scalar with which $ c_\mu$ acts in $ V^\lambda$. We will write $ \zeta_\mu^\lambda$ for the (usual) character.
\end{definition}

As usual,
\begin{equation}
	\zeta_\mu^\lambda = \frac{\dim V^\lambda}{|O_\mu|} f_\mu^\lambda = \frac{2^{\ell (\mu)} z_\mu \dim V^\lambda}{2^d d!} f^\lambda_\mu.
\end{equation}


The Grothendieck group $R(\mc{H}_d)$ of the category of $\mc{H}_d$-supermodules has a scalar product, which is given in terms of characters by
\begin{equation}
	\< \phi, \psi \> = \sum_{\mu \in \OP(d)} 2^{-\ell (\mu)} z_\mu^{-1} \phi_\mu \psi_\mu.
\end{equation}

\begin{lemma}[Orthogonality of characters]
	Let $\sigma$ and $\rho$ be two partitions of size $d$. Then
	\begin{equation}\label{eq:OrtCharB}
		\< \zeta^{\rho},\zeta^{\sigma}\> = 2^{\delta (\rho )} \delta_{\rho,\sigma}\,,
		\qquad \qquad
		\sum_{\lambda \in \SP(d)} 2^{-\ell(\sigma)-\delta (\lambda)} \frac{\zeta_{\sigma}^{\lambda}\zeta_{\rho}^{\lambda}}{z_{\sigma}} = \delta_{\rho,\sigma},
	\end{equation}
	where $\sigma$ and $\rho$ are strict partitions in the first equation and odd partitions in the second, and
	\begin{equation}
		\delta (\lambda) = \begin{cases} 1 & \lambda \in \SP^- \\ 0 & \lambda \in \SP^+ \end{cases}.
	\end{equation}
\end{lemma}

\subsection{Supersymmetric functions}

\begin{definition}
	The \emph{algebra of supersymmetric functions} is defined as $\Gamma \coloneqq \C [p_1, p_3, p_5, \dotsc ]$, where the $p_k$ are the usual power sums. It has a natural pairing
	\begin{equation}
		\< p_\mu, p_\nu \> = 2^{-\ell (\mu )}z_\mu \delta_{\mu,\nu},
	\end{equation}
  where we set $p_{\mu} = p_{\mu_1} p_{\mu_2} \cdots$ for any $\mu \in \OP$.
\end{definition}

In parallel to the usual setting, there is a change of basis of $\Gamma$ from power sums to central characters. In order to state it, we extend the definition of the central characters by $f_\mu^\lambda = 0$ if $ |\mu| > |\lambda|$ and
\begin{equation}
	f_\mu^\lambda = \binom{m^\mu_1 + k}{k} f^\lambda_{\mu \cup (1^k)} \quad \text{if $k = |\lambda | - |\mu|$.}
\end{equation}

\begin{proposition}[{\cite[proposition~6.4]{Iva04}}]\label{prop:central:chars:are:susy}
	The functions $ f_\mu (\lambda ) = f_\mu^\lambda $ are supersymmetric functions such that
	\begin{equation}
		f_\mu = \frac{1}{z_\mu} p_\mu + \text{lower degree terms.}
	\end{equation}
	Hence, the collection $\set{ f_\mu | \mu \in \OP }$ forms a basis of $ \Gamma$ and the map
	\begin{equation}
		\gamma \colon \bigoplus_{d =0}^\infty \mc{Z}_d \to \Gamma \colon c_\mu \to f_\mu
	\end{equation}
	is a linear isomorphism.
\end{proposition}

\begin{proof}
	Let us give the translation between the notation of \cite[definition 6.3]{Iva04} and ours: his $\theta^\lambda_\mu$ are our $ \zeta^\lambda_\mu$, and $ n^{\downarrow k} = \frac{n!}{(n-k)!}$. From this, $ p^\#_\mu (\lambda) = z_\mu f_\mu^\lambda$. Then, the statement is \cite[proposition~6.4(b-c)]{Iva04}.
\end{proof}

\begin{remark}
	In the ordinary case, the analogous theorem requires shifted symmetric functions. Here, instead, the supersymmetric functions only take values at strict partitions.
\end{remark}

\begin{definition}[{\cite{MMN20}}]\label{defn:spin:compl:cycles}
	For $r$ a positive even integer, define the \emph{spin completed $(r+1)$-cycles} to be
	\begin{equation}
		\bar{c}_{r+1} \coloneqq \frac{1}{r+1} \gamma^{-1} (p_{r+1}).
	\end{equation}
\end{definition}

\begin{remark}
	In fact, \cite{MMN20} do not quite define the spin completed $(r+1)$-cycles individually, but rather the space spanned by them. We choose the normalisation here to be consistent with the normalisation for the non-spin case, as defined in \cite{OP06}, as well as having the leading term of $\bar{c}_{r+1}$ be $c_{(r+1)}$. This is also used implicitly in \cite[equation~(2.9)]{Lee19} for $r=2$ and \cite[equation~(5.1)]{Lee19} for general $r$.
\end{remark}

There is another basis for $ \Gamma$, analogous to the Schur functions for the usual space of symmetric functions, which is given by \emph{Schur $Q$-functions}. First define $Q_{n,m}$ by
\begin{equation}
	\sum_{n,m=0}^\infty Q_{n,m} z^n w^m = \Biggl( \exp \biggl( 2\sum_{k=0}^\infty \frac{ p_{2k+1}}{2k+1} \bigl(z^{2k+1} + w^{2k+1}\bigr) \biggr) -1 \Biggr) \frac{z-w}{z+w}.
\end{equation}
Now, let $ \lambda \in \SP$. If $ \ell (\lambda ) $ is odd, we append it with a $0$. Then the matrix $\mc{Q}_{\lambda} = (Q_{\lambda_j, \lambda_k})_{j,k} $ is antisymmetric of even order, and we define 
\begin{equation}
	Q_\lambda = \Pf(\mc{Q}_{\lambda})\,,
\end{equation}
where $\Pf$ is the \emph{Pfaffian}, the square root of the determinant (which in this case is a perfect square).

\begin{proposition}
	The set $\set{ Q_\lambda | \lambda \in \SP}$ is an orthogonal basis for $ \Gamma$,
	\begin{equation}
		\< Q_\lambda, Q_\rho \> = 2^{\ell (\lambda )} \delta_{\lambda,\rho}.
	\end{equation}
	The base change between $\set{ Q_\lambda | \lambda \in \SP}$ and $\set{ p_\mu | \mu \in \OP}$ is given by
	\begin{equation}
		Q_\lambda =
		2^{\lfloor \ell (\lambda )/2 \rfloor} \!\! \sum_{\mu \in \OP(d)} z_\mu^{-1} \zeta^\lambda_\mu p_\mu,
		\qquad\quad
		p_\mu =
		\sum_{\lambda \in \SP(d)} 2^{-\ell (\mu ) - \lceil \ell (\lambda)/2\rceil} \zeta_\mu^\lambda Q_\lambda.
	\end{equation}
\end{proposition}

The link with the representation theory of the Sergeev algebra is given by the following result.

\begin{theorem}[{\cite{Ser84,Joz90}}]
	The direct sum of Grothendieck rings, $ T = \bigoplus_{d=0}^{\infty} R(\mc{H}_d)$, is a graded ring and
	\begin{equation}
		\Ch_d \colon R(\mc{H}_d) \to \Gamma \colon \phi \mapsto \!\! \sum_{\mu \in \mc{OP}(d)} \!\!\! z_\mu^{-1} \phi_\mu p_\mu
	\end{equation}
	fit together in a map $\Ch \colon T \otimes_\Z \C \to \Gamma$ which is an isomorphism of graded rings with scalar product. Furthermore,
	\begin{equation}
		\Ch ( \zeta^\lambda ) = 2^{-\lfloor \ell (\lambda )/2 \rfloor} Q_\lambda.
	\end{equation}
\end{theorem}

\subsection{Operator formalism}

In the non-spin case, the semi-infinite wedge space is a realisation of a highest weight module of the relevant infinite Clifford algebra. In the spin case, there is still an infinite Clifford algebra, but its highest weight module does not have such a nice description (although it can be described in terms of \emph{half-line Maya diagrams}, cf. \cite{FWZ09}). As the non-spin and spin theory are related to A- and B-type Dynkin diagrams, respectively, we also refer to them by these letters. This material is taken from \cite{You89,Iva04,Lee19,KVdL19}.

\begin{definition}
	Let $V$ be a vector space with a quadratic form $ q$. Then the \emph{Clifford algebra} $\Cl (V,q)$ is the unital free algebra on $V$ modulo the relations
	\begin{equation}
		v^2 = q(v) 1 \quad \text{for all } v \in V.
	\end{equation}
	There is a canonical embedding $ V \into \Cl (V,q)$, and we identify $V$ with its image.
\end{definition}

\begin{remark}
	As  $\mathop{\textup{char}} \C \neq 2$, a quadratic form is equivalent to a bilinear form, and the relation reads $ \{ v, w\} = 2\< v, w\> 1$. We will write $ \Cl (V, \< \plh,\plh \>) $ in this case.
\end{remark}

\begin{example}
	Let $V^d$ be a $d$-dimensional $\C$-vector space with standard basis $ ( \xi_1, \dotsc, \xi_d)$ and Euclidean inner product $ \< \xi_k,\xi_l \> = \delta_{k,l}$. Then $ \Cl (V^d, \< \plh,\plh\> ) = \Cl_d$ from \cref{CliffAlg}.
\end{example}

Just like in this example, every Clifford algebra has a $\Z/2\Z$ grading, obtained by placing $ V $ in degree $ 1$.

\begin{definition}
	Define an inner product space $ (W^B \coloneqq \C \{ \phi_n\}_{n \in \Z}, \< \phi_k,\phi_l\> \coloneqq \frac{(-1)^k}{2} \delta_{k+l})$ and call it the space of \emph{neutral} (or \emph{uncharged}) \emph{Fermions}. Denote its Clifford algebra by $ \Cl^B \coloneqq \Cl (W^B, \< \plh,\plh \>)$. It has anticommutation relations 
	\begin{equation}
		\{ \phi_k, \phi_l \} = (-1)^k \delta_{k+l}\,.
	\end{equation}
	The subspace $ L_0^B \coloneqq \C \{ \phi_n \}_{n < 0}\subset W^B$ is maximally isotropic, and $ \mf{F}^B \coloneqq \Cl^B /\Cl^B L_0^B $ is the (unique graded) highest weight module for $ \Cl^B$. We call it \emph{fermionic Fock space (of type B)} and write $ | 0\> $ or $ v_\emptyset $ for the class of $ 1$.

	\smallskip

	The fermionic Fock space has an inherited inner product, and we use the notation $ \< \mc{O} \> \coloneqq \< v_\emptyset, \mc{O} v_\emptyset \>$ for $\mc{O} \in \Cl^B$. We call this the \emph{vacuum expectation value} of $\mc{O}$.
\end{definition}

\begin{remark}
	The fermionic Fock space of type B is unique as a \emph{graded} highest weight module. Forgetting the grading, $ \phi_0$ may also act as $ \pm \frac{1}{\sqrt{2}}$.
\end{remark}

\begin{remark}
	The above construction can be carried out for $A$-, $B$-, $C$-, and $D$-type theories, cf. \cite{JM83}.
	\begin{itemize}
		\item
		In the A case, there are two infinite sets of fermions: $\psi_n$ and $\psi_n^*$, indexed by $n \in \Z+\frac{1}{2}$. This is because $A \sim \mf{gl} $ does not have a natural pairing, so one needs to add the dual space to get it.

		\item
		The C case is constructed from $(W^C \coloneqq \C \{ \chi_n \}_{n \in \Z'} , (\chi_k,\chi_l) \coloneqq (-1)^{k-1/2} \frac{\delta_{k+l}}{2})$ as a Weyl algebra (in contrast to a Clifford algebra). Hence, the $ \chi_n$ are bosons, not fermions.

		\item
		The D case is constructed from the space $(W^D \coloneqq \C\{ \phi_n \}_{n \in \Z'} , \< \phi_k, \phi_l\> \coloneqq \frac{\delta_{k+l}}{2})$.
	 \end{itemize}
\end{remark}

\begin{lemma}\leavevmode\label{lem:VEV:Fock:basis}
	\begin{enumerate}
		\item
		The vacuum expectation values of quadratic expressions in $\phi$'s are
		\begin{equation}
			\< \phi_k \phi_l \> = (-1)^l \delta_{k+l} \Big( \delta_{l > 0} + \frac{\delta_l}{2} \Big)\,;
		\end{equation}

		\item
		For $ \lambda \in \SP$, define
		\begin{equation}
			v_\lambda \coloneqq
			\begin{cases}
				\phi_{\lambda_1} \dotsb \phi_{\lambda_{\ell (\lambda )}} v_{\emptyset}
				& \lambda \in \SP^+ \\
				\sqrt{2} \phi_{\lambda_1} \dotsb \phi_{\lambda_{\ell (\lambda )}} \phi_0 v_{\emptyset}
				& \lambda \in \SP^-
			\end{cases}.
		\end{equation}
		Then an orthonormal basis of $\mf{F}^B_0$, the degree zero subspace, is given by $\set{ v_\lambda | \lambda \in \SP }$.
	\end{enumerate}
\end{lemma}

In order to develop the usual theory, we need to get bosons, which are quadratic in the fermions.

\begin{definition}\label{defn:B:norm:ord}
	We define the \emph{normal ordering} of a quadratic expressions in the fermions $\phi$ to be
	\begin{equation}
		\normord{\phi_j \phi_k} = \phi_j \phi_k - \< \phi_j \phi_k \>\,.
	\end{equation}
\end{definition}

\begin{lemma}\label{lem:comm:phi:phi} 
	The commutation relation between quadratic elements in the $\phi_i$ reads
	\begin{equation}\label{eqn:comm:phi:phi}
		[\phi_i \phi_j, \phi_k \phi_l] = (-1)^j \delta_{j+k} \phi_i \phi_l - (-1)^i \delta_{i+k} \phi_j \phi_l + (-1)^j \delta_{j+l} \phi_k \phi_i - (-1)^i \delta_{i+l} \phi_{k} \phi_j\,.
	\end{equation}
	For the normally ordered variants,
	\begin{equation}\label{eqn:comm:Norm:ord:phi:phi}
	\begin{split}
		& [ \normord{\phi_i \phi_j},\normord{\phi_k \phi_l}] = \\
		& \quad
		(-1)^j \delta_{j+k} \normord{\phi_i \phi_l} - (-1)^i \delta_{i+k} \normord{\phi_j \phi_l} + (-1)^j \delta_{j+l} \normord{\phi_k \phi_i} - (-1)^i \delta_{i+l} \normord{\phi_k \phi_j}
		\\
		& \qquad
		+ \delta_{i+j+k+l} \bigg(
			(-1)^{j+l} \delta_{i+l} \Big(
				\delta_{l>0} - \delta_{j>0} + \frac{\delta_l-\delta_j}{2}
			\Big)
			+ (-1)^{i+l} \delta_{j+l} \Big(
				\delta_{i>0} - \delta_{l>0} + \frac{\delta_i -\delta_l}{2} \Big)
		\bigg) \, .
	\end{split}
	\end{equation}
\end{lemma}
\begin{proof}
	The first formula is straightforward calculation, while the second formula is exactly the same, but requires correcting the central part.
\end{proof}

\begin{definition}
	Let $\mf{a}_\infty = \{ (a_{mn})_{m,n \in \Z} \,|\, a_{mn} = 0 \text{ for } |m-n| \gg 0\}$ be the bi-infinite general linear algebra. It has a basis $\set{ E_{j,k} = (\delta_{m,j}\delta_{n,k})_{m,n} }_{j,k \in \Z}$. Define an involution $\iota \colon E_{j,k} \mapsto (-1)^{j+k} E_{-k,-j}$. Then define the infinite Lie algebra
	\begin{equation}
		\mf{b}_\infty \coloneqq \bigl\{ g \in \mf{a}_\infty \bigm| \iota (g) = -g\bigr\}.
	\end{equation}
	It has a basis $\set{ B_{j,k} = E_{-j,k} - (-1)^{j+k} E_{-k,j} }_{j>k}$ and two irreducible representations on $\mf{F}^B_0$: a linear, $ \pi$, and a projective, $ \hat{\pi}$, given on this basis by
	\begin{equation}
		\pi \colon B_{j,k} \mapsto (-1)^j \phi_j\phi_k \,, \qquad \hat{\pi} \colon B_{j,k} \mapsto (-1)^j \normord{\phi_j\phi_k} \,.
	\end{equation}
	We write $\hat{\mf{b}}_\infty$ for the central extension of which $\hat{\pi}$ is a representation. From now on, we will write $ F_{j,k}\coloneqq (-1)^j \phi_j \phi_k$ and denote $ \hat{F}_{j,k} \coloneqq (-1)^j \normord{\phi_j \phi_k}$.
\end{definition}

\begin{definition}
	For any odd integer $m$, define the \emph{bosonic operators} $\alpha^B$ as:
	\begin{equation}\label{eqn:B:boson}
		\alpha^B_m
		\coloneqq
		-\frac{1}{2} \sum_{k \in \Z} \hat{F}_{k,-k-m}
		=
		-\!\! \sum_{k> -m/2} \! F_{k,-k-m} \, .
	\end{equation}
	The operators $\alpha^B_m$ could also be defined for even arguments, but those would vanish. We also define the generating series
	\begin{equation}
		H^B \coloneqq \sum_{m \in 2 \N +1} \alpha^B_m t_m\,, \qquad t_m = \frac{p_m}{m}.
	\end{equation}
\end{definition}


\begin{lemma}\label{lem:F:hat:symmetry:Heisenberg} ~
	\begin{enumerate}
		\item The elements $\hat{F}_{j,k}$ satisfy $\hat{F}_{j,k} = (-1)^{j+k+1} \hat{F}_{k,j}$.

		\item The bosonic operators satisfy
		\begin{equation}
			\bigl[ \alpha^B_m,\alpha^B_n \bigr] = \frac{m}{2}\delta_{m+n} \,,
		\end{equation}
		and the subalgebra they span is called the \emph{Heisenberg subalgebra} $\mf{H}^B$.
	\end{enumerate}
\end{lemma}


The Heisenberg algebra $ \mf{H}^B$ has a unique (up to unique isomorphism) irreducible representation, given by $ \Gamma$, with action $ \rho( \alpha^B_m ) = m\frac{\del}{\del t_m} $ and $ \rho (\alpha^B_{-m} ) = \frac{1}{2} t_m$ (for $ m = 2n+1 >0$).

\begin{theorem}[Boson-fermion correspondence, \cite{You89}]\label{thm:BF:corresp}
	As a representation of $\mf{H}^B$, $\mf{F}^B_0$ is irreducible, and therefore isomorphic to $\Gamma$. Explicitly,
	\begin{equation}
		\Phi^B \colon \mf{F}^B_0 \to \Gamma \colon v \mapsto v_\emptyset^* e^{H^B}v
	\end{equation}
	is an isomorphism. On a basis, $\Phi^B (v_\lambda ) = 2^{-\ell (\lambda )/2} Q_\lambda ( \tfrac{1}{2} p )$.
\end{theorem}

\begin{corollary}\label{cor:boson:action}
	Inverting $\Phi^B$ and using the base change between $Q_\lambda$ and $p_\mu$, we get
	\begin{equation}
		\alpha^B_{-\mu} v_{\emptyset}
		=
		\sum_{\lambda \in \SP (|\mu| )} \frac{\zeta_{\mu}^{\lambda}}{2^{\delta(\lambda)/2 + \ell(\mu)}} v_{\lambda}.
	\end{equation}
\end{corollary}

Another natural sequence of element of $\hat{\mf{b}}_{\infty}$ is given by the fermionic cut-and-join operators, which will play an important role in the theory of spin Hurwitz numbers.

\begin{definition}
	For any positive even integer $r$, define the \emph{fermionic cut-and-join operators} $\mc{F}^B$ as: 
	\begin{equation}\label{eqn:B:power:sum}
		\mc{F}^B_{r+1}
		\coloneqq
		\frac{1}{2} \sum_{k \in \Z} k^{r+1} \hat{F}_{k,-k}
		=
		\sum_{k >0} k^{r+1} F_{k,-k}.
	\end{equation}
	Again, the operators $\mc{F}^B_{r+1}$ could also be defined for odd $r$, but those would vanish.
\end{definition}

Using that $\hat{F}_{k,-k} v_\lambda = v_\lambda$ if $k \in \lambda$ and zero otherwise, we get the following result.

\begin{lemma}\label{lem:power:sum:eigenvalue}
	The $v_\lambda $ are eigenvectors of $ \mc{F}^B_{r+1}$, with eigenvalue $p_{r+1}(\lambda)$.
\end{lemma}

Note that \cite{You89,Lee19} give strong links between the $A$ and $B$ type theories.
A comparison between these theories is given in the following table.
\begin{center}
\begin{tabular}{l|c|c}
	\toprule
	&
	\textbf{$\bm{A}$ type} 
	&
	\textbf{$\bm{B}$ type}
	\\
	\midrule
	Index set &
	$\Z' = \Z + \frac{1}{2}$
	&
	$\Z$
	\\[1ex]
	Lie algebra
	&
	$\mf{a}_\infty$
	&
	$\mf{b}_\infty$
	\\[1ex]
	Vector space
	&
	$W^A = \bigoplus_{k \in \Z'} (\C \psi_k \oplus \C \psi_k^*)$
	&
	$W^B = \bigoplus_{k \in \Z} \C \phi_k
	$
	\\[1ex]
	Inner product
	&
  \pbox{20cm}{
  $\< \psi_k, \psi_l^* \> = \delta_{k,l}/2$ \\
  $\< \psi_k^{*}, \psi_l^{*} \> = \< \psi_k ,\psi_l \> = 0$
  }
	&
	$\< \phi_k, \phi_l \> = (-1)^k \delta_{k+l}/2$
	\\[1ex]
	Isotropic subspace
	&
	$L_0^A = \bigoplus_{k < 0} (\C \psi_k \oplus \C \psi_{-k}^*)$
	&
	$L_0^B = \bigoplus_{k <0} \C \phi_k$
	\\[1ex]
	Clifford algebra
	&
	$\Cl^A = \Cl (W^A, \<\plh, \plh\> )$
	&
	$\Cl^B = \Cl (W^B,\< \plh, \plh\>)$
	\\[1ex]
	CCR 
	&
	\pbox{20cm}{
	$\{ \psi_k ,\psi_l^{*} \} = \delta_{k,l}$ \\
	$\{ \psi_k^{*}, \psi_l^{*} \} = \{ \psi_k ,\psi_l \} = 0$
	}
	&  
	$
	\{ \phi_k, \phi_l \} = (-1)^{k} \delta_{k+l}$
	\\[2ex]
	Fermionic Fock space
	&
	$\mf{F}^A = \Cl^A/\Cl^A L_0 \cong \bigwedge^{\frac{\infty}{2}} V$
	&
	$\mf{F}^B =  \Cl^B / \Cl^B L_0^B$
	\\[1ex]
	Main subspace
	&
	$\mf{F}^A_0 \cong \bigwedge^{\frac{\infty}{2}}_0 V$
	&
	$\mf{F}^B_0$
	\\[1ex]
	Basis
	&
	$\{ v_\lambda = \psi_{\lambda_1} \dotsb \psi_{\lambda_{\ell (\lambda)}} \ket{0} \mid \lambda \in \mc{P} \}$
	&
	$\{ v_\lambda = 2^{\delta(\lambda)/2}\phi_{\lambda_1} \dotsb \phi_{\lambda_{\ell (\lambda )}} (\phi_0) \ket{0} \mid \lambda \in \SP\}$
	\\[1ex]
	Bosons
	&
	$\alpha_n = \sum\limits_{k \in \Z'} \normord{\psi_k \psi^{*}_{k + n}}$
	&
	$\alpha_n^B 
	=\frac{1}{2} \sum\limits_{k \in \Z} (-1)^{k+1} \normord{\phi_k \phi_{-k-n}}$
	\\[1ex]
	Trivial bosons 
	&
	no such relation
	&
	$\alpha_{2n}^B = 0$
	\\[1ex]
	Bosonic CCR
	&
	$[\alpha_m, \alpha_n] = m \delta_{m+n}$
	&
	$[\alpha_{m}^B, \alpha_{n}^B] = \frac{m}{2} \delta_{m+n} 
	$
	\\[1ex]
	CJ operators
	&
	$\mc{F}_{r+1} = \sum\limits_{k \in\Z'} k^{r+1} \normord{ \psi_k \psi^*_k}$
	&
	$\mc{F}^B_{r+1} = \frac{1}{2} \sum\limits_{k \in \Z} (-1)^k k^{r+1} \normord{\phi_k \phi_{-k}}$
	\\[1ex]
	OP operators \eqref{B-OPops}
	&
	$ \hat{\mc{E}}_n(z) = \sum\limits_{k \in \Z'} e^{z(k-\frac{n}{2})}\normord{\psi_{k-n} \psi^*_k}$
	&
	$\hat{\mc{E}}^{B}_n(z) = \sum\limits_{k \in \Z} \frac{(-1)^{k+1}}{2} e^{-(k+\frac{n}{2})z} \normord{\phi_{k} \phi_{-k-n}}$
	\\[1ex]
	Hamiltonian
	&
	$H(\bm{t}) = \sum_{n \in \N} \alpha_n t_n$
	&
	$H^B(\bm{t}) = \sum_{n \in 2 \N +1} \alpha_n^B t_n$
	\\[1ex]
	Bosonic Fock space
	&
	$\Lambda = \C [t_1, t_2, t_3, \dotsc ]$
	&
	$\Gamma = \C [t_1, t_3, t_5, \dotsc ]$
	\\[1ex]
	B-F correspondence
	&
  \pbox{20cm}{
  	$\Phi \colon \mf{F}_0^A \iso \Lambda \colon v \mapsto v_\emptyset^* e^{H(t)} v$ \\
   $v_\lambda \mapsto s_\lambda (p), \, p_n = m t_n$
  }
	&
  \pbox{20cm}{
  	$\Phi^B \colon  \mf{F}_0^B \iso \Gamma \colon v \mapsto v_\emptyset^* e^{H^B(t)} v$ \\
    $v_\lambda \mapsto 2^{-\ell(\lambda )} Q_\lambda (p/2), \, p_n = m t_n$
  }
	\\[1ex]
	Weyl group
	&
	$\mf{S}_\infty = \lim_{d \to \infty} \mf{S}_d$
	&
	$\mf{H}_\infty = \lim_{d \to \infty} \mf{H}_d$
	\\
	\bottomrule
\end{tabular}
\end{center}

\vspace{1cm}
\section{Introduction to spin Hurwitz numbers}
\label{sec:spin:HNs}

Spin Hurwitz numbers are weighted counts of covers of a curve with a spin structure or theta characteristic, with a sign given by the parity~\cite{EOP08}. This is captured in the following definitions. We switch between the languages of divisors and line bundles without mentioning it. For more background on spin Hurwitz numbers in relation to integrable hierarchies and supersymmetric functions, see also \cite{Lee19,MMN20,MMNO21}.

\begin{definition}\label{defn:spin:structure:parity}
	A \emph{spin structure} or \emph{theta characteristic} on a curve $C$ is a line bundle $ \vartheta \to C$ such that $ \vartheta^{\otimes 2} \cong \omega_C $. A \emph{spin curve} is a pair $ (C,\vartheta )$ of a curve $C$ with a spin structure $\vartheta$ on it. Define the \emph{parity} of a spin structure $\vartheta \to C$ as
	\begin{equation}
		p(\vartheta ) \equiv h^0 (C;\vartheta ) \pmod{2}.
	\end{equation}
\end{definition}

The parity is a deformation invariant of $(C,\vartheta )$, a fact proved for smooth curves by Riemann in the language of theta functions, or more abstractly by Mumford \cite{Mum71} in the algebraic and Atiyah \cite{Ati71} in the analytic setting. Mumford's proof was extended to nodal spin curves by Cornalba \cite[section 6]{Cor89}. Moreover, $p$ is a quadratic form on the space of theta characteristics $ \mc{S}(C)$, polarised by the cup product $\lambda$ on the $2$-torsion Jacobian $J_2(C) = \{ L \in \Pic (C) \mid L^2 \cong \mc{O}_C\} \cong H^1(C;\Z/2\Z)$:
\begin{equation}
	\phi (\vartheta ) + \phi (\vartheta \otimes L ) + \phi (\vartheta \otimes M) + \phi (\vartheta \otimes L \otimes M) = \lambda (L,M)\,.
\end{equation}
As such, for a genus $g$ curve, there are $ 2^{g-1} (2^g+1)$ even theta characteristics and $ 2^{g-1}(2^g-1)$ odd ones. This was proved by the same authors with the same methods. We also remark that the parity $ p( \vartheta )$ is the Arf invariant of the associated quadratic form $ q_\vartheta (L) = p(\vartheta ) + p(\vartheta \otimes L)$ on $ J_2(C)$, but we will not use this.

\smallskip

Spin structures can be pulled back along branched covers, as long as all ramifications are odd: in that case the ramification divisor is even.

\begin{definition}
	Let $ (B,\vartheta )$ be a spin curve and $f \colon C \to B $ a branched cover with only odd ramifications. Denote by $R_f$ its ramification divisor. Then the \emph{twisted pullback of $\vartheta$ along $f$} is
	\begin{equation}
		N_{f,\vartheta} \coloneqq f^*\vartheta \otimes \mc{O}(\tfrac{1}{2} R_f) .
	\end{equation}
	It is a spin structure on $ C$.
\end{definition}

\begin{definition}\label{defn:general:spin:HNs}
	Let $ (B,\vartheta )$ be a spin curve, $d$ a non-negative integer, $ x_1, \dotsc, x_k \in B$, and $ \mu^1, \dotsc, \mu^k \in \OP(d)$. The \emph{spin Hurwitz number} is defined as
	\begin{equation}
		H(B,\vartheta ; \mu^1, \dotsc, \mu^k ) \coloneqq \sum_{[f \colon C \to B]} \frac{(-1)^{p(N_{f,\vartheta})}}{|\Aut (f)|}\,,
	\end{equation}
	where the sum is over all isomorphism classes of connected ramified covers with ramification profile $\mu^j $ over $x_j$ and unramified anywhere else. As usual, when dealing with disconnected covers, we add a superscript $ \bullet$. Sometimes, we also write a superscript $\circ$ to denote connected covers.
\end{definition}

The following result is called Gunningham's formula~\cite{Gun16}, e.g. in~\cite{Lee19}. It is a generalisation of~\cite[theorem~2]{EOP08}, which is the case $g(B) = 1$, $p(\vartheta ) = 1$, and it gives the analogue of the monodromy representation for spin Hurwitz numbers.

\begin{theorem}[{\cite{Gun16,Lee18}}]\label{thm:Gunningham:formula}
	Disconnected spin Hurwitz numbers can be expressed in terms of characters of the Sergeev group as
	\begin{equation}
		H^\bullet (B,\vartheta ;\mu^1, \dotsc, \mu^k)
		=
		2^{( \sum_j (\ell (\mu^j) -d) - d \chi (B))/2} \sum_{\lambda \in \SP(d)} (-1)^{p(\vartheta ) \ell (\lambda )} \biggl( \frac{\dim V^\lambda}{2^{\delta (\lambda )/2} d!} \biggr)^{\chi (B)}\prod_{j=1}^k f_{\mu^j}^\lambda  .
	\end{equation}
\end{theorem}


As we will be interested in the case $(B,\vartheta) = (\P^1,\mc{O}(-1))$, i.e. $ \chi (B)=2$ and $ p(\vartheta ) = h^0 (\mc{O}(-1)) = 0$, and no more than two generic ramifications $\mu,\nu$, let us give the the appropriate definitions and a simplified formulae.

\begin{definition}\label{defn:spin:HNs}
	Let $r$ be a positive even integer. The \emph{$(r+1)$-completed cycles spin double Hurwitz numbers} for genus $g$, denoted $h^{r,\theta}_{g;\mu,\nu}$, are defined by
	\begin{equation}
		h^{r,\theta}_{g;\mu,\nu}
		\coloneqq
		\frac{|\Aut{(\mu)}| \, |\Aut{(\nu)}|}{b!} H(\P^1, \mc{O}(-1); \mu, \nu, (\bar{c}_{r+1})^b ) \,,
	\end{equation}
	where $ b = \frac{2g-2+ \ell (\mu) + \ell (\nu)}{r}$ is needed to obtain genus $g$ source curves, as follows from the Riemann--Hurwitz formula.

	\smallskip

	The \emph{$(r+1)$-completed cycles spin single Hurwitz numbers} for genus $g$, $h^{r,\theta}_{g;\mu} $, are defined by
	\begin{equation}
		h^{r,\theta}_{g;\mu}
		\coloneqq
		\frac{|\Aut{(\mu)}|}{b!}  H(\P^1, \mc{O}(-1); \mu, (\bar{c}_{r+1})^b ) \,,
	\end{equation}
	where $b = \frac{2g-2+ \ell (\mu) + |\mu|}{r}$ is needed to obtain genus $g$ source curves, again from the Riemann--Hurwitz formula. For this definition, recall that we identified $\mc{Z}_d$ with $\C\{ \mc{OP}_d \}$. We use this identification and define the Hurwitz numbers by multilinear extension on the latter space.
\end{definition}

\begin{remark}\label{rem:spin:vs:compl:cycl}
	Certain papers (e.g.~\cite{MSS13,SSZ15,BKLPS21,KLPS19,DKPS23}) on the non-spin version of these numbers use the term $ r$-spin Hurwitz numbers for what would be called $(r+1)$-completed cycles Hurwitz numbers here, to emphasise the relation to $r$-spin structures on the moduli spaces of curves. 
\end{remark}

\begin{remark}
	In the non-spin case, the case $r=1$ is the `usual' case, with the most generic, simple, ramifications. However, for spin Hurwitz numbers, $r=1$ is not allowed, as all ramifications must be odd. Since $r=2$ already requires the cycles to be completed, it is natural to develop the entire completed cycles theory immediately.
\end{remark}

\begin{corollary}
	Gunningham's formula applied to $(r+1)$-completed cycles spin double and single Hurwitz numbers yields
	\begin{align}
	\label{eqn:Gunningham:spin:double:HNs}
		h^{\bullet,r,\theta}_{g;\mu,\nu}
		& =
		\frac{2^{1-g}}{b! \prod_{i=1}^{\ell(\mu)} \mu_i \prod_{j=1}^{\ell(\nu)} \nu_{j}}
		\sum_{\lambda \in \SP(d)}
			\frac{\zeta_{\mu}^\lambda}{2^{\delta(\lambda )/2 + \ell (\mu)}}
			\frac{\zeta_{\nu}^\lambda}{2^{\delta (\lambda )/2 + \ell (\nu)}}
			\biggl( \frac{p_{r+1}(\lambda)}{r+1} \biggr)^b \,;
		\\
	\label{eqn:Gunningham:spin:single:HNs}
		h^{\bullet,r,\theta}_{g;\mu}
		& =
		\frac{2^{1-g}}{b!\prod_{i=1}^{\ell(\mu)} \mu_i}
		\sum_{\lambda \in \SP(d)}
			\frac{\zeta_{\mu}^\lambda}{2^{\delta(\lambda )/2 + \ell (\mu)}} 
			\frac{\dim V^\lambda}{2^{\delta (\lambda )/2 + d}d!}
			\biggl( \frac{p_{r+1}(\lambda)}{r+1} \biggr)^b.
	\end{align}
\end{corollary}

Note that, as always, $ \zeta^\lambda_{(1^d)} = \dim V^\lambda $. In order to study these numbers, we gather them in generating functions.

\begin{definition}
	The \emph{free energies} for spin completed cycles single Hurwitz numbers are
	\begin{equation}
	F^{r,\theta}_{g,n} (e^{x_1}, \dotsc, e^{x_n}) = 
		\sum_{ \substack{
		\mu_1, \dots, \mu_n > 0
		\\
		\mu_i \text{ odd }}} 
		h_{g;\mu}^{r,\theta}
		e^{\mu_1 x_1} \cdots e^{\mu_n x_n}\,.
	\end{equation}
\end{definition}

\section{Cut-and-join operators and Okounkov--Pandharipande operators}
\label{sec:CJ:OP:operators}

In the non-spin case, double Hurwitz numbers with completed cycles satisfy some important properties:
\begin{enumerate}
	\item They can be expressed as vacuum expectation values of a certain product of bosonic operators and Okounkov--Pandharipande operators.

	\item Their generating functions is a $2d$ Toda tau function, and it satisfies an evolution equation known as cut-and-join equation.

	\item They satisfy piecewise polynomiality and wall crossing formulae.
\end{enumerate}
The goal of this section is two-fold: first, we obtain the generating series for the spin cut-and-join operators, and second, we develop the B-type analogue of the Okounkov--Pandharipande operators, a simultaneous generalisation of the operators $\mc{F}^B$ and the bosonic operators $\alpha^B$, \cref{eqn:B:power:sum,eqn:B:boson}. These are the B-type analogues of, respectively, \cite[theorem~5.2]{SSZ12} or \cite[theorem~5.3]{Ros08}, and the algebra of operators $\mc{E}_m(z)$ of \cite[equation~(2.15)]{OP06}. This section is strongly related to \cite[section 3.5]{MMNO21}, which introduces the basic objects we study here.

\smallskip

In the next section we will relate the above constructions to properties (1--3) for spin double Hurwitz numbers with completed cycles.

\subsection{Cut-and-join operators}

The completed cut-and-join operators are defined (in the non-spin case) to be the analogues of multiplication by the completed cycles in the symmetric algebra. We will develop this in the spin case. For this, it is useful to understand the relations between different spaces introduced in \cref{sec:spin:algebra}:

\begin{equation}\label{eqn:B:diagram}
\begin{tikzcd}[ampersand replacement=\&]
	\& {\bigoplus_{d = 0}^\infty \mathcal{Z}_d} \\
	{\mathcal{OP}^B} \& \Gamma \& {\mathfrak{F}^B_0} \\
	{\mathfrak{H}^B} \& {\hat{\mathfrak{b}}_\infty} \& \Gamma
	\arrow["\Phi^B", leftrightarrow, from=2-3, to=3-3]
	\arrow[hook, from=3-1, to=3-2]
	\arrow["{\mathcal{F}}"', hook', from=2-2, to=3-2]
	\arrow[squiggly, from=2-2, to=2-3]
	\arrow["\hat{\pi}"{description}, squiggly, from=3-2, to=2-3]
	\arrow["{\hat{\sigma}^B}"{pos=0.9}, squiggly, from=3-2, to=3-3]
	\arrow["\gamma", leftrightarrow, from=1-2, to=2-2]
	\arrow[dashed, hook, from=3-1, to=2-1]
	\arrow[dashed, hook', from=2-2, to=2-1]
	\arrow[dashed, hook, from=2-1, to=3-2]
	\arrow["\rho"{description}, bend right=20pt, squiggly, from=3-1, to=3-3]
\end{tikzcd}
\end{equation}

Here, the wavy arrows indicate an action of the tail on the head. The dashed arrows indicate the B-Okounkov--Pandharipande operators constructed later in this section, the map $\mc{F} \colon \Gamma \to \hat{\mf{b}}_\infty $ sends $p_{r+1}$ to $\mc{F}_{r+1}$. Note that $\Gamma $ occurs twice in this diagram, with one instance acting on the other. This is \emph{not} the natural action of an algebra on itself as a module, but rather the one given by $(p_{r+1}, Q_\lambda (\frac{1}{2} p) ) \mapsto p_{r+1}(\lambda) Q_\lambda (\frac{1}{2}p)$, as follows from \cref{lem:power:sum:eigenvalue,thm:BF:corresp}. In fact, in the non-spin version of this diagram, we have the space of shifted symmetric functions $\Lambda^{\ast}$ acting on the usual space of symmetric functions $\Lambda$.

\smallskip

An important ingredient in this diagram is the shape of $\sigma^B$, given by the vertex operator of \cite{DKM81}, which translates the action of the quadratic elements $\varphi_{j}\varphi_{k}$ (and their normal product) to the space of supersymmetric functions $\Gamma$.

\begin{definition}
	Denote by
	\begin{equation}
		\tilde{\xi}(\mathbf{t},x)
		\coloneqq
		\sum_{k =0}^\infty t_{2k+1} x^{2k+1}\,,
		\qquad
		\tilde{\de}
		\coloneqq
		\Big( \frac{\de}{\de t_1}, \frac{1}{3} \frac{\de}{\de t_3},\dotsc \Big)\,,
	\end{equation}
	and define the \emph{BKP vertex operators} by
	\begin{align}
		Z^B(x,y)
		& \coloneqq
		\frac{1}{2} \frac{x-y}{x+y}
			e^{\tilde{\xi} (\mathbf{t},x) + \tilde{\xi}(\mathbf{t},y)}
			e^{-2 \tilde{\xi}(\tilde{\de}, x^{-1}) -2 \tilde{\xi}(\tilde{\de},y^{-1})} \,, \\
		\hat{Z}^B(x,y)
		& \coloneqq
		\frac{1}{2} \frac{x-y}{x+y}
		\left(
			e^{\tilde{\xi} (\mathbf{t},x) + \tilde{\xi}(\mathbf{t},y)}
			e^{-2 \tilde{\xi}(\tilde{\de}, x^{-1}) -2 \tilde{\xi}(\tilde{\de},y^{-1})}
			- 1 \right).
	\end{align}
\end{definition}

\begin{proposition}[\cite{DKM81,DJKM82}]\label{prop:BKP:vertex:action}
	The BKP vertex operators satisfy $Z^B(x,y) = -Z^B(y,x)$, and expanding $Z^B(x,y) = \sum_{j,k \in \Z} Z_{j,k} x^j y^{-k}$, the action
	\begin{equation}
		\sigma^B \colon \mf{b}_\infty \to \End (\Gamma ) \colon F_{j,k} \mapsto (-1)^k Z_{-j,k}
	\end{equation}
	is well-defined, and is intertwined with $ \pi $ by the boson-fermion correspondence $\Phi^B$. 

	\smallskip

	Similarly, $ \hat{Z}^B(x,y) = -\hat{Z}^B(y,x)$, and expanding $\hat{Z}^B(x,y) = \sum_{j,k \in \Z} \hat{Z}_{j,k} x^j y^{-k}$, the action
	\begin{equation}
		\hat{\sigma}^B \colon \hat{\mf{b}}_\infty \to \End (\Gamma ) \colon \hat{F}_{j,k} \mapsto (-1)^k \hat{Z}_{-j,k}
	\end{equation}
	is well-defined, and is intertwined with $ \hat{\pi} $ by the boson-fermion correspondence $\Phi^B$.
\end{proposition}

We can use the above result to translate the multiplication by the spin completed cycles as the \emph{spin completed cut-and-join operators} on $\Gamma$.

\begin{definition}
	Consider the map $\hat{\sigma}^B \circ \mc{F} \circ \gamma \colon \bigoplus_{d = 0}^\infty \mc{Z}_d \to \End (\Gamma)$. Define the \emph{$(r+1)$-th spin completed cut-and-join operator} $W_{r+1} \in \End (\Gamma)$ as the image of the operator of multiplication by the spin completed $(r+1)$-cycle $\bar{c}_{r+1}$.
\end{definition}

\begin{proposition}\label{prop:B:CJ:gen:series}
	The generating series of the spin completed cut-and-join operator is given by
	\begin{equation}
		W^{B}(z)
		\coloneqq
		\sum_{\substack{r \ge 1 \\ r \textup{ even}}} \frac{1}{r!} W_{r+1}^{B} z^{r+1}
		=
		\coth(\tfrac{z}{2}) \sum_{n=1}^\infty \sum_{\substack{k_1 + \dotsb + k_n = 0\\ k_i \textup{ odd}}}
			\frac{2^{n-2}}{n!}
			\normord{
				\prod_{i=1}^n \varsigma(k_i z) \frac{a_{k_i}}{k_i}
			},
	\end{equation}
	where 
	\begin{equation}
		\varsigma(z) \coloneqq 2\sinh(\tfrac{z}{2}),
		\qquad\qquad
		a_k \coloneqq \rho (\alpha^B_k)
		=
		\begin{cases} k\de_{t_k} & k >0 \\ \frac{1}{2} t_{-k} & k < 0\end{cases}\,.
	\end{equation}
	As usual, the normal product of the $a_k$'s put all derivatives on the right-hand side.
\end{proposition}

\begin{proof}
	We track the completed cycles through the diagram \eqref{eqn:B:diagram}. By \cref{defn:spin:compl:cycles}, $ \gamma (\bar{c}_{r+1}) = \frac{1}{r+1} p_{r+1}$, so $ \mc{F} \circ \gamma (\bar{c}_{r+1}) = \frac{\mc{F}^B_{r+1}}{r+1}$. Thus, using \cref{prop:BKP:vertex:action}, we find
	\begin{equation*}
	\begin{split}
		W^{B}(z)
		&= \sum_{s=0}^\infty \frac{1}{(2s+1)!} \hat{\sigma}^B(\mc{F}^B_{2s+1}) \, z^{2s+1}
		\\
		&= \frac{1}{2} \sum_{s=0}^\infty \sum_{k \in \Z} \frac{(kz)^{2s+1}}{(2s+1)!} \hat{\sigma}^B(\hat{F}_{k,-k})
		\\
		&= \frac{1}{4} \sum_{k \in \Z} (e^{kz} - e^{-kz}) (-1)^k \hat{Z}_{k,k}
		\\
		&= \frac{1}{2} [x^0] \hat{Z}^B(xe^{z/2}, -xe^{-z/2}).
	\end{split}
	\end{equation*}
	Computing $\hat{Z}^B(xe^{z/2}, -xe^{-z/2})$, we find
	\begin{equation*}
		W^{B}(z)
		=
		 \frac{1}{4} [x^0 ] \coth(\tfrac{z}{2}) \Bigg(
			\exp \bigg( \sum_{j=0}^\infty \varsigma \big( (2j+1)z) t_{2j+1}x^{2j+1} \bigg)
			\exp \bigg( 2 \sum_{k=0}^\infty \varsigma\big( (2k+1)z) \frac{\del_{t_{2k+1}}}{2k+1}  x^{-2k-1} \bigg) - 1
		\Bigg) \, ,
	\end{equation*}
	and expanding the exponentials finishes the proof.
\end{proof}

\subsection{B-Okounkov--Pandharipande operators}

From the proof of \cref{prop:B:CJ:gen:series}, we see that the generating series of completed cut-and-join operators is the constant coefficient in $x$ of $ Z^B(xe^{z/2},-xe^{-z/2})$. The other coefficients are interesting as well: they form an algebra that interpolates between the cut-and-join operators $\mc{F}^{B}$ and the bosonic operators $\alpha^B$. In the non-spin case, such an algebra was introduced by Bloch--Okounkov~\cite{BO00} and packed into generating functions by Okounkov--Pandharipande in \cite{OP06}. There is no particular reason to define them in a representation, so we give them as follows:

\begin{definition}
	The \emph{B-Okounkov--Pandharipande operators} are defined by
	\begin{align}\label{B-OPops}
		\mc{E}^{B}_m(z)
		& \coloneqq
		-\frac{1}{2} \sum_{k \in \Z} e^{-(k+m/2)z} F_{k,-k-m} \in \mf{b}_\infty \bbraket{z} \, ,
		\qquad\quad m \in \Z \, , \\
		\hat{\mc{E}}^B_0(z)
		& \coloneqq
		-\frac{1}{2} \sum_{k \in \Z} e^{-kz} \hat{F}_{k,-k} \in \hat{\mf{b}}_\infty \bbraket{z}\,.
	\end{align}
	Notice that the boson operators and the fermionic cut-and-join operators are given by
	\begin{equation}
		\alpha_m^B = \mc{E}^{B}_m(0)\,,
    \qquad\quad
		\mc{F}^B_{r+1} = (r+1)! [z^{r+1}] \hat{\mc{E}}^B_0(z)\,.
	\end{equation}
\end{definition} 

\begin{remark}
	The operators $\Omega_{mn}$ of \cite[section 4.3]{MMNO21} are the coefficients of the B-Okounkov--Pandharipande operators: $ \hat{\mc{E}}^B_m(z) = \sum_{n} \frac{1}{n!}\Omega_{mn} z^n$. 
\end{remark}

\begin{lemma}\label{lem:B:OP:as:vertex:operator}
	Under the map $\sigma^B$, the B-Okounkov--Pandharipande operators are given by
	\begin{equation}
		\sigma^B \bigl( \mc{E}^{B}_m(z) \bigr) = -\frac{1}{2} [x^m] Z^B(xe^{z/2}, -xe^{-z/2})\,.
	\end{equation}
	In particular, specialising the above equation, we find that the generating function of completed cut-and-join operators $W^B(z)$ can be obtained from the B-Okounkov--Pandharipande operators as
	\begin{equation}
		W^B(z) = \hat{\sigma}^B\big( \hat{\mc{E}}^B_0(z)\big) \,.
	\end{equation}
\end{lemma}

\begin{proof}
	Straightforward computation, using that $W^B $ is odd.
\end{proof}

We collect here some other useful properties of the B-Okounkov--Pandharipande operators. The first one is the vacuum expectation value formula, that in the A-setting reads
\begin{equation}
 	\langle\mathcal{E}_m(z)\rangle = \frac{\delta_m}{\varsigma(z)},
 	\qquad\qquad
 	\varsigma(z) \coloneqq 2\sinh(\tfrac{z}{2}).
 \end{equation} 
In the B-setting, we use the following analogous definition (the archaic Greek letter $\qoppa$ reads "\textit{qoppa}")
\begin{equation}
	\qoppa(z)
	\coloneqq
	\frac{1}{2} \cosh(\tfrac{z}{2})
	=
	\frac{1}{2} + \frac{z^2}{16} + \frac{z^4}{768} + \frac{z^6}{92160} + O(z^8)
\end{equation}	

\begin{lemma}\label{lem:OP:VEV}
	The following results hold.
	\begin{enumerate}
		\item
		The vacuum expectation values of the B-Okounkov--Pandharipande operators are given by
		\begin{equation}
			\Braket{ \mc{E}^{B}_m(z) }
			=
			\frac{\delta_m}{4} \coth \left( \tfrac{z}{2} \right) = \delta_m \frac{\qoppa(z)}{\varsigma(z)}
		\end{equation}

		\item
		The operators $\mc{E}^{B}_m$ and $\hat{\mc{E}}^{B}_0$ obey the parity relations
		\begin{equation}
			\mc{E}^{B}_m(-z) = (-1)^{m+1} \mc{E}^{B}_m(z),
			\qquad\quad
			\hat{\mc{E}}^{B}_0(-z) = - \hat{\mc{E}}^{B}_0(z).
		\end{equation}

		\item
		The subspace of $\mf{b}_\infty $ spanned by the coefficients $[z^k]\mc{E}^{B}_m(z)$ is a Lie subalgebra. Explicitly,
		\begin{equation}\label{eqn:spin:OP:comm}
			\bigl[ \mc{E}^{B}_m(z),\mc{E}^{B}_n(w) \bigr]
			=
			\frac{1}{2}
			\varsigma\bigl(\det\begin{bsmallmatrix}
				m & z \\
				n & w
			\end{bsmallmatrix}\bigr) \mc{E}^{B}_{n+m}(z+w)
			+ \frac{(-1)^n}{2}
			\varsigma\bigl(\det\begin{bsmallmatrix}
				m & - z \\
				n & w
			\end{bsmallmatrix}\bigr) \mc{E}^{B}_{n+m} (z-w) \,.
		\end{equation}
	\end{enumerate}
\end{lemma}

\begin{proof}
	The first property follows from \cref{lem:VEV:Fock:basis}:
	\begin{equation*}
	\begin{split}
		\Braket{ \mc{E}^{B}_m(z) }
		&=
		\biggl\langle -\frac{1}{2} \sum_{k \in \Z} e^{-(k+m/2)z} F_{k,-k-m} \biggr\rangle
		=
		-\frac{1}{2} \delta_m \biggl( \sum_{k>0} e^{-kz} + \frac{1}{2} \biggr) 
		\\
		&=
		-\frac{\delta_m}{2} \Big( \frac{e^{-z}}{1- e^{-z}} + \frac{1}{2} \Big)
		=
		-\frac{\delta_m}{4} \frac{1+ e^{-z}}{1-e^{-z}} = -\frac{\delta_m}{4} \coth( -\tfrac{z}{2} ).
	\end{split}
	\end{equation*}
	For the parity property, use that $\hat{F}_{j,k} = (-1)^{j+k+1} \hat{F}_{k,j}$ and \cref{lem:F:hat:symmetry:Heisenberg}. To conclude, the commutation relation follows from an explicit calculation.  	\begin{equation*}
	\begin{split}
		\bigl[ \mc{E}^{B}_m(z),\mc{E}^{B}_n(w) \bigr]
		&= \frac{1}{4} \sum_{k,l \in \Z} e^{-(k+m/2)z-(l+n/2)w} \big[ F_{k,-k-m},F_{l,-l-n} \big]
		\\
		&= \frac{1}{4} \sum_{k,l \in \Z} e^{-(k+m/2)z-(l+n/2)w} \Big( \delta_{l-k-m} F_{k,-l-n} - \delta_{k-l-n} F_{l,-k-m} \\
		&\hspace{4cm}  - (-1)^m \delta_{k+l} F_{-k-m,-l-n} + (-1)^m \delta_{-k-m-l-n} F_{l,k} \Big)
		\\
		&= \frac{1}{4} \sum_{k \in \Z} e^{-(k+m/2)z-(k+m+n/2)w} F_{k,-k-m-n} - \frac{1}{4} \sum_{l \in \Z} e^{-(l+n+m/2)z-(l+n/2)w} F_{l,-l-m-n} \\
		& \qquad \qquad - \frac{1}{4} \sum_{l \in \Z} e^{-(-l+m/2)z-(l+n/2)w} (-1)^m F_{l-m,-l-n} \\
		& \qquad \qquad + \frac{1}{4} \sum_{l \in \Z} e^{(l+n+m/2)z-(l+n/2)w} (-1)^m F_{l,-l-m-n}
		\\
		&= \frac{1}{4} \sum_{k \in \Z} \Big( e^{-(k+m/2)z-(k+m+n/2)w} - e^{-(k+n+m/2)z-(k+n/2)w} \\
		& \qquad - e^{(k+m/2)z-(k+m+n/2)w} (-1)^m + e^{(k+n+m/2)z-(k+n/2)w} (-1)^m \Big) F_{k,-k-m-n}
		\\
		&= \frac{1}{4} \sum_{k \in \Z} \Big( e^{nz/2-mw/2} - e^{-nz/2+mw/2} \Big) e^{-(k+(m+n)/2)(z+w)}F_{k,-k-m-n}\\
		&  \qquad+(-1)^m \frac{1}{4} \sum_{k \in \Z} (-e^{-nz/2-mw/2} + e^{nz/2+mw/2} ) e^{(k+(m+n)/2)(z-w)} F_{k,-k-m-n}
		\\
		&= -\sinh \Big( \frac{nz-mw}{2} \Big)  \mc{E}^{B}_{n+m}(z+w) + (-1)^m \sinh \Big( \frac{-nz-mw}{2} \Big) \mc{E}^{B}_{n+m} (w-z) \,.
	\end{split}
	\end{equation*}
	Now simplify the signs using the fact that the hyperbolic sine is an odd function, and employ the parity property above.
\end{proof}

\section{Properties of spin double Hurwitz numbers}
\label{sec:properties:spin:HNs}

In this section we employ the algebra of spin bosonic operators to analyse and derive several structural properties about spin double Hurwitz numbers. These properties will be described and referred to as:
\begin{enumerate}
	\item Vacuum expectation in terms of the algebra of $\mc{E}_m^{\textup{B}}(z)$;
	\item Generating series;
	\item Strong chamber polynomiality;
	\item Wall-crossing formulae;
	\item Integrability and cut-and-join equation.
\end{enumerate}
Moreover, for single Hurwitz numbers, we derive the following property:
\begin{enumerate}
	\item[(6)] Quasi-polynomiality.
\end{enumerate}
Each of these results have been observed and proved for several non-spin Hurwitz enumerative problems over the past years, by developing new techniques via the Fock space formalism, and represent major advancements in the field of Hurwitz theory. By now these techniques are more consolidated (see e.g.\cite{HKL18, KLS19}), and we can prove analogous results by a suitable adaptation of these methods. We therefore derive the results and refer to the original proofs, only pointing out the necessary adaptations.

\subsection{Vacuum expectation in terms of the algebra of \texorpdfstring{$\mc{E}_m^{\mathrm{B}}(z)$}{B-OP operators}}

Point (1) is a straightforward application of Gunningham's formula. The case $r = 2$ can be found in \cite[equation~3.10]{Lee19}, and the case $r > 2$ can be deduced from \cite{MMN20,MMNO21}.

\begin{proposition}\label{prop:spin:double:HNs:VEV}
	For $\mu, \nu$ odd partitions of lengths $m$ and $n$, the disconnected spin double Hurwitz numbers with $(r+1)$-completed cycles are given by
	\begin{equation}
		h_{g;\mu,\nu}^{\bullet, r,\theta}
		=
		\frac{2^{1-g} (r!)^b}{b! \cdot \prod_i \mu_i \prod_j \nu_j} \bigl[ z_1^{r+1} \cdots z_b^{r+1} \bigr]
		\Bigg\<
			\prod_{i=1}^{\ell(\mu)} \mc{E}^B_{\mu_i}(0)
			\prod_{p=1}^b \hat{\mc{E}}^B_0(z_p) 
			\prod_{j=1}^{\ell(\nu)} \mc{E}^B_{-\nu_j}(0)
		\Bigg\> \, .
	\end{equation}
	Here $b = \frac{2g - 2 + m + n}{r}$ is given by the Riemann--Hurwitz formula.
\end{proposition}

\begin{proof}
	Use \cref{lem:power:sum:eigenvalue,cor:boson:action} on Gunningham's formula for spin double Hurwitz numbers \cref{eqn:Gunningham:spin:double:HNs}.
\end{proof}

\subsection{Strong chamber polynomiality}
\label{subsec:chamber:poly:WC}

It is known for the non-spin case that some rich structure underlies the double Hurwitz numbers $h^{r,\theta}_{g;\mu,\nu}$, when considered as functions of $\mu$ and $\nu$ for fixed lengths $m$ and $n$. In fact, the entries of the odd partitions $\mu$ and $\nu$ can be seen as coordinates in an affine space that gets divided into chambers: in each chamber the Hurwitz numbers can be represented by a chamber-dependent polynomial, and its homogeneous decomposition also enjoys several constraints.

\begin{definition} \label{defn:space:arguments:spin:double:HNs}
	Let us define the subspace 
	\begin{equation}
		\mc{H}(m,n)
		\coloneqq
		\Biggl\{ (\mu,\nu ) \; \Bigg| \; \mu \in \N^m,\nu \in \N^n \textrm{ odd, such that }\sum_{i=1}^m \mu_i = \sum_{j=1}^n \nu_j \Biggr\} \subset \N^m \times \N^n \,,
	\end{equation}
	where $\mu = (\mu_1, \dots, \mu_m)$ and $\nu = (\nu_1, \dots, \nu_n)$ and we view spin double Hurwitz numbers as a function in the following sense:
	\begin{equation}
		h_{g}^{r,\theta} : \mc{H}(m,n)\to\Q \colon \qquad  (\mu,\nu)\mapsto h^{r,\theta}_{g;\mu,\nu}.
	\end{equation}
\end{definition}

\begin{definition}\label{defn:hyp:arrangement}
	Let $I\subsetneq \bbraket{m}$ and $J\subsetneq \bbraket{n}$ be non-empty proper subsets. Define the hyperplane (or \emph{wall}) indexed by $(I,J)$ as the set
	\begin{equation}
		\mc{W}_{I,J}
		\coloneqq
		\Biggl\{ (\mu,\nu) \in \mc{H}(m,n) \; \Bigg| \; \sum_{i \in I} \mu_i = \sum_{j \in J} \nu_j \Biggr\}.
	\end{equation}
	We define the \emph{hyperplane arrangement} $\mc{W}(m,n) \subset \mc{H}(m,n)$ to be the union of all the walls $\mc{W}_{I,J}$. A connected component of $\mc{H}(m,n) \setminus \mc{W}(m,n)$ is called a \emph{chamber}, and is denoted with the letter $\mf{c}$.
\end{definition}

\begin{theorem}[Strong piecewise polynomiality]\label{thm:piecewise}
	Let $g$ be a non-negative integer and let $m$, $n$ be positive integers such that $(g, n+m) \neq (0,2)$ and $b = \frac{2g-2+m+n}{r}$ is a positive integer. Then the function
	\begin{equation}
		h_{g}^{r,\theta} \colon \mc{H}(m,n) \to \Q \colon
		\qquad
		(\mu,\nu)\mapsto h^{r,\theta}_{g;\mu,\nu}
	\end{equation}
	is piecewise polynomial with the walls given by the hyperplanes of $\mc{W}(m,n)$. In other words, within each chamber $\mf{c}$ of the hyperplane arrangement $\mc{W}(m,n)$ there exists a polynomial $P_{g}^{r,\theta, \mf{c}}$ such that
	\begin{align}
		h^{r,\theta}_{g;\mu,\nu} &=
		P_{g}^{r,\theta, \mf{c}}(\mu,\nu),
		\qquad\qquad\qquad
		\text{for all }\;\; (\mu,\nu)\in \mf{c}.
	 \end{align}
	Moreover, $P_{g}^{r,\theta, \mf{c}}$ has the homogeneous degree decomposition
	\begin{equation}
		P_{g}^{r,\theta, \mf{c}}(\mu, \nu)
		=
		\sum_{k=0}^g P_{g,k}^{r,\theta, \mf{c}}(\mu, \nu),
		\qquad
		\deg_{\mu,\nu}(P_{g,k}^{r,\theta, \mf{c}}) = 2g - 1 + b -2k.
	\end{equation}
\end{theorem}

Its proof follows from \cref{prop:spin:double:HNs:VEV} and from the generating series of the vacuum expectations involved, which are computed in the next section.

\subsection{Generating series}

This section derives the generating series for double spin Hurwitz numbers by computing the vacuum expectations appearing in \cref{prop:spin:double:HNs:VEV}. The method used is based on a simple commutation procedure: positive energy operators are commuted to the right until they vanish against the vacuum. This refines an algorithm by Johnson for double Hurwitz numbers \cite{Joh15}. We start by defining some particular subclass of operators $\mc{E}^{\textup{B}}$ that play a role in this commutation procedure.

\begin{definition}
	For $I \subseteq \bbraket{n}$, $J \subseteq \bbraket{m}$, $K \subseteq \bbraket{b}$, and $\textbf{a} \in \{\pm 1\}^{K}$, define
	\begin{equation}
		\mc{E}'(I, J, K, \textbf{a}) = \mc{E}^B_{|\mu_I| - |\nu_J|}(z_{K, \textbf{a}})
	\end{equation}
	where $z_{K, \textbf{a}} \coloneqq \sum_{k \in K} a_k z_k$ and $\mu_I = \sum_{i \in I} \mu_i$ for a partition $\mu$. Define moreover for disjoint pairs $ I,L \subseteq \bbraket{n}$, $J,M \subseteq \bbraket{m}$, $K,N \subseteq \bbraket{b}$, and for $ \textbf{a} \in \{ \pm 1\}^K $, $ \textbf{b} \in \{ \pm 1\}^N$.
	\begin{equation}
		\bvs{I}{J}{K}{\textbf{a}}{L}{M}{N}{\textbf{b}} =\varsigma\left(\text{det}\begin{bmatrix} |\mu_I|-|\nu_J| & z_{K, \textbf{a}} \\ |\mu_L|-|\nu_M| & z_{N, \textbf{b}} \end{bmatrix}\right)\,.
	\end{equation}
\end{definition}

The commutation relation \eqref{eqn:spin:OP:comm} expressed in the new notation turns into:
\begin{equation}
\begin{split}
	\bigl[ \mc{E}'(I, J, K, \textbf{a}), \mc{E}'(L, M, N, \textbf{b}) \bigr]
	& =
	\frac{1}{2} \bvs{I}{J}{K}{\textbf{a}}{L}{M}{N}{\textbf{b}} 
		\mc{E}'(I \cup L, J \cup M, K \cup N, \textbf{a} \cup \textbf{b}) \\
	&\quad
	- \frac{(-1)^{|I| + |J|}}{2} \bvs{I}{J}{K}{-\textbf{a}}{L}{M}{N}{\textbf{b}} 
		\mc{E}'(I \cup L, J \cup M, K \cup N, (-\textbf{a})\cup \textbf{b})
	\,,
\end{split}
\end{equation}
On the other hand, double Hurwitz numbers are expressed in terms of the $\mc{E}'$ as 
\begin{equation}\label{eq:genvacuum}
	h_{g;\mu,\nu}^{\bullet, r,\theta}
	= (-1)^{m+n}
	\frac{2^{1-g}(r!)^b }{b!} [z_1^{r+1} \cdots z_b^{r+1}]
	\corr{
		\prod_{i=1}^n \frac{\mc{E}'(\{i\},\emptyset,\emptyset,\emptyset)}{\mu_i}
		\prod_{p=1}^b \hat{\mc{E}}'(\emptyset, \emptyset,\{p\}, \{+1\})
		\prod_{j=1}^m \frac{\mc{E}'(\emptyset,\{j\},\emptyset,\emptyset)}{\nu_j}
	}\,,
\end{equation}
where the $\hat{\mc{E}}'$ symbol refers as usual to the absence of the central extension term.	
	
\begin{definition}
	A \emph{commutation pattern} $P$ is a set of tuples
	\begin{equation}
		\left\{
			\bigl( I^P_t, J^P_t, K^P_t, \textbf{a}^P_t; L^P_t, M^P_t, N^P_t, \textbf{b}^P_t \bigr)
		\right\}_{t \in \bbraket{n+m+b-1}}
	\end{equation}
	where for each $t$, $I^P_t, L^P_t \subseteq \bbraket{n}$, $J^P_t, M^P_t \in \bbraket{m}, K^P_t, N^P_t \in \bbraket{b}, \textbf{a}^P_t \in \{\pm 1\}^{K^P_t}, \textbf{b}^P_t \in \{\pm 1\}^{N^P_t}$, such that we get a non-vanishing contribution to the vacuum expectation value \eqref{eq:genvacuum} when we go through the algorithm that commutes the right-most positive-energy operator to the right, in such a way that the $t$-th commutator computed is 
	\begin{equation}
		\bigl[
			\mc{E}'(I^P_t, J^P_t, K^P_t, \textbf{a}^P_t),
			\mc{E}'(L^P_t, M^P_t, N^P_t, \textbf{b}^P_t)
		\bigr].
	\end{equation}
	A fundamental point is that commutation patterns do not depend on the specific partitions $\mu$ and $\nu$, but only on the signs of the expressions $|\mu_I| - |\nu_J|$ (that is, the energies of the operators produced during the commutation process), and therefore only on the chamber $\mf{c}$ considered. This consideration allows to define the set $\CP^{\mf{c}}$ of commutation patterns $P$ relative to a chamber $\mf{c}$ as independent of $(\mu,\nu)$.
\end{definition}

\begin{remark}
	Within a chamber $\mf{c}$ of $\mc{W}(n,m)$ there are no disconnected covers. In fact, in order to have one of a certain degree $d$, a splitting $\bbraket{n} = I \cup I^{\vee}$ and $\bbraket{m} = J \cup J^{\vee}$ such that $|\mu_I| = |\nu_J| = d$ is needed, for the ramifications indices of $\mu_I$ and of $\nu_J$ to lie on the same connected component of the cover. Therefore within a chamber $\mf{c}$ of $\mc{W}(n,m)$ the notations $h_{g;\mu,\nu}^{\bullet, r,\theta}$ and $h_{g;\mu,\nu}^{\circ, r,\theta}$ coincide and we refer to them simply as $h_{g;\mu,\nu}^{r,\theta}$.
\end{remark}

\begin{theorem}\label{thm:HNgenseries}
	Within a chamber $\mf{c}$ of $\mc{W}(n,m)$ we have:\label{thm:genserieschamber}
	\begin{equation}
	\label{eq:genserieschamber}
		h_{g;\mu,\nu}^{r,\theta}
		=
		\frac{2^{1-g}(r!)^b }{b! \cdot \prod_{i} \mu_i \prod_{j} \nu_j} [z_1^{r+1} \cdots z_b^{r+1}]
		\left(\frac{1}{2}\right)^{\tau}
		\sum_{P \in \CP^{\mf{c}}} 
				\frac{(-1)^{\iota(P)} }{\varsigma\bigl(
					z_{\bbraket{b}, \textbf{${\textup{\textbf{a}}}^P_\tau \cup {\textup{\textbf{b}}}^P_\tau$}}
				\bigr) }
				\qoppa \bigl(
					z_{\bbraket{b}, {\textup{\textbf{a}}}^P_\tau \cup {\textup{\textbf{b}}}^P_\tau}
				\bigr)
		\prod_{t=1}^{\tau} 
		\bvs{I^P_t}{J^P_t}{K^P_t}{\textup{\textbf{a}}^P_t}{L^P_t}{M^P_t}{N^P_t}{\textup{\textbf{b}}^P_t},
	\end{equation}
	where $\tau = m+n+b-1$ is the total number of commutations for each commutation pattern, and 
	\begin{equation}
		\iota(P) = |D| + \sum_{i \in D} |I^P_{i}| + |J^P_{i}|,
	\end{equation}
	where $D \subset [\tau]$ is the set indexing the times the summand $-{\textup{\textbf{a}}}^P_t$ was chosen (as opposed to the first summand, involving ${\textup{\textbf{a}}}^P_t$ with no minus sign), out of the $t$-th commutator.
\end{theorem}
 \begin{proof}
	The proof is a straightforward adaptation of the argument in \cite{SSZ12}.
\end{proof}
	
\begin{remark}
	This statement can be slightly generalised to disconnected Hurwitz numbers on the walls, see \cite[theorem~4.6]{SSZ12}.
\end{remark}

\begin{proof}[{Proof of \cref{thm:piecewise}}]
	The proof is an adaptation of the argument in \cite{SSZ12}. The only differences consist in the introduction of an extra signed summand at each commutator and the introduction of the $\coth(z)$ function replacing the $1/\sinh(z)$ function. The latter does not spoil the polynomial argument, as of course $\coth(z) = \cosh(z)/\sinh(z)$, so that multiplying by $\cosh(z)$ does not introduce new poles (and therefore, already proved to be removable) and also does not spoil the parity argument in the degrees.

	\smallskip

	The introduction of new summands does not change the fact that the sum over $\CP^{\mf{c}}$ and the product over $t$ are finite, and the degrees in the $\mu_i$ and in the $\nu_j$ are coupled to degrees in the $z_k$ as in \cite{SSZ12}, therefore collecting the coefficient of $z_1^{r+1} \cdots z_b^{r+1}$ again guarantees a polynomial in the $\mu_i$ and $\nu_j$. The parity of the functions involved determines again even jumps in number $g$ for what concerns the homogeneous degrees, again starting at the same top degree $2g - 1 + b$.
\end{proof}

\begin{remark}
	The only part of the statement \cite[theorem~6.4]{SSZ12} that is in principle not guaranteed anymore by \cref{thm:piecewise} is the positivity of the polynomials in the homogeneous decomposition. Such a statement could still exist in some form, but it falls beyond the scope of the current work.
\end{remark}

\subsubsection{Generating series for one-part double Hurwitz numbers}

This section contains the derivations of an explicit formula for one-part spin double Hurwitz numbers, that is, the case when the ramification over infinity is a total ramification $(d)$, where $d$ is the degree of the cover. We also derive two specialisations of this formula, having in mind conjectural applications in Gromov--Witten theory. 

\begin{proposition}\label{prop:one:part:spin:double:HNs}
	Let $\mu = (\mu_1,\dots,\mu_n)$ be an odd partition of $d$. The one-part spin double Hurwitz numbers are given by: 
	\begin{equation}
		h_{g;(d), \mu}^{r,\theta}
		=
		\frac{(r!)^b}{b!} \frac{d^{b-1}}{2^{b+g-1}} [z_1^{r} \cdots z_b^{r}] \prod_{p=1}^b \mc{S} \left(z_p d\right) 
		\sum_{k_1, \dots, k_{b} \in \{\pm 1\}}
		\frac{\prod_{i=1}^n \mc{S} \left(\mu_iz_{\bm{k}}\right)}{\mc{S} (z_{\bm{k}})} (z_{\bm{k}})^{n} \, \mc{K}(z_{\bm{k}}),
	\end{equation}
	where $z_{\bm{k}} \coloneqq \sum_{p=1}^bk_p z_p$. 
\end{proposition}
\begin{proof}
  It is a straightforward specialisation of \ref{thm:HNgenseries}. Since there is a single negative energy operator within the vacuum expectation, that operator has to be involved in every commutations until the end of the commutation procedure, therefore the sum over commutation patterns collapses.
\end{proof}

\begin{corollary}
	Let $\mu$ be an odd integer. The one-part spin single Hurwitz numbers are given by:
	\begin{equation}
		h_{g;(\mu)}^{r,\theta}
		=
		\frac{(r!)^b}{b! \mu !}  \frac{\mu^{b-1}}{2^{b+g-1}} [z_1^{r} \cdots z_b^{r}] \prod_{p=1}^b \mc{S} \left(z_p \mu\right) 
		\sum_{k_1, \dots, k_{b} \in \{\pm 1\}}
		\!\!\!\!\! \mc{S} (z_{\bm{k}})^{\mu - 1} (z_{\bm{k}})^{\mu} \, \mc{K}(z_{\bm{k}})\,,
	\end{equation}
	where $z_{\bm{k}} \coloneqq \sum_{p=1}^b k_p z_p$. 
\end{corollary}
\begin{proof}
	From \cref{defn:spin:HNs}, we see that $ h^{r,\theta}_{g;\mu} = \frac{1}{|\Aut (1^d)|} h^{r,\theta}_{g;\mu,(1^d)}$. Using that $ |\Aut (1^d)| = d!$ and $ d=\mu$ in our case, the result follows from \cref{prop:one:part:spin:double:HNs}.
\end{proof}

\begin{corollary}
	For $b=1$ we have:
	\begin{equation} 
		h_{g; (\mu)}^{r,\theta}
		=
		\frac{(r)!}{2^{g-1} \mu!} [z^{2g}] \, \qoppa(z) \mc{S}(z \mu) \mc{S} (z)^{\mu - 1},
		\qquad\qquad
		g = \frac{r-\mu + 1}{2}.
	\end{equation}
\end{corollary}

The formula above is conjecturally related to the Gromov--Witten correlator of spin $\P^1$ with a single point insertion relative to the partition $\mu$, $\langle \tau_{r/2}(\omega), \mu\rangle^{\mathbb{P}^1, +}_{|\mu|}$, via the conjectural spin GW/H correspondence mentioned in \cref{eqn:spin:GW:H}.

\subsection{Wall-crossing formulae}

We are now armed to define and derive wall-crossing formulae for double spin Hurwitz numbers.

\begin{definition}	
	 A \emph{wall-crossing formula} for spin double Hurwitz numbers is an expression for the quantity
	\begin{equation}
		\WC_{g,I,J}^{r,\theta}(\mu, \nu)
		\coloneqq
		h_g^{r,\theta}\big{|}_{\mf{c}_1} - h_g^{r,\theta}\big{|}_{\mf{c}_2}
	\end{equation}
	for two neighbouring chambers $\mf{c}_1$ and $\mf{c}_2$ separated by the wall $W_{I,J} = \{ (\mu,\nu) \mid |\mu_I| - |\nu_J| = 0 \}$, where we fix $\mf{c}_1$ as chamber with $|\mu_I| < |\nu_J|$ and $\mf{c}_2$ the chamber with $|\mu_I| > |\nu_J|$.
\end{definition}
     
\begin{definition}
	Define the generating series for multi-completed cycles as:
	\begin{equation}
		H_{\mu,\nu}(z_1,\dots,z_b)
		\coloneqq
		\frac{2^{1-g} }{b!} \corr{
			\prod_{i=1}^m \frac{\mc{E}^{B}_{\mu_i}(0)}{\mu_i}
			\prod_{p=1}^b \hat{\mc{E}}_0^B(z_p)
			\prod_{j=1}^n \frac{\mc{E}^{B}_{-\nu_j}(0)}{\nu_j}
			} \, .
	\end{equation}
	Moreover, within a chamber $\mf{c}$, let $H^{\textbf{v}}_{\mu,\nu}(z_{\bbraket{b}})$ be the sum of the contributions in the statement of theorem \ref{thm:genserieschamber} of all those commutation patterns $P \in \CP^{\mf{c}}$ whose last sign vector $\textbf{a}^P_{\tau} \cup \textbf{b}^P_{\tau} $ is $\textbf{v}$.
\end{definition}   

\begin{remark} 
	Piecewise polynomiality results extend to multi-completed cycles generating series, in the sense that collecting the coefficient of arbitrary $z_1^{r_1 + 1} \cdots z_b^{r_b + 1}$ (as opposed to along the diagonal $r_i = r$) imposes completed cycles of different sizes on the ramifications, but the piecewise polynomiality structure still stands. 
\end{remark}  

What is peculiar about wall-crossing formulae is that the vast majority of the terms arising from the vacuum expectations defining $\WC_{g,I,J}$ cancel out, and what remains can be expressed as a finite sum of quadratic terms in the generating series $H$.

\begin{theorem}[Wall-crossing formula]\label{thm:wall:crossing}
	Let $\delta = |\mu_I| - |\mu_J|$. We have:
	\begin{equation}
	\begin{split}
		&\WC_{g,I,J}^{r, \theta}(\mu,\nu)
		= 		[z_1^{r+1} \cdots z_b^{r+1}].\\
		&\sum_{\substack{
			K \sqcup K^{c} = \bbraket{b} \\
			\textup{\textbf{a}} \in \{\pm 1\}^K \\
			\textup{\textbf{b}} \in \{\pm 1\}^{K^c}
		}} 
		\frac{
			H^{\textup{\textbf{a}}}_{\mu_I, \nu_{J} \cup \delta}(z_{K})}{\langle{\mc{E}'(I,J,K,\textup{\textbf{a}}) \mc{E}_{-\delta}(0)}\rangle}
			 \cdot 
			\delta 
			\cdot
			\langle{ \mc{E}'(I,J,K,\textup{\textbf{a}}) \mc{E}'(I^c, J^c, K^c, \textup{\textbf{b}})}\rangle
			\cdot 
			\delta 
			\cdot 
			\frac{H^{\textup{\textbf{b}}}_{\mu_{I^c} \cup \delta, \nu_{J^{c}}}(z_{K^c})}{
			\langle{\mc{E}_{\delta}(0)\mc{E}'(I^c, J^c, K^c, \textup{\textbf{b}})}\rangle
		}	
	\end{split}
	\end{equation}
\end{theorem}

\begin{remark}
The formula above suggests the appearance of two poles arising from the vacuum expectations in the denominator, but these poles are removable. For example, the first vacuum expactation gives
\begin{equation*}
	\frac{1}{\langle {\mc{E}'(I,J,K,\textup{\textbf{a}}) \mc{E}_{-\delta}(0)}\rangle}
	=
	2\frac{\varsigma(z_{K, \textbf{a}})}{\varsigma(\delta z_{K, \textbf{a}})\qoppa(z_{K, \textbf{a}})}\frac{1}{((-1)^{\delta} - 1)},
\end{equation*}
giving a pole for even $\delta$. In fact, the entire commutator $\langle {\mc{E}'(I,J,K,\textup{\textbf{a}}) \mc{E}_{-\delta}(0)}\rangle$ simplifies against each term of the expansion of $H^{\textup{\textbf{a}}}_{\mu_I, \nu_{J} \cup \delta}(z_{K})$ by definition (we inserted the operator $ \mc{E}_{-\delta}(0)$ to the right and we ask the linear combination of the variables to be $z_{K, \textbf{a}}$, so it is always possible to run the commutation process in such a way that $\langle {\mc{E}'(I,J,K,\textup{\textbf{a}}) \mc{E}_{-\delta}(0)}\rangle$ is the last step of it, for each commutation pattern). The same reasoning leads to the simplification of $\langle{\mc{E}_{\delta}(0)\mc{E}'(I^c, J^c, K^c, \textup{\textbf{b}})}\rangle$ against $H^{\textup{\textbf{b}}}_{\mu_{I^c} \cup \delta, \nu_{J^{c}}}(z_{K^c})$.
\end{remark}

\begin{proof}
	The proof is a straightforward adaptation of the argument in \cite[theorem~6.6]{SSZ12}.
\end{proof}

\subsection{Integrability and cut-and-join equation}
\label{subsec:intgrability:CJ}

Integrability properties of spin double Hurwitz numbers were investigated in \cite{Lee19} and further generalised in \cite{MMN20,MMNO21}. Such results can also be easily obtained from \cref{sec:spin:algebra,sec:spin:HNs,sec:CJ:OP:operators}, and we include them below for completeness.

\begin{definition}
	Define the generating series of spin double Hurwitz numbers with $(r+1)$-completed cycles as
	\begin{equation}
		\mc{Z}^{r,\theta}(\bm{p},\bm{q};t)
		\coloneqq
		\sum_{m,n,g \ge 0} \sum_{\substack{\mu_1,\dots,\mu_n \text{ odd} \\ \nu_1,\dots,\nu_m \text{ odd}}}
			h_{g;\mu,\nu}^{\bullet, r,\theta} 2^{g-1} t^{b}
			\frac{p_{\mu_1} \cdots p_{\mu_m}}{m!} \frac{q_{\nu_1} \cdots q_{\nu_n}}{n!}.
	\end{equation}
	The summands with $b =\frac{2g-2+m+n}{r} \not\in \N$ or $|\mu| \ne |\nu|$ are set to $0$.
\end{definition}

\begin{theorem}[{\cite{Lee19,MMN20,MMNO21}}] ~ 
	\begin{enumerate}
		\item The generating series $\mc{Z}^{r,\theta}$ can be expressed as the following vacuum expectation value
		\begin{equation}
			\mc{Z}^{r,\theta}(\bm{p},\bm{q};t)
			=
			\Braket{
				e^{\sum_{m > 0} \alpha_{m}^{B} \frac{p_m}{m}}
				e^{t\frac{\mc{F}^B_{r+1}}{r+1}}
				e^{\sum_{n > 0} \alpha_{-n}^{B} \frac{q_n}{n}}
			} \, ,
		\end{equation}
		and it is a hypergeometric tau function of the $2$-BKP hierarchy (identically in $t$). Its expansion in terms of Schur $Q$-functions is given by
		\begin{equation}
			\mc{Z}^{r,\theta}(\bm{p},\bm{q};t)
			=
			\sum_{\lambda \in \mc{SP}} 2^{-\ell (\lambda)} e^{t p_r (\lambda )} Q_\lambda \big(\tfrac{1}{2}\bm{p}\big) Q_\lambda \big(\tfrac{1}{2}\bm{q}\big)\,.
		\end{equation}
		More generally, defining
		\begin{align}
			\mc{Z}^{\theta}(\bm{p},\bm{q};\bm{t})
			&=
			\Braket{
				e^{\sum_{m > 0} \alpha_{m}^{B} \frac{p_m}{m}}
				e^{\sum_{r>0} t_r\frac{\mc{F}^B_{r+1}}{r+1}}
				e^{\sum_{n > 0} \alpha_{-n}^{B} \frac{q_n}{n}}
			}
			\\
			&= \sum_{\lambda \in \mc{SP}} 2^{-\ell (\lambda)} e^{\sum_{r > 0} t_r p_r (\lambda )} Q_\lambda \big(\tfrac{1}{2}\bm{p}\big) Q_\lambda \big(\tfrac{1}{2}\bm{q}\big)\,,
		\end{align}
		this is a hypergeometric $2$-BKP tau-function in $ \bm{p}$ and $ \bm{q}$ (identically in $\bm{t}$). Furthermore, $\mc{Z}^\theta (\bm{p} = \delta_{j,1}, \bm{q} = \delta_{k,1}; \bm{t})$ is an $\infty$-soliton KdV tau-function in $ \bm{t}$ and $ \mc{Z}^\theta (\bm{p} = \delta_{j,1}, \bm{q}; \bm{t})$ is a $2$-BKP tau-function in $ \bm{q}$ and $ \bm{t}$.
		\item The generating series $\mc{Z}^{\theta}$ satisfies the the partial differential equation (the spin cut-and-join equation):
		\begin{equation}
			\frac{\de}{\de t_r} \mc{Z}^{\theta} = W^B_{r+1} \mc{Z}^{\theta}
		\end{equation}
		for any $r$. Consequently,
		\begin{equation}
			\frac{\de}{\de t} \mc{Z}^{r,\theta} = W^B_{r+1} \mc{Z}^{r,\theta}\,.
		\end{equation}
		\end{enumerate}
\end{theorem}

The essential ingredient for this theorem is Gunningham's formula, \cref{thm:Gunningham:formula}, which expresses Hurwitz numbers in terms of characters. Most results then follow almost immediately from the definitions and various isomorphisms, such as the boson-fermion correspondence. However, the fact that $Z^\vartheta( \delta_{j,1},\delta_{k,1};\bm{t}) $ is a soliton KdV tau-function in $\bm{t}$ is not just a formal consequence.

\smallskip

Furthermore, \cite[Theorem~1.1]{Lee19} finds a particular linear combination of completed $ 1$-, $2$-, and $3$-cycles for both the spin and non-spin case, such that the generating series for the spin case squares to the generating series of the non-spin case. It would be interesting to develop such a square formula for higher $r$ as well. Note that in general, any BKP tau-function squares to a KP tau-function \cite[proposition~4]{DJKM82}. The content of Lee's result is that both of these tau-functions have a geometric interpretation in terms of spin and ordinary Hurwitz numbers.

\subsection{Quasi-polynomiality}

Here we show how our main conjecture implies the polynomiality of single Hurwitz numbers up to a particular non-polynomial prefactor. This property is usually required and employed to prove that a certain enumerative problem is generated by topological recursion. We assume the main conjecture and derive it from the algebro-geometric interpretation for the multidifferentials $\omega_{g,n}$ (see \cref{sec:spin:HNs:ELSV}).

\begin{proposition}\label{cor:quasi-poly}
	Let $g \geq 0$ and $n\geq 1$ such that $2g - 2 +n >0$. \cref{conj:spin:HNs:TR} implies that, for $\mu = (\mu_1,\dots,\mu_n) \in \mc{OP}(d)$ with fixed remainders $\braket{\mu_1}, \dots, \braket{\mu_n}$ modulo $r$, there exists a polynomial $P_{g,n}^{ \braket{\mu},r}$ of degree $3g - 3 +n$ such that the spin Hurwitz numbers are given by
	\begin{equation}
		h_{g;\mu}^{r,\theta}
		= 
		\left( \prod_{i=1}^n \frac{ \mu_i^{[\mu_i]}}{[\mu_i]!}\right)
		P_{g,n}^{ \braket{\mu},r}(\mu_1, \dots, \mu_n).
	\end{equation}
\end{proposition}
\begin{proof}
	It is an immediate corollary of \cref{thm:ELSV:CohFT}.
%
\end{proof}

\section{A conjectural spectral curve for spin Hurwitz numbers}
\label{sec:conj}

We now move our attention to single spin Hurwitz numbers. In the non-spin case, single Hurwitz numbers are known to satisfy the topological recursion. As explained in \cref{sec:TR:CohFTs}, topological recursion is a way to compute enumerative quantities from $(0,1)$ and $(0,2)$ data (the spectral curve), recursively on $2g-2+n$ with $n = \ell(\mu)$.

\smallskip

In Hurwitz theory, the cut-and-join equation is, roughly speaking, a recursive procedure to compute Hurwitz numbers. Thus, one would expect topological recursion to hold and, applying the Eynard--DOSS correspondence of \cref{thm:Eyn:DOSS}, one can obtain an ELSV-type formula:
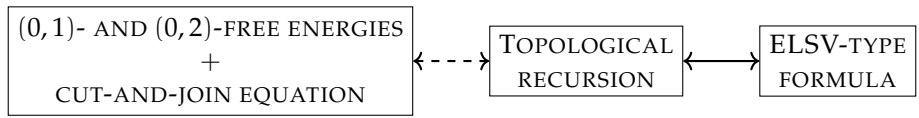
\begin{center}
\begin{tikzpicture}
	\matrix [column sep=10mm, row sep=5mm] {
	  \node (HT) [draw, shape=rectangle, align=center] {\textsc{$(0,1)$- and $(0,2)$-free energies} \\ $+$ \\ \textsc{cut-and-join equation}}; &
	  \node (TR) [draw, shape=rectangle, align=center] {\textsc{Topological} \\ \textsc{recursion}}; &
	  \node (ELSV) [draw, shape=rectangle, align=center] {\textsc{ELSV-type} \\ \textsc{formula}}; \\
	};
	\draw[<->, thick, dashed] (HT) -- (TR);
	\draw[<->, thick] (TR) -- (ELSV);
\end{tikzpicture}
\end{center}
To actually prove topological recursion for Hurwitz problems, an extra property is required, namely quasi-polynomiality. However, the actual shape of the spectral curve can be obtained looking at the $(0,1)$- and $(0,2)$-free energies.

\smallskip

In this section, we compute such unstable free energies for single spin Hurwitz numbers via the fermion formalism, and we conjecture that they can be computed by topological recursion on a specific spectral curve. We then give evidence for this conjecture by proving a closed formula for the $(g,1)$-free energies, and compare them to the topological recursion correlators for $g = 1$.

\subsection{The spectral curve}

Let $r$ be a positive even integer. Consider the spectral curve $(\Sigma,x,y,B)$ given by $\Sigma = \P^1$ and 
\begin{equation}\label{eqn:spin:Hurwitz:SC}
	x(z) = \log(z) - z^{r}, 
	\qquad\quad 
	y(z) = z\,,
	\qquad\quad
	B(z_1, z_2) 
	=
	\frac{1}{2} \bigg( \frac{1}{(z_1-z_2)^2} + \frac{1}{(z_1+z_2)^2} \bigg) dz_1 dz_2 \, . 
\end{equation}
As noted in the introduction, such spectral curve does not satisfy the usual axioms of topological recursion, since $B$ has poles when the two arguments approach different ramification points. Nevertheless, one can define the topological recursion correlators $\omega_{g,n}^{r,\theta}$ via the Eynard--Orantin topological recursion formula \eqref{eqn:TR}. Our main conjecture relates the multidifferentials $\omega_{g,n}^{r,\theta}$ and the spin Hurwitz numbers free energies.

\begin{conjecture}\label{conj:spin:HNs:TR}
	The coefficients obtained by expanding the correlators $\omega_{g,n}(z_1, \dots, z_n)$ near $e^{x_i} = 0$ are exactly the $(r+1)$-completed cycles spin single Hurwitz numbers:
	\begin{equation}
		\omega_{g,n}^{r,\theta}(z_1, \dots, z_n)
		-
		\delta_{g,0}\delta_{n,2} \, \omega_{0,2}^{r,\theta}(e^{x(z_1)},e^{x(z_2)})
		= 
		d_1 \cdots d_n
		F^{r,\theta}_{g,n} (e^{x_1}, \dotsc, e^{x_n})\Big|_{x_i = x(z_i)} \,.
	\end{equation}
\end{conjecture}

As we stated in the introduction, this is now a theorem, proved by Alexandrov and Shadrin after we informed them about our conjecture.

\begin{theorem}[{\cite[Theorem~7.1]{AS23}}]\label{thm:AS}
	\Cref{conj:spin:HNs:TR} holds.
\end{theorem}

The remaining part of the section is devoted to proving the conjecture for $g=0$ and for $(g,n) = (1,1)$.

\subsection{Spin Hurwitz numbers for \texorpdfstring{$g=0$}{g=0}}

First, we consider the case of the source curve being of genus zero, i.e. the numbers $ h^{r,\theta}_{0;\mu} $. In this case, two simplifications occur.

\smallskip

As for $g=0$ the source is $ \P^1 $, it has a unique spin structure, and we see that in \cref{defn:spin:HNs}
\begin{equation}
	p(N_{f,\mc{O}(-1)}) = p(\mc{O}(-1)) = 0
\end{equation}
for any cover in the count. Therefore, for $g=0$, the spin Hurwitz numbers are actual counts, without any sign.

\smallskip

For the other simplification, recall the Riemann--Hurwitz formula: for a ramified cover of Riemann surfaces $ f \colon S \to T$ of degree $ d$ with ramification profiles $\mu^i$,
\begin{equation}
	2 - 2g(S) = d(2-2g(T)) - \sum_i (d - \ell (\mu^i))\,.
\end{equation}
Looking at the definition of spin single Hurwitz numbers, \cref{defn:spin:HNs}, we see that $b = \frac{2g-2+ \ell (\mu ) + d}{r}$ is chosen such that we do get a genus $g$ source curve if we have one branch point with ramification profile $\mu$ and $b$ branch points with ramification profile $(r+1, 1,\dotsc, 1)$ (recall that we always have $g(T) = 0$). However, the definition uses the spin completed cycles $ \bar{c}_{r+1}$. From \cref{defn:spin:compl:cycles,prop:central:chars:are:susy} we see that
\begin{equation}
	\bar{c}_{r+1} - c_{r+1} \in \bigoplus_{d=0}^{r} \mc{Z}_d\,.
\end{equation}
In particular, for any $\mu \in \mc{OP}(d)$ with $ d \le r$, we have $r +1 - \ell (\mu \cup (1^{r+1-d} )) = d - \ell(\mu) < r = r + 1 - \ell ((r+1))$. It follows that for any cover of $\P^1$ with ramification profiles $\mu$ and $b$ choices of partitions occurring with non-zero coefficients in $ \bar{c}_{r+1}$, the genus of the source curve is at most $g$, with equality occurring exactly if we choose $(r+1)$ every time. This occurrence of source curves having lower than expected genus is called \emph{genus defect} in e.g. \cite{SSZ15}.

\smallskip

In particular, for $g = 0$, there are no contributions from the completions $ \bar{c}_{r+1} - c_{r+1}$, as this would require a connected source curve with negative genus. This means that in the definition of $ h^{r,\theta}_{0;\mu}$, we may replace the $ \bar{c}_{r+1} $ by the $c_{r+1}$.

\smallskip

The same argument holds for the non-spin case: there we may also replace the completed cycles $ \bar{C}_{r+1} \in \bigoplus_{d=0}^\infty Z \C [\mf{S}_d] \cong \C \{ \mc{P} \} $ by the non-completed $C_{r+1}$ for $g=0$ Hurwitz numbers (for more background on completed cycles for the non-spin case, see e.g. \cite{OP06,SSZ12,BKLPS21}). Hence, because $c_{r+1}$ and $C_{r+1}$ represent the same partition $(r+1)$, we get:

\begin{proposition}\label{prop:genus:zero:spin=nonspin}
	Let $r$ be a positive even integer and $ \mu \in \mc{OP}$. Then the spin and non-spin Hurwitz numbers with these arguments are equal:
	\begin{equation}
		h^{r,\theta}_{0;\mu} = h^r_{0;\mu}\,.
	\end{equation}
\end{proposition}

\begin{remark}
	Note that this equality does not necessarily hold for disconnected Hurwitz numbers, as for those the (arithmetic) genus may be negative.
\end{remark}

\begin{corollary}\label{cor:free:energy:genus:zero}
	Let $r$ be a positive even integer. The genus zero free energies for the spin case are the anti-symmetrisations in all arguments of those in the non-spin case:
	\begin{equation}
		F^{r,\theta}_{0,n}(e^{x_1},\dotsc, e^{x_n})
		=
		2^{-n} \sum_{\epsilon_1, \dotsc, \epsilon_n \in \{ \pm 1\}} \Biggl( \prod_{i=1}^n \epsilon_j \Biggr) F^{r}_{0,n} (\epsilon_1 e^{x_1}, \dotsc, \epsilon_n e^{x_n})\,.
	\end{equation}
	In particular, \cref{conj:spin:HNs:TR} holds in genus zero:
	\begin{equation}
		\omega_{0,n}^{r,\theta}(z_1, \dots, z_n)
		-
		\delta_{n,2} \, \omega_{0,2}^{r,\theta}(e^{x(z_1)},e^{x(z_2)})
		= 
		d_1 \cdots d_n
		F^{r,\theta}_{0,n} (e^{x_1}, \dotsc, e^{x_n})\Big|_{x_i = x(z_i)} \,.
	\end{equation}
\end{corollary}

\begin{proof}
	The first part follows directly from \cref{prop:genus:zero:spin=nonspin}: antisymmetrising is the same as restricting to odd powers of $e^{x_i}$ in the series expansion.\par
	For the second part, the unstable free energies for the non-spin case were computed in \cite{MSS13,KLPS19}. In that case the $(0,1)$-free energy is already odd, so it must be equal to its spin counterpart $\omega_{0,1}^{r,\theta}$. For $\omega_{0,2}^{r,\theta}$, the formula in the corollary is the antisymmetrisation of the non-spin case. The case $n > 2$ follows by induction.
\end{proof}

We give a second proof of the formula for $\omega_{0,1}^{r,\theta}$ in \cref{prop:omega01}, using the operator formalism. There, it is first step towards a general formula for $F_{g,1}^{r,\theta}$. A computation of $\omega_{0,2}^{r,\theta}$ in the operator formalism can be found in \cref{app:fermion:calcs}.

\subsection{Operator formalism for spin single Hurwitz numbers}
\label{subsec:operator:formalism:spin:HNs}

In analogy to \cite[definition 3.3]{KLPS19}, we get the following formula for spin single Hurwitz numbers as a consequence of Gunningham's formula \ref{eqn:Gunningham:spin:single:HNs}. It is the single version of \cref{prop:spin:double:HNs:VEV}.

\begin{proposition}\label{prop:spin:single:HNs:VEV}
	For $\mu$ an odd partition, the disconnected spin single Hurwitz numbers with $(r+1)$-completed cycles are given by
	\begin{equation}\label{eqn:spin:single:HNs:VEV}
		h_{g;\mu}^{\bullet, r,\theta}
		=
		2^{1-g} \bigl[ u^{rb} \bigr]
		\Bigg\<
			e^{\alpha^B_1} 
			e^{u^{r}\frac{\mc{F}^B_{r+1}}{r+1}}
			\prod_{i=1}^{\ell(\mu)} \frac{\alpha^B_{-\mu_i}}{\mu_i}
		\Bigg\> \, .
	\end{equation}
	Here $b = \frac{2g - 2 + \ell(\mu) + |\mu|}{r}$ is given by the Riemann--Hurwitz formula.
\end{proposition}

In order to use this formula, note that $\alpha_1^B $ and $\mc{F}^B_{r+1}$ both annihilate the vacuum: $\alpha_1^B |0\> = \mc{F}^B_{r+1} |0\> = 0$. Therefore, their exponentials fix the vacuum, and we may insert them in \cref{eqn:spin:single:HNs:VEV} as
\begin{equation}
	h_{g;\mu}^{\bullet, r,\theta}
	=
	2^{1-g} \bigl[ u^{rb} \bigr]
	\Bigg\<
		e^{\alpha^B_1} e^{u^{r}\frac{\mc{F}^B_{r+1}}{r+1}}
		\prod_{i=1}^{\ell(\mu)} \frac{\alpha^B_{-\mu_i}}{\mu_i}
		e^{-u^{r} \frac{\mc{F}^B_{r+1}}{r+1}} e^{-\alpha^B_1}
	\Bigg\> \, .
\end{equation}
The advantage of this formulation is that we now have $\ell(\mu)$ factors of the same shape, which we can describe uniformly.

\begin{lemma}
	Let $\mu$ be an odd positive integer. Then:
	\begin{equation}
		\mathcal{O}_{-\mu}^{B,r}(u)
		\coloneqq
		e^{u^{r}\frac{\mc{F}^B_{r+1}}{r+1}} \alpha^B_{-\mu} e^{-u^{r} \frac{\mc{F}^B_{r+1}}{r+1}}
		=
		-\sum_{l > \mu/2} e^{u^{r}\frac{l^{r+1} - (l - \mu)^{r+1}}{r+1}} F_{l,\mu - l} .
	\end{equation}
\end{lemma}

\begin{proof}
	Recall $\alpha_{-\mu}^{B} = - \sum_{l > \mu /2} F_{l,\mu -l}$ and $\mc{F}^B_{r+1} = \sum_{k>0} k^{r+1} F_{k,-k}$. More generally, write $\alpha_{-\mu}^f \coloneqq - \sum_{l> \mu /2} f_l(\mu) F_{l,\mu - l}$ for any function $f$ of $(l,\mu)$. Then
	\begin{align*}
		[\mc{F}^B_{r+1}, \alpha_{-\mu}^f ]
		&=
		\sum_{k>0} \sum_{l > \mu/2} - k^{r+1} f_l(\mu ) [F_{k,-k} ,  F_{l,\mu - l}]\\
		&=
		\sum_{k>0} \sum_{l> \mu/2} - k^{r+1} f_l(\mu ) \big( \delta_{-k+l} F_{k,\mu-l} - \delta_{k+l} F_{-k,\mu-l}  +\delta_{\mu-l-k} F_{l,k} - \delta_{k+\mu-l} F_{l,-k} \big).
	\end{align*}
	The second summand vanishes because of the conditions in the sum, and for the others, the $k$-sum gives certain restrictions in $\mu$:
	\begin{align*}
		[ \mc{F}^B_{r+1}, \alpha_{-\mu}^f ] 
		&=
		- \sum_{l > \mu/2} l^{r+1} f_l(\mu) F_{l,\mu-l}
		- \sum_{\mu/2 < l < \mu} (\mu-l)^{r+1} f_l(\mu) F_{l,\mu-l}
		+ \sum_{\mu < l} (l-\mu)^{r+1} f_l(\mu) F_{l,\mu-l}
		\\
		&=
		- \sum_{l>\mu/2} \big(l^{r+1} - (l-\mu)^{r+1} \big) f_l(\mu) F_{l,\mu-l} = \alpha_{-\mu}^g \, ,
	\end{align*}
	where $g_l(\mu) = \big(l^{r+1} - (l-\mu)^{r+1} \big) f_l(\mu)$. Applying this to $\alpha_{-\mu}^{B}$, we get
	\[
		(\ad_{\mc{F}^B_{r+1}})^n \alpha_{-\mu}^{B}
		=
		- \sum_{l>\mu/2} \big(l^{r+1} - (l-\mu)^{r+1} \big)^n F_{l,\mu-l}
	\]
	and, to conclude,
	\begin{align*}
		\mathcal{O}_{-\mu}^{B,r}(u)
		&=
		\exp{\Bigl( {\ad_{u^{r} \frac{\mc{F}^B_{r+1}}{r+1}}} \Bigr)} \alpha_{-\mu}^B \\
		&=
		\sum_{b=0}^\infty \frac{u^{rb}}{b!(r+1)^b} (\ad_{\mc{F}^B_{r+1}})^b (\alpha_{-\mu}^B )
		=
		-\sum_{b=0}^\infty  \sum_{l>\mu/2} \frac{u^{rb}}{b!} \Big(\frac{l^{r+1} - (l-\mu)^{r+1}}{r+1} \Big)^b  F_{l,\mu-l} \\
		&=- \sum_{l > \mu/2} e^{u^{r} \frac{l^{r+1} - (l - \mu)^{r+1}}{r+1}}  F_{l,\mu - l}.
	\end{align*}
\end{proof}

\begin{proposition}\label{prop:conj:operator}
	Let $\mu$ be an odd positive integer. Then:
	\begin{equation}\label{eqn:conj:operator}
		e^{\alpha_1^B} \mathcal{O}_{-\mu}^{B, r}(u) e^{-\alpha_1^B}
		=
		- \sum_{t=0}^\infty \sum_{ l = \frac{\mu+1}{2} - \lfloor \frac{t}{2} \rfloor }^\infty \!\!\!
		\frac{(-\Delta )^t}{t!} f(l) F_{l,-l + (\mu-t)} + \frac{1}{2} \frac{(-\Delta )^{\mu-1}}{\mu!} f ( 1 )\,,
	\end{equation}
	where $\Delta$ is the forward difference operator: $\Delta f (l) = f(l+1) - f(l)$, and $f(l) = e^{u^{r} \frac{l^{r+1} - (l-\mu )^{r+1}}{r+1}}$, omitting its dependence on $r, \mu, u$ for short.
\end{proposition}

The proof of this proposition is very combinatorial and can be found in \cref{app:fermion:calcs}.

\subsection{One-part spin Hurwitz numbers}

In order to compute spin Hurwitz numbers with $\ell(\mu) = 1$, also known as \emph{one-part Hurwitz numbers}, we may first realise that in this case, there is no difference between connected and disconnected counts: if the ramification profile over a point has length $1$, clearly this connects the source. We need to compute the quantity
\begin{equation}
	\corr{
		e^{ \alpha^B_{1}}
		e^{u^{r}\frac{\mc{F}^B_{r+1}}{r+1}}
		\frac{\alpha^B_{-\mu}}{\mu}
		e^{-u^{r}\frac{\mc{F}^B_{r+1}}{r+1}}
		e^{- \alpha^B_{1}}
		}
	=
	\frac{1}{\mu}
	\corr{
		e^{\alpha^B_{1}} \mathcal{O}_{-\mu}^{B,r}(u) e^{- \alpha^B_{1}}
	},
\end{equation}
for $\mu$ an odd integer. First note that, since $\big\< F_{i,j} \big\> = (-1)^i \big\< \phi_i \phi_j \big\> = \delta_{i+j}\delta_{j > 0} + \frac{1}{2} \delta_{i+j} \delta_{j}$, the vacuum expectation of each summand from \cref{eqn:conj:operator} including $\phi$'s vanishes. Hence we get that
\begin{equation}\label{OneAlphaCorrelator}
	h_{g;\mu}^{r,\theta}
	=
	2^{1-g} \frac{[u^{2g-1+\mu}]}{\mu}
	\corr{e^{\alpha^B_{1}} \mathcal{O}_{-\mu}^{B,r}(u) e^{- \alpha^B_{1}}}
	=
	\frac{[u^{2g-1+\mu}]}{2^g \mu^2 (\mu-1)!} (-\Delta )^{\mu-1} f(1).
\end{equation}
Let us start by considering the $(0,1)$-free energy. This is expressed in terms of the Lambert $W$ function, which satisfies $W(t)e^{W(t)} = t$ and it has an expansion
\begin{equation}
	W(t) = -\sum_{m=1}^\infty \frac{m^{m-1}}{m!} (-t)^m.
\end{equation}
More generally, we will consider the expansion of an arbitrary power of the Lambert function:
\begin{equation}\label{eqn:Lambert:expansion}
	\biggl( \frac{W(-t)}{-t} \biggr)^{\alpha} = \sum_{m=0}^\infty \frac{\alpha(m+\alpha)^{m-1}}{m!} t^m .
\end{equation}
See \cite[equation~(2.36)]{CGHJK96} as a reference.

\begin{proposition}\label{prop:omega01}
	For $r = 2s$, the unstable $(0,1)$-correlator is given by
	\begin{equation}
		\omega_{0,1}^{r,\theta} = y \, dx
		=
		dF_{0,1}^{r,\theta}(e^x) = \frac{1}{(-2s)^{1/2s}} W(-2se^{2sx})^{\frac{1}{2s}} dx.
	\end{equation}
\end{proposition}

\begin{proof}
	Recall that $f(l) = e^{u^{2s} P^{2s}_{\mu} (l)}$, where $P^{2s}_{\mu} (l) = \frac{l^{2s+1} - (l - \mu)^{2s+1}}{2s+1}$ is a polynomial in $l$ of degree $2s$ with leading coefficient $\mu$. Hence,
	\begin{align*}
		F_{0,1}^{r,\theta}(e^x)
		=
		\sum_{\substack{\mu > 0 \\ \mu \text{ odd}}} h_{0;\mu}^{r,\theta} e^{\mu x} 
		= \sum_{\substack{\mu > 0 \\ \mu \text{ odd}}} \frac{[ u^{\mu-1} ]}{\mu^2 (\mu-1)!} (-\Delta )^{\mu-1} f ( 1 ) e^{\mu x}
		= \sum_{m=0}^\infty \frac{[u^{2m} ]}{(2m+1)^2(2m)!} \Delta^{2m} f ( 1 ) e^{(2m+1) x}.
	\end{align*}
	Here we wrote $\mu = 2m+1$. Because $u$ only occurs in this formula in combinations $u^{2s} P^{2s}_{2m+1} (l)$, we find that $m$ must be divisible by $s$, so we write $m = sb$ to get the following formula.
	\begin{align*}
		F_{0,1}^{r,\theta}(e^x)
		&=
		\sum_{b=0}^\infty \frac{1}{(2sb+1)(2sb+1)!b!} \Delta^{2sb} (P^{2s}_{2sb+1})^b(1) e^{(2sb+1) x}
		\\
		&= \sum_{b=0}^\infty \frac{1 }{(2sb+1)(2sb+1)!b!} (2sb)! (2sb+1)^b e^{(2sb+1) x}
		\\
		&= \sum_{b=0}^\infty \frac{(2sb+1)^{b-2}}{b!} e^{(2sb+1) x}\,,
	\end{align*}
	using that $\Delta^d$ acting on a polynomial in $l$ of degree $d$ with leading coefficient $A$, followed by the evaluation at $l = 1$, gives $d! A$. Therefore,
	\begin{equation*}
		dF_{0,1}^{r,\theta}(e^x) = \sum_{b=0}^\infty \frac{(2sb+1)^{b-1}}{b!} e^{(2sb+1) x} dx .
	\end{equation*}
	Using the expansion formula \eqref{eqn:Lambert:expansion}, we get
	\begin{equation*}
	\begin{split}
		dF_{0,1}^{r,\theta}(e^x)
		& =
		e^x \sum_{b=0}^\infty \frac{(2bs+1)^{b-1}}{b!} e^{2bs x} dx
		= e^x \sum_{b=0}^\infty \frac{\frac{1}{2s}(b+\frac{1}{2s})^{b-1}}{b!} \big(2s e^{2sx} \big)^b dx 
		\\
		& =e^x \biggl( \frac{W(-2se^{2sx})}{-2se^{2sx}} \biggr)^{\frac{1}{2s}} dx 
		\\
		& = \frac{1}{(-2s)^{1/2s}} W(-2se^{2sx})^{\frac{1}{2s}} dx,
	\end{split}
	\end{equation*}
	which coincides with the $(0,1)$-correlator $\omega_{0,1}^{r,\theta} = y \, dx$ associated to the spectral curve of~\cref{eqn:spin:Hurwitz:SC}.
\end{proof} 

Now let us consider the free energies for $n=1$ and arbitrary $g$. Using the same techniques as before,
\begin{equation}
	F_{g,1}^{r,\theta}(e^x)
	=
	\sum_{\substack{\mu > 0 \\ \mu \text{ odd}}} h_{g;\mu}^{r,\theta} e^{\mu x}
	=
	\sum_{m=0}^\infty \frac{[u^{2g+2m} ]}{2^g(2m+1)^2(2m)!}\Delta^{2m} f(1) e^{(2m+1) x}.
\end{equation}
Again, take $r = 2s$. Because $u$ only occurs in this formula in combinations $u^{2s} P^{2s}_{2m+1} (l)$, we get $m+g = sb$, for some $b \geq \ceil{\frac{g}{s}}$ which indexes the exponential Taylor expansion. Hence we can absorb the $m$-sum in terms of the new $b$-sum as
\begin{equation}
	F_{g,1}^{r,\theta}(e^x)
	=
	\sum_{b = \ceil{g/s}}^\infty \frac{1}{2^g(2sb - 2g+1)^2(2sb - 2g)!b!} \Delta^{2sb-2g} (P_{2sb - 2g + 1}^{2s})^b(1) e^{(2sb - 2g+1) x}
\end{equation}
Recall that $P_{2sb - 2g + 1}^{2s}$ is a polynomial in $l$ of degree $2s$:
\begin{equation}
\begin{split}
	P_{2sb - 2g + 1}^{2s}(l) &=  \sum_{t=0}^{2s} \binom{2s+1}{t} \frac{l^t}{(2s+1)} (-1)^t (2sb-2g+1)^{2s+1 - t}\\
	& = (2sb-2g+1)l^{2s} - (2sb-2g+1)^2 \frac{(2s)}{2!}l^{2s-1} + (2sb-2g+1)^3 \frac{(2s)(2s-1)}{3!}l^{2s-2} - \cdots
\end{split}
\end{equation}
Therefore, taking the $b$-th power, it gives a polynomial $(P_{2sb - 2g + 1}^{2s})^b(l) = \sum_{a=0}^{2sb} C_a l^{2sb - a}$ in $l$ of degree $2sb$ with coefficients
\begin{equation}
\begin{split}\label{eqn:Ca}
	C_0 & = (2sb-2g+1)^b, \\
	C_1 & = -\binom{b}{1} \binom{2s+1}{2}\frac{(2sb-2g+1)^{b+1}}{2s+1}, \\
	C_2 & = \binom{b}{1}  \binom{2s+1}{3}\frac{(2sb-2g+1)^{b+2}}{2s+1} + \binom{b}{2}  \binom{2s+1}{2}^2 \frac{(2sb-2g+1)^{b+2}}{(2s+1)^2}, \\
	& \;\; \vdots \\
	C_a & = (-1)^a
		\sum_{\lambda \vdash a} \binom{b}{\{ 
			\lambda^T_i - \lambda^T_{i+1} \}_{i\geq 1} }
		\Bigg( \prod_{i=1}^{\ell (\lambda)} \frac{1}{2s+1}\binom{2s+1}{\lambda_i + 1} \Bigg)
		(2sb-2g+1)^{a+b},
\end{split}
\end{equation}
where the multinomial coefficient is
\begin{equation}
	\binom{b}{\{  \lambda^T_i - \lambda^T_{i+1} \}_{i\geq 1} } \coloneqq \frac{b!}{(b - \ell (\lambda ))! \prod_{i\geq 1} (\lambda_i^T - \lambda_{i+1}^T)! }.
\end{equation}
We omit the dependence on $s,g,b$ for short. Applying the forward difference operators, only few of these terms are going to contribute, as $\Delta^{\nu}.l^{n} = 0$ whenever $\nu > n$. For us it means that only the terms for $a = 0, \dots, 2g$ contribute non-trivially. The expression $\Delta^{\nu}.l^{n}|_{l=1}$ can be explicitly computed in terms of Stirling numbers of the second kind.

\begin{lemma}
	Let $S(p,q)$ be the Stirling numbers of the second kind and let $h_m$ be the complete homogeneous symmetric polynomials of degree $m$. We have:
	\begin{equation}
		\Delta^{\nu}.l^n  \big{|}_{l=1}
		=
		\nu! \, h_{n-\nu}(x_1, \dots, x_{\nu + 1}) \big{|}_{x_i = i}
		=
		\nu! \, S(n+1,\nu+1).
	\end{equation}
\end{lemma}

\begin{proof}
	Apply iteratively the Leibniz rule for the forward difference operator
	\begin{equation*}
		(\Delta.(f\cdot g))(l) = (\Delta.f)(l) \cdot g(l) + f(l + 1) (\Delta.g)(l)
	\end{equation*}
	to obtain
	\begin{equation*}
		\frac{1}{\nu!}\Delta^{\nu}.l^n   = \sum_{0 \leq i_1 \leq \dots \leq i_{n-\nu} \leq \nu} \prod_{j=1}^{n-\nu} (l + i_j),
	\end{equation*}
	which by definition is the symmetric complete homogeneous polynomial $h_{n-\nu}(x_1, x_2, \dots, x_{\nu+1})$ evaluated at $x_i = l+i-1$. The further evaluation at $l=1$ yields $h_{n-\nu}(1, \dots, \nu+1)$, which is well-known in the literature to be the Stirling number of second kind $S(n+1,\nu+1)$.
\end{proof}

We obtain:

\begin{proposition}\label{prop:one:part:single:spin:HNs}
	For $r = 2s$, the free energy $F_{g,1}^{r,\theta}$ is given by 
	\begin{equation}
		F_{g,1}^{r,\theta}(e^x)
		= 
		\sum_{b = \ceil{g/s}}^\infty \frac{1}{2^g(2sb - 2g + 1)^2b!} \left( \sum_{a=0}^{2g} C_a S(2sb-a + 1, 2sb-2g + 1) \right) e^{(2sb - 2g+1) x},
	\end{equation}
	where the $S(p,q)$ are Stirling numbers of second kind, and the coefficients $C_a$ are expressed in \eqref{eqn:Ca}. Equivalently, the one-part spin Hurwitz numbers with $(r+1)$-completed cycles of genus $g$ and degree $\mu$ have the following closed formula:
	\begin{equation}
		h^{r,\theta}_{g;\mu}
		=
		\frac{1}{2^g \mu^2 b!} \left( \sum_{a=0}^{2g} C_a S(\mu + (2g-a), \mu) \right),
	\end{equation}
	where by Riemann--Hurwitz $b = \frac{2g - 1 + \mu}{r}$.
\end{proposition}

\begin{example}
	For instance, for $g=1$ we obtain for $\mu = 2sb - 1$:
	\begin{equation}
		h^{r,\theta}_{1;\mu}
		=
		\frac{1}{2 \mu^2 b!}
		\left(
			\mu^{b} S(\mu+2,\mu)
			- sb \mu^{b+1} S(\mu+1,\mu)
			+ \left(b \frac{s(2s-1)}{3} + \frac{b(b-1)}{2} s^2 \right) \mu^{b+2} S(\mu, \mu)
		\right).
	\end{equation}
	Stirling numbers of the form $S(\mu + \ell, \mu)$ for $\ell \in \N$ are related to the Stirling polynomials $S_l(x)$ by
	\begin{equation}
		S(\mu + \ell, \mu) = (-1)^\ell \cdot \binom{\mu + \ell}{\mu} S_{\ell}(-(\mu + 1)),
		\qquad\quad
		\left( \frac{t}{1 - e^{-t}} \right)^{x+1} = \sum_{\ell=0}^{\infty} S_{\ell}(x) \frac{t^\ell}{\ell!},
	\end{equation}
	from which we obtain the following closed formula:
	\begin{equation}\label{eqn:g1:n1:closed}
		h^{r,\theta}_{1;\mu}
		=
		\frac{s^2}{12} \frac{\mu^{b-1}}{(b-1)!} \left( \mu + \frac{1}{s} \right) \,.
	\end{equation}
\end{example}

In the following we compute the genus one spin Hurwitz numbers via topological recursion, implemented through the software \texttt{Mathema\-tica}. Employing the symmetric properties of the spectral curve \eqref{eqn:spin:Hurwitz:SC}, the Taylor expansion of the Galois involutions around each ramification points (cf. \cref{sec:TR:CohFTs}) reads:
\begin{equation}
	\sigma_i(z) = J^i \sigma(J^{-i} z),
	\qquad
	\sigma_0(z)
	=
	a_0 -(z-a_0)
	-(2s)^{\frac{1}{2s}} \tfrac{2s-3}{3}(z-a_0)^2
	-(2s)^{\frac{2}{2s}} \tfrac{(2s-3)^2}{9}(z-a_0)^3
	+\mc{O}\bigl((z-a_0)^4\bigr),
\end{equation}
where $a_i \coloneqq \frac{J^i}{(2s)^{\frac{1}{2s}}}$, $J \coloneqq e^{\frac{\iu\pi}{s}}$. As a consequence, the $(1,1)$-correlators turn out to be
\begin{equation}
	\omega_{1,1}^{r,\theta}(z)
	=
	d \sum_{i=0}^{s-1} \frac{1}{a_i} \, F\left( r;\frac{z}{a_i} \right),
	\qquad
	F(r;z)
	\coloneqq
	\frac{ z (z^4 -s z^2 + 1+s) }{24s (1 - z^2)^3}.
\end{equation}
A simple computation shows that
\begin{equation}
\begin{split}
	\omega_{1,1}^{r,\theta}(z)
	& =
	d \left(
		\left( \frac{sz}{1 - 2s z^{2s}} \frac{d}{dz} + 1 \right) \frac{s z^{2s-1}}{12(1-2s z^{2s})}
	\right) =
	d \sum_{b \ge 1} \frac{s^2}{12} \frac{\mu^{b-1}}{(b-1)!} \left( \mu + \frac{1}{s} \right) e^{\mu x(z)},
\end{split}
\end{equation}
where we set $\mu = 2sb -1$. In the last equation, we used the expansion of the Lambert function \eqref{eqn:Lambert:expansion} and some algebraic manipulation that we omit. In particular, this confirms \cref{conj:spin:HNs:TR} for $(g,n) = (1,1)$.

\section{Intermezzo: equivariant topological recursion}
\label{sec:equiv:TR}

In this section, we develop a theory of topological recursion for spectral curves with symmetries. The spectral curve \cref{eqn:spin:Hurwitz:SC} is a specific example of such a spectral curve with symmetries, and we will use the theory developed in this section to find the ELSV-type formula associated with this spectral curve in \cref{sec:spin:HNs:ELSV}. Note that the spectral curves considered here do not fit the mold of \cref{defn:spectral:curve}, nor of any other extension of topological recursion we are aware of.

\subsection{Global to local spectral curves with symmetries}

\begin{definition}\label{defn:equiv:SC}
	Let $G$ be a finite group. A \emph{$G$-equivariant spectral curve} $\mathcal{S} = (\Sigma,\phi, x,y,B, \chi, \upsilon, \beta)$ consists of
	\begin{itemize}
		\item a Riemann surface $\Sigma$ (not necessarily compact or connected) with a free action $ \phi \colon G \times \Sigma \to \Sigma $, which we will often write $\phi (\gamma,x) = \phi_\gamma x = \gamma x$;
		\item a function $x \colon \Sigma \to \C$ such that its differential $dx$ is meromorphic and has finitely many zeros $a_1,\dots,a_r$ that are simple;
		\item a meromorphic function $y \colon \Sigma \to \C$ such that $d y$ does not vanish at the zeros of $d x$;
		\item a symmetric bidifferential $B$ on $\Sigma \times \Sigma$;
		\item three one-dimensional representations $ \chi, \upsilon, \beta \colon G \to \C^\times $,
	\end{itemize}
	such that for any $ \gamma \in G$
	\begin{equation}
		dx (\gamma z) = \chi_\gamma \, dx (z) \,, \qquad y (\gamma z) = \upsilon_\gamma \, y(z)\,, \qquad B(\gamma z_1,z_2) = \beta_\gamma \, B(z_1,z_2)\,,
	\end{equation}
	and $ B(z_1,z_2) - B^{G,\beta}(z_1,z_2) $ is regular as $ z_1 \to \gamma z_2$, where
	\begin{equation}
		B^{G,\beta}(z_1,z_2) \coloneqq \frac{1}{|G|} \sum_{\eta \in G} \beta_\eta^{-1} \frac{d (\eta z_1) dz_2}{(\eta z_1 - z_2)^2}\,.
	\end{equation}

	The \emph{topological recursion} is defined by the usual \cref{eqn:TR,eqn:TR:kernel}.
\end{definition}

\begin{example}
	The spectral curve in \cref{conj:spin:HNs:TR} is a $ \Z/2\Z$-equivariant spectral curve, with $ \phi_1 z = -z$, $ \chi $ the trivial representation, and $ \upsilon = \beta $ the sign representation. 
\end{example}

\begin{remark}
	For $ \Sigma = \P^1$, the usual $B(z_1,z_2)= \frac{dz_1 \, dz_2}{(z_1-z_2)^2} $ is invariant under simultaneous transformation of $ z_1$ and $z_2$ under $\Aut (\P^1) =  \mathord{PGL}(2,\C )$. The $B^{G,\beta}$ is only invariant under the centraliser of $ G$ in $ \Aut (\P^1)$.
\end{remark}

\begin{lemma}
	For any $G$-equivariant spectral curve as in \cref{defn:equiv:SC}, we have $\beta^2 = 1$, i.e. $ \beta \colon G \to \{ \pm 1\}$.
\end{lemma}
\begin{proof}
	It is sufficient to consider the polar part $ B^{G,\beta}$ as $z_1$ approaches any element of the $G$-orbit of $ z_2$. Let $ \gamma \in G$.
	\begin{equation*}
	\begin{split}
		\beta_\gamma B^{G,\beta}(z_1,z_2) &= B^{G,\beta}(\phi_\gamma z_1,z_2) = \frac{1}{|G|} \sum_{\eta \in G} \beta_\eta^{-1} \frac{d (\phi_\eta \phi_\gamma z_1) dz_2}{(\phi_\eta \phi_\gamma z_1 - z_2)^2}
		\\
		&= \frac{1}{|G|} \sum_{\eta \in G} \beta_\eta^{-1} \frac{d (\phi_\gamma \phi_{\gamma^{-1}\eta \gamma} z_1) dz_2}{(\phi_\gamma \phi_{\gamma^{-1}\eta\gamma} z_1 - z_2)^2}
		\\
		&\sim \frac{1}{|G|} \sum_{\eta' \in G} \beta_{\gamma \eta' \gamma^{-1}}^{-1} \frac{d (\phi_{\eta'} z_1) d(\phi_\gamma^{-1} z_2)}{(\phi_{\eta'} z_1 - \phi_\gamma^{-1}z_2)^2}
		\\
		&= B^{G,\beta}(z_1, \phi_\gamma^{-1} z_2) = \beta_\gamma^{-1} B^{G,\beta}(z_1,z_2)\,,
	\end{split}
	\end{equation*}
	using that $\beta$ is one-dimension, and where $ \sim$ means `equality up to regular terms'. From this it follows that $ \beta_\gamma^2 = 1$ for all $\gamma \in G$, so $ \beta^2 = 1$.
\end{proof}

Let $ a$ be a ramification point. Then $ dx(\gamma a) = \chi_\gamma dx(a) = 0$, so $\gamma a $ is a ramification point as well. Choose a local coordinate $ \zeta_e $ such that $ \zeta_e(a) = 0$ and $ x(z) = \zeta_e(z)^2 + x(a)$ around $z=a$ (we call such a coordinate \emph{adapted to $x$ at $a$}). Also choose a square root $ \sqrt{\chi}$, i.e. a function $ \sqrt{\chi} \colon G \to \C^\times $ such that $ (\sqrt{\chi})_\gamma^2 = \chi_\gamma$. We also require $ (\sqrt{\chi})_e = 1$.

\smallskip

Then define, near $\gamma a$,  $ \zeta_\gamma (z) = (\sqrt{\chi})_\gamma \zeta_e(\gamma^{-1}z)$. We get $ \zeta_\gamma(\gamma a) = 0$ and, as $ x(z) = \chi_\gamma x(\gamma^{-1} z) + C(z)$ for some locally constant $ C $, around $ z= \gamma a$ we get 
\begin{equation}
	x(z) = \chi_\gamma x(\gamma^{-1}z) + C(z) = \chi_\gamma \zeta_e(\gamma^{-1} z)^2 + x(a) +C(z) = \zeta_\gamma (z)^2 + x(\gamma a) \,.
\end{equation}
Therefore, $ \zeta_\gamma$ is a local coordinate adapted to $ x$ at $\gamma a$. Note that the ambiguity in the choice of $ \sqrt{\chi}$ corresponds to the fact that $ -\zeta_\gamma $ is also a local coordinate adapted to $x$ at $ \gamma a$.

\smallskip

If we define the \emph{times} as the Taylor coefficients of $ y$ in these coordinates via 
\begin{equation}
	y(z )\eqqcolon \sum_{k \geq 0} h^\gamma_k \zeta_\gamma^k(z) \quad \textup{near } z = \gamma a\,,
\end{equation}
then around $ z = \gamma a$,
\begin{equation}
	\sum_{k\geq 0} h^\gamma_k \zeta_-^k(z) = y(z)
	=
	\upsilon_\gamma y( \gamma^{-1} z)
	=
	\upsilon_\gamma \sum_{k \geq 0} h^e_k \zeta^k_e (\gamma^{-1}z)
	=
	\sum_{k \geq 0} (\upsilon_\gamma (\sqrt{\chi})_\gamma^{-k} h^e_k) \zeta_-^k (z)\,,
\end{equation}
so $ h^\gamma_k = \upsilon_\gamma  (\sqrt{\chi})_\gamma^{-k} h^e_k$.

\smallskip

Similarly, define \emph{jumps} as the Taylor coefficients of $B$ in these coordinates via 
\begin{equation}
	B(z_1, z_2)
	\eqqcolon
	\frac{\beta_{\gamma \eta}}{|G|} \frac{d\zeta_\gamma (z_1) \, d\zeta_\eta (z_2)}{(\zeta_\gamma (z_1) - \zeta_\eta (z_2))^2}
	+
	\sum_{k,l\geq 0} B^{\gamma,\eta}_{k,l} \zeta_\gamma^k(z_1) \zeta_\eta^l(z_2) d\zeta_\gamma(z_1) d\zeta_\eta (z_2)
	\quad \textup{near } z_1 = \gamma a, z_2 = \eta a\,.
\end{equation}
This expansion is justified by the assumption $ B = B^{G,\beta} + \textup{reg}$. Then by expanding around $ z_1 = \gamma a$ and $z_2 = \eta a$, we find $ B^{\gamma,\eta}_{k,l} = \beta_{\gamma\eta} (\sqrt{\chi})_\gamma^{-k-1} (\sqrt{\chi})_\eta^{-l-1} B^{e,e}_{k,l}$.

\smallskip

This analysis is local around $G$-orbits of ramification points. If we have $s|G|$ ramification points $ \{a_i \}_{i \in I}$ such that the index set $I$ has a free $G$-action $ G \times I \to I \colon (\gamma,i) \mapsto \gamma i$ and $ \gamma a_i = a_{\gamma i}$, then in the same way we can choose a system of representatives $ \{ a_i \}_{i \in \bar{I}} $ of the $G$-action and local coordinates $ \zeta_i$ adapted to $x$ around $ a_i$ such that $ \zeta_{\gamma i}(z) = (\sqrt{\chi})_\gamma \zeta_i (\phi_\gamma^{-1}z)$ for $ i \in \bar{I}$. And then with the expansions
\begin{align}
	y(z)
	&\eqqcolon
	\sum_{k \geq 0} h^i_k \zeta_i^k(z) \qquad\text{near } z = a_i; \\
	B(z_1, z_2)
	&\eqqcolon
	\frac{ \sum_{g \in G} \beta_{g} \delta_{i,gj} }{|G|}
		\frac{d\zeta_i(z_1) \, d\zeta_j(z_2)}{(\zeta_i(z_1) - \zeta_j(z_2))^2}
	+\!\!
	\sum_{k,l\geq 0} B^{i,j}_{k,l} \zeta_i^k(z_1) \zeta_j^l(z_2) d\zeta_i(z_1) d\zeta_j(z_2)
	\quad\text{near } z_1 = a_i, z_2 = a_j\, 
\end{align}
we get the relations
\begin{equation}\label{eqn:equiv:time:and:jump}
	h^{\gamma i}_k = \upsilon_\gamma (\sqrt{\chi})_\gamma^{-k} h^i_k\,,
	\qquad 
	B^{\gamma i,\eta j}_{k,l} = B^{\eta j, \gamma i}_{l,k} = \beta_{\gamma\eta} (\sqrt{\chi})_\gamma^{-k-1} (\sqrt{\chi})_\eta^{-l-1} B^{i,j}_{k,l}\,,
\end{equation}
for $ i,j \in \bar{I}$. This is the structure for the local $G$-equivariant spectral curve, as defined in the next subsection.

\subsection{The quotient of a local equivariant curve}

\begin{definition}\label{defn:local:equiv:SC}
	For a fixed $s \ge 1$, a \emph{local $G$-equivariant spectral curve $\mc{S}_G$} is given by the following data: we fix an index set $I $ of size $s|G|$, with free $G$-action and system of representatives $ \bar{I}$, and as above three one-dimensional representations $ \chi, \upsilon, \beta $ of $G$ such that $ \beta^2 = 1$, along with a square root $ \sqrt{\chi}$. The curve is $ \bigsqcup_{i \in I} \Spec \C \bbraket{\zeta_i}$, with $ \zeta_{\gamma i} = (\sqrt{\chi})_\gamma \zeta_i$ for $ i \in \bar{I}$, and
	\begin{equation}\label{eqn:Blocal:x:y}
		x(\zeta_i ) = \zeta_i^2 + \alpha_i,
		\qquad
		y(\zeta_i ) = \sum_{k \ge 0} h_{k}^i \zeta_i^k\,,
	\end{equation}
	and
	\begin{equation}\label{eqn:Blocal:B}
		B(\zeta_{i,1},\zeta_{j,2})
		=
		\frac{\sum_{\gamma \in G} \beta_\gamma \delta_{i, \gamma j}}{|G|} \frac{d\zeta_{i,1} d\zeta_{j,2}}{(\zeta_{i,1} - \zeta_{j,2})^2}
		+
		\sum_{k,l \ge 0} B^{i,j}_{k,l} (\zeta_{i,1})^k (\zeta_{j,2})^l d\zeta_{i,1} d\zeta_{j,2}
	\end{equation}
	with the conditions given in \cref{eqn:equiv:time:and:jump}.
\end{definition}

To lighten notation, we may write $ \zeta $ for any $ \zeta_i$ if the index is not important. For multidifferentials, we may also use an index in $ \{ 0,\dotsc, n\}$ to denote which argument we mean, so we explicitly choose $I \cap \N = \emptyset$ to avoid confusion. We define the topological recursion correlators as usual:
\begin{equation}\label{eqn:BTR}
	\omega_{g,n+1}(\zeta_0, \zeta_{\bbraket{n}})
	=
	\sum_{i \in I} \Res_{\zeta_i \to 0} K(\zeta_0,\zeta_i )
	\bigg(
		\omega_{g-1,n+2}(\zeta_i,-\zeta_i, \zeta_{\bbraket{n}})
		+ \!\!\!\!\!
		\sum_{\substack{g_1+g_2=g\\ J_1 \sqcup J_2 = \bbraket{n}}}^{\text{no }(0,1)} \!\!\!\!\!\!
			\omega_{g_1,1+|J_1|}(\zeta_i,\zeta_{J_1})
			\omega_{g_2,1+|J_2|}(-\zeta_i,\zeta_{J_2})\!
	\bigg),
\end{equation}
where
\begin{equation}\label{eqn:B:recursion:kernel}
	K(\zeta_0,\zeta_i) \coloneqq \frac{ \frac{1}{2} \int_{-\zeta_i}^{\zeta_i} B(\zeta_0, \mathord{\cdot})}{\bigl( y(\zeta_i) - y(-\zeta_i) \bigr) dx(\zeta_i)}\,.
\end{equation}

\begin{lemma}\label{lem:antisymmetry:correlators}
	All of the $ \omega_{g,n}$ constructed via \eqref{eqn:BTR} (i.e those with $ 2 g -2 + n >0$) are $G$-equivariant in any of the individual variables: for any $ 1 \leq m \leq n$, 
	\begin{equation}
		\omega_{g,n}(\zeta_{\bbraket{m-1}}, \zeta_{\gamma i,m}, \zeta_{\bbraket{m+1,n}}) = \beta_\gamma \omega_{g,n}(\zeta_{\bbraket{m-1}}, \zeta_{i,m}, \zeta_{\bbraket{m+1,n}})\,.
	\end{equation}
\end{lemma}

\begin{proof}
	This is a standard induction argument. Noting that $ \omega_{0,1} $ does not occur in \cref{eqn:BTR}, and the base case $ \omega_{0,2} = B$ follows immediately from \cref{eqn:Blocal:x:y,eqn:Blocal:B,eqn:equiv:time:and:jump}.\par
	For the induction step, there are two cases: if $m > 1$, it follows immediately from the induction hypothesis and the shape of \eqref{eqn:BTR}. If $ m =1$, it follows from \eqref{eqn:BTR} together with
	\begin{equation*}
		K(\zeta_{\gamma i,0},\zeta_j) 
		\coloneqq \frac{\int_{-\zeta_j}^{\zeta_j} B(\zeta_{\gamma i,0}, \mathord{\cdot})}{2(\omega_{0,1}(\zeta_j) - \omega_{0,1}(-\zeta_j))}
	 	= \frac{\int_{-\zeta_j}^{\zeta_j} \beta_\gamma B(\zeta_{i,0}, \mathord{\cdot})}{2(\omega_{0,1}(\zeta_j) - \omega_{0,1}(-\zeta_j))}
		= \beta_\gamma K(\zeta_{i,0},\zeta_j )\,.
	\end{equation*}
\end{proof}

\begin{definition}
	Let $ \mc{S}_G$ be a $G$-equivariant spectral curve. The \emph{$G$-quotient spectral curve} or \emph{reduced local spectral curve} $\mc{S}_G/G = \mc{S}_{\textup{red}}$ associated to $\mc{S}_G$ is the curve $ \bigsqcup_{i \in \bar{I}} \Spec \C \bbraket{\zeta_i}$, with the same $x$, $ y$, and $B$ restricted to this curve. 
	Then define the \emph{reduced correlators} $ \omega^{\textup{red}}_{g,n}$ via the usual (non-equivariant) local topological recursion on $\mc{S}^\textup{red}$. The reduced correlators can be extended to the full index range $ I$ by $G$-equivariance as in \cref{lem:antisymmetry:correlators}; we will denote these extended reduced correlators by the same symbols.
\end{definition}

\begin{proposition}\label{prop:B:equal:reduceded}
	For a local $G$-spectral curve $\mc{S}_G$, the correlators $\omega_{g,n}$ defined via \cref{eqn:Blocal:x:y,eqn:Blocal:B,eqn:equiv:time:and:jump,eqn:BTR,eqn:B:recursion:kernel} are zero, unless $ \chi \cdot \upsilon = \beta$, in which case they are equal to the extended reduced correlators $ \omega^{\textup{red}}_{g,n}$ defined via $\mc{S}_{\textup{red}}$, up to powers of $|G|$:
	\begin{equation}
		\omega_{g,n}(\zeta_{\bbraket{n}}) = |G|^{2g-2+n}\omega^{\textup{red}}_{g,n}(\zeta_{\bbraket{n}}) \,.
	\end{equation}
\end{proposition}

\begin{proof}
	This is again proved by induction. The main step is the reduction of the sum over $s|G|$ ramification points in \cref{eqn:BTR} to a sum over $s$ ramification points. For this, we need
	\begin{equation*}
		K(\zeta_0,\zeta_{\gamma i})
		=
		\beta_\gamma \chi_\gamma^{-1} \upsilon_\gamma^{-1}K(\zeta_0,\zeta_i )\,.
	\end{equation*}
	By \cref{lem:antisymmetry:correlators}, we may as well prove equality on $ \mc{S}_{\textup{red}}$, i.e. with indices restricted to $ \bar{I}$. Then, from the definitions, we see
	\begin{align*}
	\omega_{0,2}(\zeta_1,\zeta_2) &= \omega^{\textup{red}}_{0,2}(\zeta_1,\zeta_2)\,;\\
	K(\zeta_0,\zeta) &= K^{\textup{red}}(\zeta_0,\zeta)\,.
	\end{align*}
	For the induction step, we calculate
	\begingroup
	\allowdisplaybreaks
	\begin{align*}
		\omega_{g,n+1}(\zeta_0, \zeta_{\bbraket{n}})
		&=
		\sum_{i \in I} \Res_{\zeta_i \to 0}  K(\zeta_0,\zeta_i )
		\bigg(
			\omega_{g-1,n+2}(\zeta_i,-\zeta_i,\zeta_{\bbraket{n}})
			\\
			& \hspace{5cm}+ \!\!\!
			\sum_{\substack{g_1+g_2=g \\ J_1 \sqcup J_2 = \bbraket{n}}}
				\omega_{g_1,1+|J_1|}(\zeta_i,\zeta_{J_1})
				\omega_{g_2,1+|J_2|}(-\zeta_i,\zeta_{J_2})
		\bigg)\\
		&=
		\sum_{\substack{\gamma \in G\\ i\in \bar{I}}}   \Res_{\zeta_i \to 0} K(\zeta_0,\zeta_i ) \beta_\gamma \chi_\gamma^{-1} \upsilon_\gamma^{-1}
		\bigg(
			\beta_\gamma^2 \omega_{g-1,n+2}(\zeta_i,-\zeta_i,\zeta_{\bbraket{n}})
			\\
			&\hspace{5cm} + \!\!\!
			\sum_{\substack{g_1+g_2=g \\ J_1 \sqcup J_2 = \bbraket{n}}}
				\beta_\gamma \omega_{g_1,1+|J_1|}(\zeta_i,\zeta_{J_1})
				\beta_\gamma \omega_{g_2,1+|J_2|}(-\zeta_i,\zeta_{J_2})
		\bigg)\\
		&=
		\sum_{i \in \bar{I}} |G| \delta_{\chi \upsilon = \beta} \Res_{\zeta_i \to 0} K(\zeta_0,\zeta_i)
		\bigg(
			\omega_{g-1,n+2}(\zeta_i,-\zeta_i,\zeta_{\bbraket{n}})
			\\
			&\hspace{5cm} + \!\!\!
			\sum_{\substack{g_1+g_2=g \\ J_1 \sqcup J_2 = \bbraket{n}}}
				\omega_{g_1,1+|J_1|}(\zeta_i,\zeta_{J_1})
				\omega_{g_2,1+|J_2|}(-\zeta_i,\zeta_{J_2})
		\bigg)\\
		&=
		|G|^{2g-2 + (n+1)} \delta_{\chi \upsilon = \beta} \sum_{i \in \bar{I}} \Res_{\zeta_i \to 0} K^{\textup{red}}(\zeta_0, \zeta_i )
		\bigg(
		\omega^{\textup{red}}_{g-1,n+2}(\zeta_i,-\zeta_i,\zeta_{\bbraket{n}})
		\\
		&\hspace{5cm} + \!\!\!
		\sum_{\substack{g_1+g_2=g \\ J_1 \sqcup J_2 = \bbraket{n}}}
			\omega^{\textup{red}}_{g_1,1+|J_1|}(\zeta_i,\zeta_{J_1})
			\omega^{\textup{red}}_{g_2,1+|J_2|}(-\zeta_i,\zeta_{J_2})
		\bigg)\\
		&= 
		\delta_{\chi \upsilon = \beta} |G|^{2g-2 + (n+1)} \omega^{\textup{red}}_{g,n+1} (\zeta_0, \zeta_{\bbraket{n}})\,.
	\end{align*}
	\endgroup
	The first equality is the definition, the second applies equivariance, the third gathers characters. The fourth equality applies the induction hypothesis and gathers factors of $|G|$, and the last is the definition of the reduced correlators.
\end{proof}

\begin{corollary}
	The correlators calculated on $ \mc{S}_G$ are invariant under permutation of the variables.
\end{corollary}

\begin{proof}
	This symmetry is well-known for usual Eynard--Orantin topological recursion \cite{EO07}, for an algebraic proof, see \cite{ABCO24}. Therefore, it follows immediately from \cref{prop:B:equal:reduceded}.
\end{proof}

Because of \cref{prop:B:equal:reduceded}, we propose the following definition:

\begin{definition}
	Let $G$ be a finite group. A \emph{(local) $G$-spectral curve} is a (local) $G$-equivariant spectral curve as in \cref{defn:equiv:SC}, respectively \cref{defn:local:equiv:SC}, such that $ \chi \cdot \upsilon = \beta$.
\end{definition}
	
Thus, $ G$-spectral curves are those $G$-equivariant curves for which the stable correlators are not necessarily trivial due to \cref{prop:B:equal:reduceded}.

\section{An ELSV-type formula for spin Hurwitz numbers}
\label{sec:spin:HNs:ELSV}

The aim of this section is to apply the Eynard--DOSS correspondence to express spin single Hurwitz numbers with completed cycles as a certain CohFT on the moduli space of curves, intersected with powers of $\psi$-classes. Secondly, we will identify such CohFT as a 'signed' Chiodo class, or equivalently as the Chiodo class twisted by the Witten $2$-spin class. Throughout this section, we consider $r = 2s$ to be an even, positive integer.

\subsection{A CohFT for spin single Hurwitz numbers}
\label{subsec:spin:CohFT}
Recall the conjectural spectral curve on $\P^1$ for spin single Hurwitz numbers with $(r+1)$-completed cycles:
\begin{equation}
	x(z) = \log(z) -z^{2s},
	\qquad
	y(z) = z,
	\qquad
	B(z_1,z_2) = \frac{1}{2}\left( \frac{1}{(z_1 - z_2)^2} + \frac{1}{(z_1 + z_2)^2} \right) d z_1 d z_2.
\end{equation}
It is $\Z/2\Z$-spectral curve, with $\chi$ the trivial representation, and $\upsilon = \beta$ the sign representation. The ramification points of the full spectral curve are given by
\begin{equation}
	a_i = \frac{J^i}{(2s)^{\frac{1}{2s}}},
	\qquad
	J \coloneqq e^{\frac{\iu\pi}{s}},
	\qquad
	i = 0,1,\dots,2s-1,
\end{equation}
with $\Z/2\Z$-action given by $(-1).i \equiv i + s \pmod{s}$. Thus, we can choose the system of representative ramification points to be $\bar{I} \coloneqq \set{0,\dots,s-1}$. As a consequence of \cref{prop:B:equal:reduceded}, together with a rescaling of $B$ (cf. \cref{rem:B:rescaling:local:TR}), we get the following result.

\begin{lemma}\label{thm:spin:HNs:SC:red:norm}
	Consider the spectral curve on $\P^1$ given by
	\begin{equation}\label{eqn:spin:HNs:SC:red:norm}
		x(z) = \log(z) - z^{2s},
		\qquad
		y(z) = z,
		\qquad
		\hat{B}(z_1,z_2) = \left( \frac{1}{(z_1 - z_2)^2} + \frac{1}{(z_1 + z_2)^2} \right) d z_1 d z_2.
	\end{equation}
	Denote by $\hat{\omega}_{g,n}^{r,\theta}$ the multidifferentials obtained by summing over the ramification points indexed by $\{0,\dots,s-1\}$. Then
	\begin{equation}
		\omega_{g,n}^{r,\theta}(z_1,\dots,z_n)
		=
		2^{1-g-n} \hat{\omega}_{g,n}^{r,\theta}(z_1,\dots,z_n).
	\end{equation}
\end{lemma}
\begin{proof}
	Reducing the set of ramification points to $\bar{I}$ gives a factor of $2^{2g-2+n}$, while the rescaling $B \mapsto \hat{B} = \frac{1}{2}B$ gives a factor of $2^{-3g+3-2n}$. 
\end{proof}

The main consequence of the above lemma is that we can now apply the Eynard--DOSS correspondence, \cref{thm:Eyn:DOSS}. Indeed, although we do not sum over all zeros of $dx$, one can easily check that this does not spoil the argument of \cite{Eyn11,DOSS14}: the computation to express the multidifferentials as a sum over stable graphs is local, and the expression of the edge weights in terms of the $R$-matrix of \cref{eqn:R:matrix:TR} only requires a compact curve and $dx$ meromorphic.

\smallskip

Let us choose $C_i = - \frac{\iu}{\sqrt{4s}}$ and $C = (2s)^{\frac{1}{2s}+1}$, so that we have local coordinates $x = - \frac{1}{2}\frac{\zeta_i^2}{2s} + x(a_i)$ and $y = a_i + \frac{a_i}{2s} \zeta_i + \mathcal{O}(\zeta_i^2)$. With this choice, $\Delta_i = \frac{a_i}{2s}$ and $t_i = \frac{J^i}{2s}$. Moreover, the underlying topological field theory on $V = \C\braket{e_0,\dots,e_{s-1}}$ is given by
\begin{equation}
	\eta(e_i,e_j) = \delta_{i,j},
	\qquad
	\mathbb{1} = \sum_{i=0}^{s-1} \frac{J^i}{2s} e_i,
	\qquad
	\varpi_{g,n}^{r,\theta}(e_{i_1} \otimes \cdots \otimes e_{i_n})
	=
	\delta_{i_1,\ldots,i_n} \Bigl( \frac{J^{i}}{2s} \Bigr)^{-2g+2-n}.
\end{equation}
Let us compute now the other ingredients for the Eynard--DOSS formula (cf. \cite{LPSZ17} for similar computations).

\begin{lemma}
	The auxiliary functions are given by
	\begin{equation}
		\xi_i(z) = 2\Delta_i \frac{z}{a_i^2 - z^2},
		\qquad
		i = 0,\dots,s-1.
	\end{equation}
\end{lemma}
\begin{proof}
	We have
	\begin{equation*}
		\xi_i(z)
		=
		\int^z \frac{\hat{B}(\zeta_i,z)}{d \zeta_i} \bigg|_{\zeta_i = 0}
		=
		\Delta_i 
		\int^z
		\left(
			\frac{1}{(a_i - z)^2}
			+
			\frac{1}{(a_i + z)^2}
		\right) d z
		=
		2\Delta_i \frac{z}{a_i^2 - z^2}\,.\tag*{\qedhere}
	\end{equation*}
\end{proof}

\begin{lemma}\label{Rmatrix}
	The $R$-matrix is given by
	\begin{equation}
		R^{-1}(u)_{i}^{j}
		=
		\frac{1}{s}
		\sum_{k=0}^{s-1} J^{(2k+1)(j-i)} \exp\Biggl(
				- \sum_{m=1}^{\infty} \frac{B_{m+1}(\frac{2k + 1}{2s})}{m(m+1)} (-u)^{m}
			\Biggr)\,,
		\qquad
		i,j = 0,\dots,s-1\,,
	\end{equation}
	where $B_{n}(a)$ is the $n$-th Bernoulli polynomial.
\end{lemma}
\begin{proof}
	Inserting the expression for $\xi_i$ in the definition of the $R$-matrix and integrating by parts, we get
	\begin{equation*}
	\begin{split}
		R^{-1}(u)_{i}^{j}
		& =
		- \sqrt{\frac{u}{2\pi}} \int_{\R} d\xi_i(\zeta_j) e^{-\frac{1}{2u} \zeta_{j}^2} \\
		& =
		- \frac{2\Delta_i}{\sqrt{2\pi u}} \int_{\R}
			\frac{z(\zeta_j)}{a_i^2 - z^2(\zeta_j)} \, e^{-\frac{1}{2u} \zeta_{j}^2} \, \zeta_j \, d \zeta_j.
	\end{split}
	\end{equation*}
	We perform now the change of variable $\zeta_j \mapsto t$ determined by $z = \frac{J^j}{(2s)^{\frac{1}{2s}}} t^{\frac{1}{2s}}$. We find
	\begin{equation*}
		-\frac{1}{2}\frac{\zeta_j^2}{2s}
		=
		x - x(a_j)
		=
		- \frac{t}{2s} + \log\biggl( \frac{J^j}{(2s)^{\frac{1}{2s}}} t^{\frac{1}{2s}} \biggr) + \frac{1}{2s} - \log\biggl(\frac{J^j}{(2s)^{\frac{1}{2s}}}\biggr) 
		=
		\frac{1}{2s}\bigl( 1 - t + \log(t) \bigr),
	\end{equation*}
	and differentiating both sides we get $-\zeta d\zeta = (\frac{1}{t} - 1) d t$. Note also that $t$ runs along the Hankel contour $C$ when $\zeta_j$ runs from $-\infty$ to $+\infty$. As a consequence,
	\begin{equation*}
	\begin{split}
		R^{-1}(u)_{i}^{j}
		& =
		\frac{2\Delta_i}{\sqrt{2\pi u}} \int_{C}
			\frac{
				J^j (2s)^{-\frac{1}{2s}} t^{\frac{1}{2s}}
			}{
				J^{2i}(2s)^{-\frac{1}{s}} - J^{2j}(2s)^{-\frac{1}{s}} t^{\frac{1}{s}}
			}
			\,
			e^{\frac{1}{u}( 1 - t + \log(t))} \left( \frac{1}{t} - 1 \right)d t \\
		& =
		\frac{J^{j-i}}{s \sqrt{2\pi u}} \int_{C}
			t^{\frac{1}{2s} - 1}
			\frac{1 - t}{1 - J^{2(j-i)} t^{\frac{1}{s}}}
			\,
			e^{\frac{1}{u}( 1 - t + \log(t))} d t.
	\end{split}
	\end{equation*}
	On the other hand, we have the geometric progression formula
	\begin{equation*}
	\begin{split}
		\frac{1 - t}{1 - J^{2(j-i)} t^{\frac{1}{s}}}
		& =
		\sum_{k=0}^{s-1} J^{2k(j-i)} t^{\frac{k}{s}}
	\end{split}
	\end{equation*}
	and the integral expression for the reciprocal of the Gamma function, together with its asymptotic expansion as $u \to 0$ involving Bernoulli polynomials:
	\begin{equation*}
		\frac{1}{\Gamma(a - u^{-1})}
		=
		\frac{\iu}{2\pi} \int_{C} (-\tau)^{\frac{1}{u} - a} e^{-\tau} \, d\tau
		\sim
		\frac{(-u)^{-\frac{1}{u} + a - \frac{1}{2}} e^{-\frac{1}{u}}}{\sqrt{2\pi}} \exp\Biggl(
			\sum_{m=1}^{\infty} \frac{B_{m+1}(a)}{m(m+1)} u^{m}
		\Biggr).
	\end{equation*}
	Thus, we get
	\begin{equation*}
	\begin{split}
		R^{-1}(u)_{i}^{j}
		& =
		\frac{e^{\frac{1}{u}}}{s \sqrt{2\pi u}} \sum_{k=0}^{s-1} J^{(2k+1)(j-i)} \int_{C}
			t^{\frac{2k + 1}{2s} - 1 + \frac{1}{u}}
			e^{-\frac{t}{u}} d t \\
		& =
		\frac{1}{s}\sqrt{\frac{u}{2\pi}}
		(-u)^{\frac{2k + 1}{2s} - 1 + \frac{1}{u}} e^{\frac{1}{u}}
		\sum_{k=0}^{s-1} J^{(2k+1)(j-i)} \int_{C}
			\tau^{\frac{2k + 1}{2s} - 1 + \frac{1}{u}}
			e^{-\tau} d \tau
		\qquad\quad\text{changing } t = u \tau \\
		& \sim
		\frac{1}{s}
		\sum_{k=0}^{s-1} J^{(2k+1)(j-i)} \exp\Biggl(
				\sum_{m=1}^{\infty} \frac{B_{m+1}(1 - \frac{2k + 1}{2s})}{m(m+1)} u^{m}
			\Biggr).
	\end{split}
	\end{equation*}
	We conclude using the property $B_{m+1}(1-a) = (-1)^{m+1} B_{m+1}(a)$.
\end{proof}

\begin{lemma}
	The translation is given by
	\begin{equation}
		\hat{T}^i(u) = \sum_{m=1}^\infty \frac{B_{m+1}(\tfrac{1}{2s})}{m(m+1)} (-u)^m,
		\qquad
		i = 0,\dots,s-1.
	\end{equation}
	Moreover, $y$ and $B$ are compatible in the sense that \cref{eqn:y:B} is satisfied.
\end{lemma}
\begin{proof}
	The translation is given by \eqref{eqn:translation:TR} and therefore does not depend on the $(0,2)$-correlator, but only on $x$ and $y$. Therefore, we can use the result of \cite[lemma~4.1]{LPSZ17}, which, with the special case and change of notation $s= 1$ and $r = 2s$, computes the relevant integral for the spectral curve for non-spin Hurwitz number with completed cycles, with the same $x$ and $y$. The result is
	\begin{equation}
		\frac{1}{\sqrt{2 \pi u}} \int_{\R} dy e^{-\frac{\zeta_i^2}{2u}}
		\sim
		\Delta_i \exp \biggl( - \sum_{m=1}^\infty \frac{B_{m+1}(\tfrac{1}{2s})}{m(m+1)} (-u)^m  \biggr)\,,
	\end{equation}
	which, when divided by $\Delta_i$, gives the result.\par
	For the compatibility, with our choice of $C_i$ we need to check that
	\begin{equation}
		\frac{1}{\sqrt{2 \pi u}} \int_{\R} dy e^{-\frac{\zeta_j^2}{2u}}
		\sim
		\sum_{i=0}^{s-1} R^{-1}(u)^j_i \, \Delta_i.
	\end{equation}
	The \textsc{lhs} is given above, while the \textsc{rhs} follows from \cref{Rmatrix}:
	\begin{equation*}
	\begin{split}
		\sum_{i=0}^{s-1} R^{-1}(u)^j_i \, \Delta_i
		& \sim
		\frac{1}{s} \sum_{i=0}^{s-1} 
		\sum_{k=0}^{s-1} J^{(2k+1)(j-i)} \exp\Biggl(
				- \sum_{m=1}^{\infty} \frac{B_{m+1}(\frac{2k + 1}{2s})}{m(m+1)} (-u)^{m}
			\Biggr)
			\frac{J^{i}}{(2s)^{\frac{1}{2s}+1}} \\
		& = 
		\frac{1}{(2s)^{\frac{1}{2s}+1}} \sum_{k=0}^{s-1}
			\bigg( \frac{1}{s} \sum_{i=0}^{s-1}  J^{-2ki} \bigg)
			J^{(2k+1)j} \exp\Biggl(
					- \sum_{m=1}^{\infty} \frac{B_{m+1}(\frac{2k + 1}{2s})}{m(m+1)} (-u)^{m}
				\Biggr) \\
		& =
		\frac{J^j}{(2s)^{\frac{1}{2s}+1}} \exp\Biggl(
				- \sum_{m=1}^{\infty} \frac{B_{m+1}(\frac{1}{2s})}{m(m+1)} (-u)^{m}
			\Biggr) .
		\end{split}
	\end{equation*}
	So the two sides are indeed equal.
\end{proof}

Consider now the change of basis on the vector space underlying the cohomological field theory, from $(e_0,\dots,e_{s-1})$ to $(v_0,\dots,v_{s-1})$:
\begin{equation}
	v_a = \sum_{i=0}^{s-1} \frac{J^{(2a+1)i}}{2s} e_i,
	\qquad\qquad
	e_i = 2 \sum_{a=0}^{s-1} J^{-(2a+1)i} v_a.
\end{equation}
In the following lemma, the indices of the Kronecker deltas are taken modulo $s$.

\begin{lemma}\label{lem:CohFT:flat:basis}
	In the basis $(v_0,\dots,v_{s-1})$, the following holds.
	\begin{itemize}
		\item The underlying topological field theory is given by
		\begin{equation}\label{eqn:TFT}
			\eta(v_a,v_b) = \frac{1}{4s} \delta_{a+b+1}
			\qquad
			\varpi^{r,\theta}_{g,n}(v_{a_1} \otimes \cdots \otimes v_{a_n}) =
			\frac{(2s)^{2g-1}}{2} \, \delta_{a_1+\cdots+a_n-g+1}.
		\end{equation}
		Moreover, the unit is given by $v_0$.
		\item The $R$-matrix is given by
		\begin{equation}\label{eqn:R:matrix}
		 	R^{-1}(u) = \exp\Biggl(
				- \sum_{m=1}^{\infty} \frac{\mathrm{diag}_{a=0}^{s-1} \, B_{m+1}(\frac{2a + 1}{2s})}{m(m+1)} (-u)^{m}
			\Biggr).
		\end{equation}
		\item The auxiliary functions $\Xi_a \coloneqq 2 \sum_{i=0}^{s-1} J^{-(2a+1)i} \xi_i$ are given by
		\begin{equation}
			\Xi_a(z)
			=
			2 (2s)^{\frac{2s-2a-1}{2s}} \sum_{m = 0}^{\infty}
				\frac{(2sm + 2s-2a-1)^{m}}{m!} e^{(2sm+2s-2a-1)x(z)}.
		\end{equation}
	\end{itemize}
\end{lemma}

\begin{proof}
	The pairing is given by
	\[
	\begin{split}
		\eta(v_a,v_b)
		=
		\frac{1}{(2s)^2} \sum_{i,j = 0}^{s-1} J^{(2a+1)i + (2b+1)j} \eta(e_i,e_j)
		=
		\frac{1}{(2s)^2} \sum_{i = 0}^{s-1} J^{2(a+b+1)i}
		=
		\frac{1}{4s} \delta_{a+b+1}.
	\end{split}
	\]
	Similarly for the TFT and the $R$-matrix elements. 
	To conclude, let us compute the expansion of the auxiliary functions after the change of basis:
	\[
		\Xi_a(z)
		=
		2 \sum_{i=0}^{s-1} J^{-(2a+1)i} \xi_i(z)
		=
		\frac{4z}{(2s)^{\frac{1}{2s}+1}} \sum_{i=1}^{s-1} J^{-2ai} \frac{1}{a_i^2 - z^2}
		=
		\frac{2 ( (2s)^{\frac{1}{2s}} z )^{2s-2a-1}}{1-2s z^{2s}}.
	\]
	On the other hand, the spectral curve equation express $y = z$ in terms of $x$ through the Lambert $W$ function:
	\[
		z = \left( \frac{W(-2s e^{2s x})}{-2s} \right)^{\frac{1}{2s}}.
	\]
	In particular, from the relation $\frac{dz^{\alpha}}{dx} = \frac{\alpha \, z^{\alpha}}{1 - 2s z^{2s}}$ with $\alpha = 2s-2a-1$, we find that
	\[
		\Xi_a
		=
		\frac{2 (2s)^{\frac{2s-2a-1}{2s}}}{2s-2a-1} \frac{dz^{2s-2a-1}}{dx}
		=
		\frac{2}{(2s-2a-1)}
			\frac{d}{dx} \left( -W(-2s e^{2s x}) \right)^{\frac{2s-2a-1}{2s}} .
	\]
	We can now use the expansion given in \cref{eqn:Lambert:expansion} to finally get
	\begin{equation*}
	\begin{split}
		\Xi_a
		& =
		\frac{2}{s} \frac{d}{dx} \sum_{m = 0}^{\infty}
			\frac{(m + \frac{2s-2a-1}{2s})^{m-1}}{m!} (2s e^{2s x})^{m+\frac{2s-2a-1}{2s}} \\
		& =
		2 (2s)^{\frac{2s-2a-1}{2s}} \sum_{m = 0}^{\infty} \frac{(2sm + 2s-2a-1)^{m}}{m!} e^{(2sm+2s-2a-1)x}.
	\end{split}	\tag*{\qedhere}
	\end{equation*}
\end{proof}

We can now express the cohomological field theory $\Omega_{g,n}^{r,\theta} = R.\varpi_{g,n}^{r,\theta}$ as a sum over stable graphs. Because of the simple expression of the underlying topological field theory and $R$-matrix, the result can be expressed in a concise way through weighted stable graphs. The following definition is an adaptation of \cite[subsection~1.1]{JPPZ17} to our setting.

\begin{definition}
	Let $\Gamma \in \mc{G}_{g,n}$ be a stable graph of type $(g,n)$, and consider some weights $a = (a_1,\dots,a_n)$, $ 0 \leq a_i \leq s-1$, satisfying the modular constraint $\sum_{i=1}^n a_i \equiv g-1 \pmod{s}$. A \emph{spin weighting modulo $s$} of $\Gamma$ is a map $w \colon H_{\Gamma} \to \set{0,\dots,s-1}$ satisfying the following modular constraints.
	\begin{itemize}
		\item
		\textsc{Vertex conditions.} For every vertex $v \in V_{\Gamma}$,
		\begin{equation}
			\sum_{h \in H_{\Gamma}(v)} w(h) \equiv g(v) - 1 \pmod{s}.
		\end{equation}
		
		\item
		\textsc{Edge conditions.} For every edge $e = (h,h') \in E_{\Gamma}$,
		\begin{equation}
			w(h) + w(h') \equiv - 1 \pmod{s}.
		\end{equation}

		\item
		\textsc{Leaf conditions.} For every leaf $\lambda_i \in \Lambda_{\Gamma}$ corresponding to the marking $i \in \set{1,\dots,n}$,
		\begin{equation}
			w(\lambda_i) \equiv a_i \pmod{s}.
		\end{equation}
	\end{itemize}
	Denote by $W_{\Gamma}^{s,\theta}(a)$ the set of spin weighting modulo $s$ of $\Gamma$.
\end{definition}

\begin{proposition}\label{prop:CohFT:sum:stable:graphs}
	The CohFT $\Omega_{g,n}^{r,\theta} = R.\varpi_{g,n}^{r,\theta}$ is given by the following sum over stable graphs
	\begin{equation}\label{eqn:CohFT:sum:stable:graphs}
	\begin{split}
		\Omega_{g,n}^{r,\theta}(v_{a_1} \otimes \cdots \otimes v_{a_n})
		=
		2^{2g-2} \sum_{\Gamma \in \mc{G}_{g,n}} & \sum_{w \in W_{\Gamma}^{s,\theta}(a)}
		\frac{s^{2g-1-h^1(\Gamma)}}{|\mathrm{Aut}{(\Gamma)}|} \xi_{\Gamma,\ast}
		\prod_{\mathclap{v \in V_{\Gamma}}} \;
			\exp\Biggl(
				\sum_{m \ge 1} \frac{(-1)^{m}B_{m+1}(\frac{1}{2s})}{m(m+1)} \kappa_m(v)
			\Biggr) \\
		\times \; & \;
		\prod_{\mathclap{\substack{e \in E_{\Gamma} \\ e = (h,h')}}} \;
			\frac{1 - \exp\Bigl(
				- \sum_{m \ge 1} \frac{(-1)^m B_{m+1}(\frac{2w(h)+1}{2s})}{m(m+1)} \bigl( (\psi_{h})^m - (-\psi_{h'})^m \bigr)
			\Bigr)}{\psi_h + \psi_{h'}} \\
		\times \; & \;
		\prod_{\mathclap{\lambda_i \in \Lambda_{\Gamma}}} \;
			\exp\Biggl(
				- \sum_{m \ge 1} \frac{(-1)^m B_{m+1}(\frac{2a_i+1}{2s})}{m(m+1)} \psi_{\lambda_i}^m
			\Biggr)
	\end{split}
	\end{equation}
	if $\sum_{i=1}^n a_i \equiv g-1 \pmod{s}$, and zero otherwise.
\end{proposition}

\begin{proof}
	From \cref{lem:CohFT:flat:basis} and the definition of unit-preserving $R$-matrix action, we get
	\begin{equation*}
	\begin{split}
		\Omega_{g,n}^{r,\theta}(v_{a_1} \otimes \cdots \otimes v_{a_n})
		= \!
		\sum_{\Gamma \in \mc{G}_{g,n}} & \sum_{w \in W_{\Gamma}^{s,\theta}(a)}\!
		\frac{1}{|\mathrm{Aut}{(\Gamma)}|} \xi_{\Gamma,\ast}
		\prod_{\mathclap{v \in V_{\Gamma}}} \; \frac{(2s)^{2g(v)-1}}{2}
			\exp\! \Biggl(
				\sum_{m \ge 1} \frac{(-1)^{m}B_{m+1}(\frac{1}{2s})}{m(m+1)} \kappa_m(v) \!
			\Biggr) \\
		\times \; & \;
		\prod_{\mathclap{\substack{e \in E_{\Gamma} \\ e = (h,h')}}} \; (4s) \cdot 
			\frac{1 - \exp\Bigl(
				- \sum_{m \ge 1} \frac{(-1)^m B_{m+1}(\frac{2w(h)+1}{2s})}{m(m+1)} \bigl( (\psi_{h})^m - (-\psi_{h'})^m \bigr)
			\Bigr)}{\psi_h + \psi_{h'}} \\
		\times \; & \;
		\prod_{\mathclap{\lambda_i \in \Lambda_{\Gamma}}} \;
			\exp\Biggl(
				- \sum_{m \ge 1} \frac{(-1)^m B_{m+1}(\frac{2a_i+1}{2s})}{m(m+1)} \psi_{\lambda_i}^m
			\Biggr).
	\end{split}
	\end{equation*}
	if $\sum_{i=1}^n a_i \equiv g-1 \pmod{s}$, and zero otherwise. The spin weightings modulo $s$ are simply keeping track of the Kronecker deltas in the TFT and the $R$-matrix. Collecting the powers of $s$, we get the exponent
	\[
		|E_{\Gamma}| + \sum_{v \in V_{\Gamma}}(2g(v) - 1) = 2g-1-h^1(\Gamma),
	\]
	and collecting the powers of $2$, we find the exponent
	\[
		2|E_{\Gamma}| + 2\sum_{v \in V_{\Gamma}}(g(v) - 1) = 2g-2.\tag*{\qedhere}
	\]
\end{proof}

We have all the ingredients now to state the main result of the section.

\begin{theorem}[Spin ELSV formula]\label{thm:ELSV:CohFT}
	Conjecture~\ref{conj:spin:HNs:TR} is equivalent to the following statement. For $r=2s$ and $\mu = (\mu_1,\dots,\mu_n) \in \mc{OP}(d)$, the spin Hurwitz numbers are given by
	\begin{equation}\label{eqn:ELSV:CohFT}
		h_{g;\mu}^{r,\theta}
		=
		2^{1-g} r^{\frac{(r+1)(2g-2+n)+d}{r}}
		\left( \prod_{i=1}^n \frac{\left( \frac{\mu_i}{r} \right)^{[\mu_i]}}{[\mu_i]!}\right)
		\int_{\overline{\mathcal{M}}_{g,n}}
		\frac{\Omega_{g,n}^{r,\theta}(v_{\braket{\mu_1}} \otimes \cdots \otimes v_{\braket{\mu_n}})}{\prod_{i=1}^n(1 - \frac{\mu_i}{r}\psi_i)} \, .
	\end{equation}
	Here we wrote $\mu_i = r[\mu_i] + r - (2 \braket{\mu_i} + 1)$, with $0 \le \braket{\mu_i} \le s-1$.
\end{theorem}

\begin{corollary}
	By \cref{thm:AS}, \cref{eqn:ELSV:CohFT} holds unconditionally.
\end{corollary}

\begin{proof}[Proof of \cref{thm:ELSV:CohFT}]
	Combining \cref{thm:spin:HNs:SC:red:norm} and the Eynard--DOSS formula of \cref{thm:Eyn:DOSS} in the basis $(v_0,\dots,v_{s-1})$, we find that the topological recursion amplitudes computed on the spin Hurwitz numbers spectral curve are given by
	\[
		\omega^{r,\theta}_{g,n}(z_{\bbraket{n}})
		=
		2^{1-g-n} (2s)^{\frac{(2s+1)(2g-2+n)}{2s}} \sum_{a_1,\dots,a_n=0}^{s-1}
		\int_{\overline{\mathcal{M}}_{g,n}}
			\Omega^{r,\theta}_{g,n}(v_{a_1}\otimes \cdots \otimes v_{a_n})
			\prod_{j=1}^n \sum_{k_j \ge 0} \psi_j^{k_j} d\Xi_{k_j,a_j}(z_j).
	\]
	With our choice of constants $C_i$, we have $- \frac{1}{\zeta_i}\frac{d}{d\zeta_i} = \frac{1}{2s} \frac{d}{dx}$. Thus,
	\[
	\begin{split}
		\omega^{r,\theta}_{g,n}(z_{\bbraket{n}})
		& =
		d_1 \cdots d_n \!\!\!\!\! \sum_{\substack{0 \le a_1,\dots,a_n < s \\ k_1,\dots,k_n \ge 0}}
		\int_{\overline{\mathcal{M}}_{g,n}}\!\!\!\!\!
		\Omega^{r,\theta}_{g,n}(v_{a_1}\! \otimes \cdots \otimes v_{a_n})
		2^{1-g-n} (2s)^{\frac{(2s+1)(2g-2+n)}{2s}} \!
			\prod_{j=1}^n \left( \frac{\psi_j}{2s} \frac{d}{dx} \right)^{k_j} \!\!\! \Xi_{a_j}(z_j) \\
		& =
		d_1 \cdots d_n \!\!\! \sum_{\substack{0 \le a_1,\dots,a_n < s \\ k_1,\dots,k_n \ge 0}}
		\int_{\overline{\mathcal{M}}_{g,n}}\!\!\!\!
		\Omega^{r,\theta}_{g,n}(v_{a_1}\otimes \cdots \otimes v_{a_n})
			2^{1-g} (2s)^{\frac{(2s+1)(2g-2+n)+\sum_i(2s-2a_i-1)}{2s}} \\
		&\qquad\qquad\qquad\qquad\qquad \times
			\prod_{j=1}^n \left(\frac{\psi_j}{2s}\right)^{k_j} \sum_{m_j \ge 0} \frac{(2sm_j+2s-2a_j-1)^{m_j+k_j}}{m_j!} e^{(2sm_j+2s-2a_j-1)x}.
	\end{split}
	\]
	For $\mu$ odd, let us write $\mu = 2s[\mu] + 2s - (2\braket{\mu} + 1)$. Then, after substituting $\mu_j = 2sm_j+2s-2a_j-1$, the above formula reads
	\[
	\begin{split}
		\omega^{r,\theta}_{g,n}(z_{\bbraket{n}})
		& =
		d_1 \cdots d_n \!\!\!\! \sum_{\substack{\mu_1,\dots,\mu_n > 0 \\ \mu_i \text{ odd}}}
		\int_{\overline{\mathcal{M}}_{g,n}}
			\frac{
				\Omega^{r,\theta}_{g,n}(v_{\braket{\mu_1}}\otimes \cdots \otimes v_{\braket{\mu_n}})
			}{
				\prod_{j=1}^n(1 - \frac{\mu_j}{2s} \psi_j)
			}
			2^{1-g} (2s)^{\frac{(2s+1)(2g-2+n)+|\mu|}{2s}}
			\prod_{j=1}^n \frac{\left(\frac{\mu_j}{2s}\right)^{[\mu_j]}}{[\mu_j]!} e^{\mu_j x}.
	\end{split}
	\]
\end{proof}


\begin{example}[{The case $(g,n)=(1,1)$}]
	Let us explicitly write the CohFT for $g = n = 1$. The only non-trivial leaf decoration is $a = 0$. Moreover, there are two stable graphs of type $(1,1)$, namely 
	$\Gamma = \!\adjustbox{valign=c}{\begin{tikzpicture}[scale=.6]
		\draw (0,0) -- (-.8,0);
		\node at (0,0) [draw, shape=circle, inner sep=.4, fill=black, text=white] {\tiny$1$};
	\end{tikzpicture}}$ 
	and 
	$\Gamma' = \!\adjustbox{valign=c}{\begin{tikzpicture}[scale=.6]
		\draw (0,0) -- (-.8,0);
		\draw (.25,0) circle (.25cm);
		\node at (0,0) [draw, shape=circle, inner sep=.4, fill=black, text=white] {\tiny$0$};
	\end{tikzpicture}}$
	with automorphism groups of order $1$ and $2$ respectively. For $\Gamma$ there is no weighting to consider. For $\Gamma'$, the only admissible weighting for the edge $e=(h,h')$ is of the form $(a,s-1-a)$ for some $a = 0,\dots,s-1$. Thus, \cref{eqn:CohFT:sum:stable:graphs} reads
	\begin{equation}
		\Omega^{r,\theta}_{1,1}(v_a)
		=
		\delta_{a,0} \left[
			s
			\Bigl(1 - \tfrac{1}{2} B_{2}\bigl(\tfrac{1}{2s}\bigr) \kappa_1 \Bigr)
			\Bigr(1 +  \tfrac{1}{2} B_{2}\bigl(\tfrac{1}{2s}\bigr) \psi_1 \Bigr)
			+
			\frac{\xi_{\Gamma',\ast}(1)}{2} \Biggl( - \sum_{a=0}^{s-1} \tfrac{1}{2} B_{2}\bigl(\tfrac{2a + 1}{2s} \bigr) \Biggr)
		\right].
	\end{equation}
	Using the Bernoulli polynomials property $\sum_{a=0}^{s-1} B_{2}(\tfrac{2a + 1}{2s}) = -\frac{1}{12s}$ and the relations $\psi_1 = \kappa_1 = \frac{1}{24} \xi_{\Gamma',\ast}(1)$ valid in $H^{2}(\overline{\mathcal{M}}_{1,1})$, we obtain
	\begin{equation}
		\Omega^{r,\theta}_{1,1}(v_a)
		=
		\delta_{a,0} \left( s + \frac{\psi_1}{2s} \right).
	\end{equation}
	In particular, the ELSV for $g = n = 1$, $\mu = 2sb - 1$ reads
	\begin{equation}
	\begin{split}
		h_{1;\mu}^{r,\theta}
		&=
		4s^2 \frac{\mu^{b-1}}{(b-1)!}
		\int_{\overline{\mathcal{M}}_{1,1}}
			\left( s + \frac{\psi_1}{2s} \right) \left( 1 + \frac{\mu}{2s} \psi_1 \right) \\
		&= 
		4s^2 \frac{\mu^{b-1}}{(b-1)!}
		\left( \int_{\overline{\mathcal{M}}_{1,1}} \psi_1 \right)
		\left( \frac{\mu}{2} + \frac{1}{2s} \right)\\
		&=
		\frac{s^2}{12}\frac{\mu^{b-1}}{(b-1)!} \left( \mu + \frac{1}{s} \right)
	\end{split}
	\end{equation}
	which coincides with the expression of \cref{eqn:g1:n1:closed} obtained via Fock space computations.
\end{example}

\begin{remark}
	For the special case $r = 2$, i.e. $s = 1$, we can use the Bernoulli polynomials property $B_{m+1}(\tfrac{1}{2}) = (2^{-m} - 1) B_{m+1}$ and Mumford's formula to rewrite the CohFT $\Omega_{g,n}^{2,\theta}(\mathbb{1}^{\otimes n})$ as a product of two Hodge classes:
	\begin{equation}
		2^{2g-2} \,
		\exp\Biggl(
			- \sum_{m \ge 1} \frac{B_{m+1}(\tfrac{1}{2})}{m(m+1)} \biggl(
				\kappa_m - \sum_{i=1}^n \psi_i^m + \tfrac{1}{2}\iota_{\ast} \delta_m
			\biggr)
		\Biggr)
		=
		2^{2g-2} \, \Lambda (1) \Lambda (-\tfrac{1}{2}) \, ,
	\end{equation}
	where $\delta_m = \sum_{k + \ell = m-1} \psi^{k} (\psi')^{\ell}$, and $\iota$ is the boundary map. In particular, the ELSV for $\mu_i = 2b_i - 1$ reads
	\begin{equation}
		h_{g;\mu}^{2,\theta}
		=
		2^{4g-4+2n} \left( \prod_{i=1}^n \frac{\mu_i^{b_i-1}}{(b_i-1)!} \right)
		\int_{\overline{\mathcal{M}}_{g,n}}
		\frac{\Lambda (1) \Lambda (-\tfrac{1}{2})}{\prod_{i=1}^n \left( 1 - \frac{\mu_i}{2} \psi_i \right)}\,,
	\end{equation}
	expressing spin single Hurwitz numbers with $3$-completed cycles in terms of double Hodge integrals.
\end{remark}

\subsection{The CohFT as spin Chiodo class}
\label{sec:sign:Chiodo}

The (non-spin) $(r+1)$-completed cycles single Hurwitz numbers can be expressed in terms of intersection numbers of the Chiodo class \cite{Chi08}, by Zvonkine's $r$-ELSV formula \cite{Zvo06}, proved to be equivalent to topological recursion for these numbers in \cite{SSZ15}. The Chiodo class is a cohomological field theory defined via the universal bundle of the moduli space of curves with an $r$-spin bundle, i.e. an $r$-th root of the canonical bundle. Recently Chiodo class has been shown to be involved as cohomological representation, sometimes unexpectedly, of over fifteen enumerative problems arising from different branches of mathematics and physics. Several of these problems belong to Hurwitz theory, other concern for instance Masur--Veech volumes \cite{ABCDGLW23, CMS23}, double ramification cycle \cite{JPPZ17}, the Norbury class and volumes of moduli spaces of hyperbolic structures on super Riemann surfaces in relation with Jackiw--Teitelboim gravity \cite{SW20}. Several properties of the Chiodo class are investigated in \cite{GLN23}.

\smallskip

In this subsection, we show that the cohomological field theory found in \cref{prop:CohFT:sum:stable:graphs} can be interpreted in a similar way. The main difference, as more often in this paper, is the introduction of a sign. Let us recall the basic setup, which we specialise to our setting. It should be noted that there are three different, but equivalent, constructions of the compactified moduli space of spin bundles: by Jarvis \cite{Jar98,Jar00} using curves with $A_{r}$-singularities at the nodes and relatively torsion-free sheaves, by Caporaso--Casagrande--Cornalba \cite{CCC07} using iterated blowups at nodes (bubbled curves), and by Abramovich--Jarvis and Chiodo \cite{AJ03,Chi08a} using stacky curves with $ \mu_{r}$-automorphisms at the nodes. For an excellent introduction into this theory, see Chiodo--Zvonkine \cite{CZ09}.

\smallskip

We fix a positive integer $r$ as above. In \cite{Ols07,Chi08a}, the authors construct alternative compactifications of $\mc{M}_{g,n}$ adapted to spin structures. In particular, we will need the moduli space of $r$-stable curves, for which the main result for us is stated below.


\begin{theorem}[{\cite{Ols07,Chi08a}}]
	The moduli functor $ \M_{g,n}(r)$ of $r$-stable $n$-pointed curves of genus $g$ forms a proper, smooth, and irreducible Deligne--Mumford stack of dimension $3g - 3+ n$, and $ \M_{g,n}(1) = \M_{g,n}$. If $r' = l r$, there is a natural surjective, finite, and flat morphism $f^{r'}_r \colon \M_{g,n}(r') \to \M_{g,n}(r)$ that is invertible on the open dense substack of smooth $n$-pointed genus $g$ curves and yields an isomorphism between coarse spaces (and hence on Chow and cohomology). The morphism is not injective: indeed, the restriction to the substack of singular $r$-stable curves has degree $l$ as a morphism between stacks.
\end{theorem}

The exact definition of an $r$-stable curve is given in \cite[definition 2.1.1]{Chi08}. Informally, an $r$-stable curve is a nodal one-dimensional stack, whose coarse space is stable, whose smooth locus is an algebraic space, and whose nodes are described by
\begin{equation}
	\big[\big(\Spec (A[z,w]/(zw-t)\big) /\mu_{r} \big] \to \Spec A\,,
\end{equation}
where $ \mu_{r}$ is the group scheme of $ r$-th roots of unity, and an element $ \xi \in \mu_{r}$ acts by $ (z,w) \mapsto (\xi z, \xi^{-1} w)$. Fix $ k \in \Z$ and $n$ numbers $ 0 \leq b_i \leq r-1$ satisfying the modular constraint
\begin{equation}
	\sum_{i=1}^n b_i \equiv (2g-2+n)k \pmod{r}.
\end{equation}
For a genus $g$ curve with $n$ markings, define
\begin{equation}
	\mc{K} \coloneqq \omega_{\textup{log}}^k \Biggl( - \sum_{i=1}^n b_i [x_i] \Biggr) = \omega^k \Biggl( \sum_{i=1}^n (k-b_i) [x_i] \Biggr)\,,
\end{equation}
where $ \omega $ is the dualising sheaf. We will be interested in the moduli space of $ r$-th roots of $ \mc{K}$: line bundles $ \mc{L}$ such that $ \mc{L}^{r} \cong \mc{K}$. Here, and in the rest of this section, all powers of line bundles are tensor powers. In the smooth case, there is a natural torsor on $\M_{g,n}$ of curves with an $r$-th root of $\mc{K}$. There are different ways of compactifying these, but the most natural for us is the construction of \cite{Chi08a}: his compactification parametrises $r$-th roots of $\mc{K}$ on $r$-stable curves. The main properties of this moduli stack are given below.

\begin{theorem}[{\cite[theorem~4.6]{Chi08a}}] ~
	\begin{enumerate}
		\item
		The moduli functor $\M_{g,\underline{b}}^{r,k}$ of $r$-th roots of $\mc{K}$ on $r$-stable curves forms a proper and smooth stack of Deligne--Mumford type of dimension $3g-3+ n$. For $k = b_1 = \dotsb = b_n =0$, $\mc{K}$ is $\mc{O}$ and the stack $\M_{g,\underline{b}}^{r,k}$ is a group stack $\mc{G}$ on $\M_{g,n}(r )$. In general $\M_{g,\underline{b}}^{r,k}$ is a finite torsor on $ \M_{g,n}(r)$ under $\mc{G}$.

		\item
		The morphism $ p \colon \M_{g,\underline{b}}^{r,k} \to \M_{g,n}(r )$ is étale. It factors through a morphism locally isomorphic to the classifying stack $ B\mu_{r} \to \Spec \C$ (a $ \mu_{r}$-gerbe) and a representable étale $r^{2g}$-fold cover; therefore it is has degree $r^{2g-1}$.
	\end{enumerate}
\end{theorem}

\begin{remark}\label{rem:Chiodo:index:restriction}
	Note that we need not restrict the $b_i$ to lie between $ 0$ and $r-1$. However, there are canonical equivalences
	\begin{equation}
		i_{\underline{c}} \colon
		\M_{g,\underline{b}+r\underline{c}}^{r,k} \to \M_{g,\underline{b}}^{r,k} \colon
		(C,\underline{x},\mc{L}) \mapsto \Bigl(C, \underline{x}, \mc{L}\big({\textstyle\sum_{i=1}^n} c_i [x_i]\big) \Bigr)\,.
	\end{equation}
	We will use these maps implicitly to always reduce to $ 0 \leq b_i\leq r-1$. Similarly, there are canonical equivalences
	\begin{equation}
		\sigma \colon
		\M_{g,\underline{b}}^{r,k+ r} \to \M_{g,\underline{b}}^{r,k} \colon (C,\underline{x},\mc{L}) \mapsto (C,\underline{x}, \mc{L} \otimes \omega^{-1} )\,.
	\end{equation}
\end{remark}

\begin{remark}\label{rem:partial:spin:restriction}
	If $ r' = lr$, there is a map $ \epsilon^{r'}_{r} \colon \M_{g,\underline{b}}^{r'} \to \M_{g,\underline{b} \pmod{r}}^{r}$, which takes the $l$-th power of the line bundle and the $ B\mu_{r'}$-structure at the nodes (and incorporates the index restricting from \cref{rem:Chiodo:index:restriction}). This factors through a $ \mu_l$-gerbe and a representable $l^{2g}$-fold cover. We get $ p \circ \epsilon^{lr}_{r} = f^{l r}_r \circ p$, and in particular $ \epsilon^{r}_1 = f^{r}_1 \circ p$.
\end{remark}

These moduli spaces have a universal curve $\pi \colon \mc{C} \to \M_{g,\underline{b}}^{r,k} $ and a universal $r$-th root $ \mc{S} \to \mc{C}$. Moreover, there is a natural substack $ Z \subset \mc{C}$ consisting of the nodes of the singular curves. This has a double cover, $ j \colon Z' \to Z \subset \mc{C}$, parametrising nodes of singular curves with a choice of branch at the node. We denote its deck transformation, i.e. the map sending a node on a curve with chosen branch to the same node with opposite branch, by $ y \mapsto \bar{y}$. On $ Z'$, there are two natural line bundles whose fibres are the cotangent lines at the two branches of the node; we denote them $ \psi $ and $ \hat{\psi}$.

\smallskip

A node with chosen branch $y \in Z'$ has two natural lines over it: $\mc{S}|_y$ and $T_y^*\tilde{C}$, where $\tilde{C}$ is the normalisation of $C$, both of which are $ \mu_{r}$-representations (by the local structure of $r$-stable curves at a node). Let $ q(y) \in \Z/r\Z$ be the number defined by $ \mc{S}|_y = (T_y^*\tilde{C})^{q(y)}$ as representations. This is locally constant on $ Z'$, and we get a decomposition
\begin{equation}\label{eqn:sing:locus:first:decomp}
	Z' = \bigsqcup_{q \in \Z/r\Z} Z'_q \,,
	\qquad
	j_q = j|_{Z'_q} \colon Z'_q \to \M_{g,\underline{b}}^{r,k}\,.
\end{equation}

\begin{theorem}[{\cite[theorem~1.1.1]{Chi08}}]
	With notation as above,
	\begin{equation}
		(d+1)!\, \ch_d (R^\bullet \pi_*\mc{S})
		=
		B_{d+1}\Bigl(\frac{k}{r}\Bigr) \kappa_d
		- \sum_{i=1}^n B_{d+1}\Bigl(\frac{b_i}{r}\Bigr) \psi_i^d
		+ \frac{r}{2} \sum_{q=0}^{r-1} B_{d+1}\Bigl(\frac{q}{r}\Bigr) (j_q)_* ( \gamma_{d-1}) \, .
	\end{equation}
	Here $ \gamma_{d-1} = \sum_{i+j=d} (-\psi)^i \hat{\psi}^j$.
\end{theorem}

Recalling that generally
\begin{equation}\label{eqn:Chern:class:from:char}
	c(-E^\bullet ) = \exp \biggl( \sum_{d=1}^\infty (-1)^d(d-1)! \; \ch_d (E^\bullet )\biggr)\,,
\end{equation} 
it can be seen that the Chiodo class
\begin{equation*}
	C(r,k;b_1, \dotsc, b_n) \coloneqq (\epsilon^r_1)_* c(-R^\bullet \pi_* \mc{S}) \in H^{\bullet}(\M_{g,n})
\end{equation*}
is a cohomological field theory (for this it is important to restrict the $b_i$, cf. \cref{rem:Chiodo:index:restriction}). This is the cohomological field theory in \cref{thm:Zvonkine:conj}.

\smallskip

In our current setting, we want to understand what the correct geometric adaptation to the construction of some Chern class should be in order to obtain the sum over stable graphs \eqref{eqn:CohFT:sum:stable:graphs}. As more often in this paper, we have to introduce a sign, again given by the parity of a theta characteristic. For this, we assume $k=2l+1$ is odd, $r = 2s$ is even, and moreover $b_i = 2a_i +1$ is odd for all $i$. Recall the notion of theta characteristic from \cref{sec:spin:HNs}, and in particular its parity, \cref{defn:spin:structure:parity}.

\smallskip

Theta characteristics on a nodal curve were considered by Cornalba \cite{Cor89}, and are a particular case of $2$-spin curves in the bubbled curves framework of \cite{CCC07}, with all the $b_i = 1$. In \cite[section~6]{Cor89}, Cornalba indicates how to extend the proof of deformation invariance of the parity to nodal curves (still defined as $h^0(C;\vartheta ) \pmod{2}$) in this case. In our language, this shows that parity is a locally constant function on $ \M_{g,\underline{1}}^{2,1}$.

\begin{definition}\label{defn:spin:Chiodo}
	Let $ g,n \in \N$ such that $2g - 2 + n >0$, $ r = 2s \in \N_+$, $ k = 2l+1 \in \Z$, and $ 0 \leq  b_1, \dotsc, b_n \leq r-1 $ such that $ b_i = 2a_i+1$ and  
	\begin{equation}
		\sum_{i=1}^n (a_i - l) \equiv (g-1)(2l+1) \pmod{s}.
	\end{equation}
	For an $r$-spin curve $(C, \underline{x},\mc{L}) \in \M_{g,\underline{b}}^{r,k}$, we get that $\sigma^l\epsilon^{r}_2(\mc{L})$ is a theta characteristic, and we define its \emph{parity} to be the parity of this theta characteristic. As the parity is locally constant, we get a decomposition
	\begin{equation}
		\M^{r,k}_{g,\underline{b}} = \M_{g,\underline{b}}^{r,k,+} \sqcup \M_{g,\underline{b}}^{r,k,-}\,.
	\end{equation}
	Here, we identify $ \Z/2\Z \cong \{ \pm 1\}$ for clearer notation. We will use these superscripts more often to denote objects restricted to these subspaces. The \emph{spin Chiodo class} is defined as
	\begin{equation}
		C^{\theta}(r,k;\underline{a}) \coloneqq (\epsilon^{r,+}_1)_* c(- R^\bullet \pi^+_* \mc{S}^+) - (\epsilon^{r,-}_1)_* c(-R^\bullet \pi^-_* \mc{S}^-) \in H^{\bullet} (\M_{g,n})\,.
	\end{equation}
\end{definition}

\begin{remark}\label{rem:parity:index:restriction}
	Recall that we restricted the indices $ b_i$ to lie between $ 0$ and $ r-1$ in \cref{rem:Chiodo:index:restriction}. This is essential: the map $\epsilon^{r}_2$ `untwists' $ \mc{L}^{r}$ to a theta characteristic. In particular, it also `untwists' at the nodes: the $Z'_q$ from \cref{eqn:sing:locus:first:decomp} get mapped to $ Z'_{q \pmod{2}}$ in $ \M_{g,\underline{1}}^{2,1}$. We may then extend the $ a_i$ to all of $\Z$ in this definition, but we get canonical equivalences between the spaces for $ a_i$ and $ a_i+s$, which preserve the parity by construction. Note that we require $ b_i = 2a_i +1$ and $ k = 2l+1$ to be odd. This is necessary to obtain an honest theta characteristic, whose parity is well-behaved. For even $ b_i$ or $k$, there is no clear map to $ \mc{M}_{g, \underline{1}}^{2,1}$, where the bundles with theta characteristic live.
\end{remark}

As suggested by A.~Chiodo, the class $C^{\theta}$ can be expressed as a product of the usual Chiodo class with Witten $2$-spin class. In other words, spin Chiodo classes naturally `live' on $ \M_{g,\underline{1}}^{2,1}$, which makes some sense, as it is related via the ELSV formula \eqref{eqn:ELSV:CohFT} to spin Hurwitz numbers.

\begin{proposition}\label{prop:spin:Chiodo:as:Chiodo:times:Witten}
	In the situation of \cref{defn:spin:Chiodo}, the spin Chiodo class can be given by multiplying the usual Chiodo class with Witten $2$-spin class on the moduli space of $2$-spin curves:
	\begin{equation}\label{eqn:spin:Chiodo:as:Chiodo:times:Witten}
		C^{\theta}(r,k;\underline{a})
		=
		(\epsilon^2_1)_* \Big( c_{W,2} \cdot (\epsilon^{r}_2)_*c(-R^\bullet \pi_*\mc{S}) \Big)\,.
	\end{equation}
\end{proposition}
\begin{proof}
	By \cite[theorem~4.6]{JKV01}, cf. also \cite[section 6.3]{Chi06}, the $2$-spin Witten class on $ \M_{g,\underline{\mf{b}}}^{2,1}$ is non-zero if and only if $b_i$ is odd, and in that case it is given by $1$ on the component of even spin curves and $-1$ on the component of odd spin curves. Therefore
	\begin{equation}
	\begin{split}
		(\epsilon^2_1)_* \Big( c_{W,2} \cdot (\epsilon^{r}_2)_*c(-R^\bullet \pi_*\mc{S}) \Big) &= (\epsilon^2_1)_* (\epsilon^{r}_2)_* \Big( ( \epsilon^{r}_2)^* c_{W,2} \cdot c(-R^\bullet \pi_*\mc{S}) \Big)\\
		&= (\epsilon^{r}_1)_* \Big( c(-R^\bullet \pi_*\mc{S})^+ - c(-R^\bullet \pi_* \mc{S})^-\Big)\,,
	\end{split}
	\end{equation}
	which is the spin Chiodo class.
\end{proof}

\begin{remark}\label{p-SpinChiodo?}
	Proposition \ref{prop:spin:Chiodo:as:Chiodo:times:Witten} allows us to extend the definition of the spin Chiodo class to even $b_i$ by taking the right-hand side of \cref{eqn:spin:Chiodo:as:Chiodo:times:Witten} as the definition, as the class naturally vanishes in that case.	It is also natural to define a `$p$-spin Chiodo class' analogously, as was also suggested to us by A.~Chiodo, by
	\begin{equation}
		C^{p\textup{-spin}}(ps,k; \underline{b}) = (\epsilon^p_1)_*\Big( c_{W,p} \cdot (\epsilon^{ps}_p)_* c(-R^\bullet \pi_*\mc{S}) \Big)\,,
	\end{equation}
	which vanishes if any of the $ b_i $ is divisible by $p$. This class, or a related construction, could be useful for the local Gromov--Witten invariants of Lee--Parker \cite{LP07,LP13}, relating the Gromov--Witten invariants of Kähler surfaces to spin Hurwitz numbers. Note that for $ p>2$, Witten's class is not zero-dimensional anymore. We will not pursue this in this paper.
\end{remark}

As suggested by D.~Zvonkine, the CohFT of the previous subsection coincides with the spin Chiodo class.

\begin{proposition}
	Let $r$ be a positive even integer. The cohomological field theory of \cref{prop:CohFT:sum:stable:graphs} is equal to the spin Chiodo class for $k=1$:
	\begin{equation}
	\Omega_{g,n}^{r,\theta}(v_{a_1} \otimes \dotsb \otimes v_{a_n}) = 2^{g-1} C^{\theta}(r,1;a_1, \dotsc, a_n)\,.
	\end{equation}
\end{proposition}
\begin{proof}
	In the non-spin case, the corresponding expansion in stable graphs is given by \cite[corollary~4]{JPPZ17}. We emphasise here the changes coming from the inserted sign. Let us start from a slightly more refined formula for $\ch (R^\bullet \pi_* \mc{S})$, given in \cite[corollary~3.1.8]{Chi08}:
	\begin{equation}\label{eqn:Chern:char:alt}
	\begin{split}
		(d+1)! \, \ch_d (R^\bullet \pi_* \mc{S})
		=
		p^* \bigg(
			B_{d+1} \Bigl( \frac{1}{r} \Bigr) \kappa_d
			- \sum_{i=1}^n B_{d+1}\Bigl( \frac{b_i}{r} \Bigr) \psi_i^d
			+ \frac{r}{2} \sum_{\substack{0 \leq l \leq g\\ I \subseteq \bbraket{n}}}
				B_{d+1}\Bigl( \frac{q(l,I)}{r} \Bigr) (i_{(l,I)})_* (\gamma_{d-1}) \bigg) & \\
		+ \frac{r}{2}\sum_{q=0}^{r-1} B_{d+1}\Bigl( \frac{q}{r} \Bigr) (j_{(\textup{irr},q)})_* (\gamma_{d-1}) & \,.
	\end{split}
	\end{equation}
	In this formula, the singular locus $Z'$ has been decomposed more than in \cref{eqn:sing:locus:first:decomp}: consider a spin curve $(C,\underline{x},\mc{L})$ with a node $y$, denote by $\nu \colon \tilde{C} \to C$ the normalisation at $y$ and $\tilde{\mc{L}} \coloneqq \nu^{\ast}\mc{L}$.
	
	\smallskip

	If the node $y$ is \emph{separating}, denote $\tilde{C} = C_1 \sqcup C_2$ with the chosen first branch $C_1$ of genus $l$ and marked points $\underline{x}_I$ for some $I \subseteq \bbraket{n}$. Then the value of $q (l,I)$ is determined by $ 2l -2 + (|I|+1) - \sum_{i \in I} b_i + q (l,I)\in r\Z$ (and in particular must be odd), so $ Z'_{(l,I)} \subset Z'_{q(l,I)}$ and 
	\begin{equation*}
		\tilde{\mc{L}}^{r}|_{C_1} \cong \omega_{C_1,\textup{log}}\Big((q-r) [y] -\sum_{i\in I} b_i[x_i] \Big)\,.
	\end{equation*}
	We denote by $i_{(l,I)} \colon V'_{(l,I)} \to \M_{g,n}(r)$ the double cover that fits into the following diagram.
	\begin{equation*}
	\begin{tikzcd}
		Z'_{(\textup{l,I})} \arrow{r}{j_{(l,I)}}\arrow{d}[swap]{p_{(l,I)}} & \M^{r,k}_{g,\underline{b}} \arrow{d}{p} \\
		V'_{(l,I)} \arrow{r}[swap]{i_{(l,I)}} & \M_{g,n}(r)
	\end{tikzcd}
	\end{equation*}
	If the node $y$ is \emph{non-separating}, then the moduli point lies in $Z'_{(\textup{irr},q)} \subset Z'_q$ for some $q$, and we have
	\begin{equation*}
		\tilde{\mc{L}}^{r} \cong \omega_{\tilde{C},\textup{log}}\Big((q-r) [y] - q[\bar{y}] -\sum_{i=1}^n b_i[x_i] \Big) \, .
	\end{equation*}
	We denote by $j_{(\textup{irr},q)} \colon Z'_{(\textup{irr},q)} \to \M^{r,k}_{g,\underline{b}}$ the double cover of the relevant singular locus.
	
	\smallskip

	Now, let us start with the local picture, near a spin curve $ (C,\underline{x},\mc{L})$ with a node $y$, and let us write $ \tilde{\mc{L}} \coloneqq \nu^*\mc{L}$, $ \vartheta \coloneqq \epsilon^{r}_2 \mc{L}$ (which is a twist of $ \mc{L}^r$, recall \cref{rem:parity:index:restriction}) and $ \tilde{\vartheta} \coloneqq \nu^* \vartheta = \epsilon^{r}_2 \tilde{\mc{L}}$. The analysis is similar to \cite[examples 6.1 \& 6.2]{Cor89}, which is in the bubbled curve formalism.
	
	\smallskip

	If $y$ is separating, then $(C_1,\underline{x}_I,\tilde{\mc{L}} |_{C_1})$ and $(C_2,\underline{x}_{I^c},\tilde{\mc{L}} |_{C_2})$ are $r$-spin curves, and have parities. Moreover, although the $ \tilde{\mc{L}}|_{C_i}$ do not determine $\mc{L}$ (there are $r$ choices of glueing at the node), it is clear that $ H^0 (C; \vartheta ) \to H^0(C_1; \tilde{\vartheta}|_{C_1}) \oplus H^0(C_2;\tilde{\vartheta}|_{C_2}) $ is an isomorphism, so the parities add. \par
	If $y$ is non-separating, and the arithmetic genus of $ C$ is $g$, then that of $ \tilde{C}$ is $ g-1$. There are two distinct cases, given by the parity of $q$.
	\begin{itemize}
		\item
		If $q = 2u+1$ is odd, then $ \vartheta $ is a theta characteristic on $C$ such that $ \tilde{\vartheta} $ is a theta characteristic on $ \tilde{C}$, and in the same way as the previous case, $ h^0 (C;\vartheta ) = h^0 (\tilde{C};\tilde{\vartheta})$. Moreover, we can glue $ \tilde{\mc{L}}$ on $ \tilde{C}$ in $r$ ways.

		\item
		If $ q = 2u$ is even, we see that $ \tilde{\vartheta}$ is not quite a theta characteristic: in fact $ \tilde{\vartheta}^2 = \omega_{\tilde{C}} ([y] + [\bar{y}])$. By the argument of \cite[example 6.2]{Cor89}, such line bundles on $ \tilde{C}$ can be glued to theta characteristics on $C$ in two ways, and the resulting theta characteristics have opposite parity.
	\end{itemize}
	Rewriting this, we get that if $q$ is odd, then $ Z'_{(\textup{irr},q)} = (Z'_{(\textup{irr},q)})^+ \sqcup (Z'_{(\textup{irr},q)})^-$, and the maps $ j_{(\textup{irr},q)}$ preserve parity. However, if $ q$ is even, then $Z'_{(\textup{irr},q)}$ does not decompose, and the maps $ j_{(\textup{irr},q)}^{\pm} \colon Z_{(\textup{irr},q)} \to \M_{g,\underline{b}}^{r,1,\pm}$ are such that $ p_*^+ (j_{(\textup{irr},q)}^+)_* =  p_*^- (j_{(\textup{irr},q)}^-)_*$.

	\smallskip

	For the global calculation, we use \cite[proposition~4]{JPPZ17}: the class $ c(- R^\bullet \pi_* \mc{S})$ equals
	\begin{equation*}
	\begin{split}
		\sum_{\Gamma \in \mc{G}_{g,n}} & \sum_{w \in W_{\Gamma}^{r,1}(a)}
		\frac{r^{|E(\Gamma)|}}{|\Aut (\Gamma)|} (\xi_{\Gamma,w})_*
		\prod_{\mathclap{v \in V_{\Gamma}}}  C_V(v) 
		\prod_{\mathclap{\lambda_i \in \Lambda_{\Gamma}}} \;
			C_\Lambda (\lambda )\\
		\times \; & \;
		\prod_{\mathclap{\substack{e \in E_{\Gamma} \\ e = (h,h')}}} \;
			\frac{1 - \exp\Bigl(
				- \sum_{m \ge 1} \frac{(-1)^m B_{m+1}(\frac{w(h)}{r})}{m(m+1)} \bigl( (\psi_{h})^m - (-\psi_{h'})^m \bigr)
			\Bigr)}{\psi_h + \psi_{h'}}\,,
	\end{split}
	\end{equation*}
	where $C_V$ and $C_\Lambda$ are the same contributions as in \cref{eqn:CohFT:sum:stable:graphs}, and $W_{\Gamma}^{r,1}(a)$ is the set of $1$-weightings modulo $r$ of $\Gamma$ (see \cite[subsection~1.1]{JPPZ17}). The sum over weights $w$ encodes the values of $q$ at the different nodes. All of the terms of this formula are pulled back from the base, so in order to calculate the spin Chiodo class, we only need to analyse per pair $ (\Gamma, w)$ what part of the contribution lies in $ \M_{g,n}^{r,1,+}$ and what part in $\M_{g,n}^{r,1,-}$. We will use the local analysis for this.

	\smallskip

	If $w$ takes an even value at a half-edge $h$ (and therefore also at $ h'$), then we know this can only be on a non-separating edge, and by our previous argument, it can be glued in two ways with opposite parity. Hence, all of these contributions cancel, and we may restrict to odd values.\par
	If $w$ only takes odd values, then each of the $\M_{g(v),n(v)}^{r,1}$ actually splits according to parity, and taking parities $ \{ p(v) \}$ results in a spin curve with parity $ \prod_v p(v)$. So for fixed parities $ \{ p(v)\}$, we get
	\begin{equation*}
		\prod_{v \in V( \Gamma)} s^{2g(v)} 2^{g(v)-1} (2^{g(v)} + p(v))
	\end{equation*}
	curves. As each $r$-spin bundle on the stratum $ \Gamma $ has automorphism size $ r^{|V_{\Gamma}|}$, subtracting the number of odd-parity curves from the number of even-parity curves, we get
	\begin{equation*}
	\begin{split}
		\sum_{p(v) \in \{ \pm 1\} : v \in V_{\Gamma}} \prod_{v \in V_{\Gamma}} p(v) s^{2g(v)-1} 2^{g(v)-2} (2^{g(v)} + p(v)) 
		&= \prod_{v \in V_{\Gamma}} \sum_{p(v) \in \{ \pm 1\}} s^{2g(v)-1} 2^{g(v)-2} (1 + p(v) 2^{g(v)}) 
		\\
		&= \prod_{v \in V_{\Gamma}} s^{2g(v)-1} 2^{g(v)-1} 
	\end{split}
	\end{equation*}
	as the degree for $ (p^+ - p^-) $ on the stratum $(\Gamma,w)$. After pushforward, $s$ and $2$ occur a total of
	\begin{equation*}
		|E_{\Gamma}| + \sum_{v \in V_{\Gamma}} (2g(v) -1) = |E_{\Gamma}| - |V_{\Gamma}| + 2 \sum_{v \in V_{\Gamma}} g(v) = \big(h^1(\Gamma ) -1\big) + 2\big(g-h^1(\Gamma) \big) = 2g - 1 - h^1(\Gamma) 
	\end{equation*}
	and
	\begin{equation*}
		|E_{\Gamma}| + \sum_{v \in V_{\Gamma}} (g(v) -1) = |E_{\Gamma}| - |V_{\Gamma}| + \sum_{v \in V_{\Gamma}} g(v) = \big(h^1(\Gamma ) -1\big) + \big(g-h^1(\Gamma) \big) = g - 1
	\end{equation*}
	times, respectively. Comparing this with \cref{eqn:CohFT:sum:stable:graphs} yields the result.
\end{proof}

\begin{corollary}\label{cor:spin:ELSV:with:Chiodo}
	Conjecture~\ref{conj:spin:HNs:TR}, i.e. \cref{thm:AS}, is equivalent to the following spin ELSV formula: for $r=2s$ and for $\mu = (\mu_1,\dots,\mu_n) \in \mc{OP}(d)$, the spin Hurwitz numbers are given by
	\begin{equation}\label{SpinELSVFormula}
	\begin{split}
		h_{g;\mu}^{r,\theta}
		&=
		r^{\frac{(r+1)(2g-2+n)+d}{r}}
		\left( \prod_{i=1}^n \frac{\left( \frac{\mu_i}{r} \right)^{[\mu_i]}}{[\mu_i]!}\right)
		\int_{\overline{\mathcal{M}}_{g,n}}
		\frac{C^{\theta}(r,1;\braket{\mu_1}, \dotsc, \braket{\mu_n})}{\prod_{i=1}^n(1 - \frac{\mu_i}{r}\psi_i)}
		\\ 
		&=
		r^{\frac{(r+1)(2g-2+n)+d}{r}}
		\left( \prod_{i=1}^n \frac{\left( \frac{\mu_i}{r} \right)^{[\mu_i]}}{[\mu_i]!}\right)
		\int_{\overline{\mathcal{M}}_{g,\underline{1}}^{2,1}}
		\frac{c_{W,2} \cdot (\epsilon^{r}_2)_*c (-R^\bullet \pi_* \mc{S})}{\prod_{i=1}^n(1 - \frac{\mu_i}{r}\psi_i)}
	\end{split}
	\end{equation}
	Here we wrote $\mu_i = r[\mu_i] + r - (2 \braket{\mu_i} + 1)$, with $0 \le \braket{\mu_i} \le s-1$, and the domain of $ \epsilon^{r}_2$ is $ \M_{g,\underline{r-(2 \braket{\mu}+1)}}^{r,1}$.
\end{corollary}

\begin{remark}
	Analogously to the non-spin case, cf. \cite[conjecture~6.1]{KLPS19}, \cite[theorem~1.18]{DKPS23}, these formulae may be generalised to spin $q$-orbifold Hurwitz numbers $h^{r,q,\theta}_{g;\mu}$ with $(r+1)$-completed cycles, i.e. spin Hurwitz numbers with one ramification profile $ (q,q, \dotsc, q)$, one given by a partition $\mu$, and all other partitions being spin completed $(r+1)$-cycles. Then $q$, $r + 1 = 2s + 1$, and $\mu$ would need to be odd, and the cohomological field theory would be
	\begin{equation}
		C^{\theta}(rq,q;\braket{\mu_1}, \dotsc, \braket{\mu_n})\,,
	\end{equation}
	with now $ \mu_i = rq[\mu_i] +rq-(2\braket{\mu_i}+1)$, and $ 0 \leq \braket{\mu_i} \leq sq-1$. The spectral curve for this problem should be
	\begin{equation}
		x(z) = \log(z) - z^{rq}, 
		\qquad\quad 
		y(z) = z^q\,,
		\qquad\quad
		B(z_1, z_2) 
		=
		\frac{1}{2} \bigg( \frac{1}{(z_1-z_2)^2} + \frac{1}{(z_1+z_2)^2} \bigg) dz_1 dz_2 \,. 
	\end{equation}
	We do not pursue this generalisation here, although we do not expect any theoretical complications.
\end{remark}

\newpage

\appendix
\addtocontents{toc}{\protect\setcounter{tocdepth}{1}}

\section{Proof of fermion calculations}
\label{app:fermion:calcs}

In this appendix, we gather long calculations in the neutral fermion formalism.

\begin{proof}[Proof of \cref{prop:conj:operator}]
	 We compute iterated commutations
	\begin{align*}
		\ad_{\alpha_{1}^B} \mathcal{O}_{-\mu}^{B, r}(u) &= \sum_{k>-1}\sum_{l > \mu/2} (-1)^{l+1} f(l) \big(  \delta_{l-k-1} \phi_k \phi_{\mu-l} + \delta_{-k-1 + \mu - l} \phi_l \phi_k + \delta_{k + \mu-l} \phi_l\phi_{-k-1}  \big)
		 \\
		& = \sum_{l > \mu/2} (-1)^{l+1} f(l) \phi_{l-1} \phi_{\mu-l}
		+ \!\!\!  \sum_{ \mu/2 < l < \mu } \!\!\! (-1)^{l+1} f(l) \phi_l \phi_{\mu-l-1}
 		+ \!\!\! \sum_{l \geq \mu} (-1)^{l+1} f(l) \phi_l \phi_{\mu-l-1}
		\\
 		& = \!\!\! \sum_{ l > \mu/2}  (-1)^{l+1} [f(l) - f(l+1)] \phi_{l}\phi_{-l + (\mu-1)}  + (-1)^{\frac{\mu - 1}{2}}f\Big(\frac{\mu+1}{2}\Big) \phi_{\frac{\mu-1}{2}} \phi_{\frac{\mu-1}{2}}
		\\
  		& = \!\!\! \sum_{ l > \mu/2}  (-1)^{l+1} (-\Delta )f(l) \phi_{l} \phi_{-l + (\mu-1)}   + \frac{f(1)}{2}\delta_{\mu,1}\,,
	\end{align*}
	where we used that $ \phi_k \phi_k = \frac{\delta_{k,0}}{2}$. As the last term is central, we get for the second commutation
	\begin{align*}
		(\ad_{\alpha_{1}^B})^2 \mathcal{O}_{-\mu}^{B, r}(u)&= \sum_{k>-1}\sum_{l > \mu/2} (-1)^{l+1} (-\Delta )f(l) \big(  \delta_{l-k-1} \phi_k \phi_{\mu-1-l} + \delta_{\mu -k - 2 - l} \phi_l \phi_k + \delta_{k+\mu-1-l} \phi_l\phi_{-k-1}  \big)
		\\
		& = - \!\!\! \sum_{l > \mu/2 - 1} (-1)^{l+1} (-\Delta )f(l+1)\phi_l \phi_{-l + (\mu-2)} 
 		+\!\!\! \sum_{l > \mu/2 } \!\!\! (-1)^{l+1} (-\Delta )f(l)\phi_{l} \phi_{-l + (\mu-2)} 
		\\
 		& = \!\!\! \sum_{ l > \mu/2}  (-1)^{l+1} (-\Delta )^2 f(l) \phi_{l}\phi_{-l + (\mu-2)}  
 		+  (-1)^{\frac{\mu -1}{2}} (-\Delta )f\Big(\frac{\mu+1}{2}\Big) \phi_{\frac{\mu-1}{2}}\phi_{\frac{\mu-3}{2} }.
	\end{align*}
	Now, in the same way as before, the action of $\ad_{\alpha_{1}^B}$ on the first term of the latter expression is easy to compute:
	\begin{align*}
		& \ad_{\alpha_{1}^B}\sum_{ l > \mu/2}  (-1)^{l+1} (-\Delta )^2 f(l)\phi_{l} \phi_{-l + (\mu-2)}   
		\\
		\qquad
		&=  \sum_{ l > \mu/2}  (-1)^{l+1} (-\Delta )^3 f(l) \phi_{l}\phi_{-l + (\mu-3)}  
		+ (-1)^{\frac{\mu - 1}{2}}(-\Delta )^2 f\Big(\frac{\mu+1}{2}\Big)\phi_{\frac{\mu-1}{2}} \phi_{\frac{\mu-5}{2} } .
	\end{align*}
	However, the action of $\ad_{\alpha_{1}^B}$ on the second term this time is not trivial. In fact
	\begin{align*}
		& \ad_{\alpha_{1}^B}  (-1)^{\frac{\mu - 1}{2}}(-\Delta ) f\Big(\frac{\mu+1}{2}\Big) \phi_{\frac{\mu-1}{2}}\phi_{\frac{\mu-3}{2}} 
		\\
		&\qquad = (-1)^{\frac{\mu - 1}{2}} (-\Delta ) f\Big(\frac{\mu+1}{2}\Big) 
		\Big( \phi_{\frac{\mu-3}{2} }\phi_{\frac{\mu-3}{2}} (\delta_{\mu \notin \{1\}} + \delta_{\mu \in \{1\}})
		+
		\phi_{\frac{\mu-1}{2} } \phi_{\frac{\mu-5}{2}}(\delta_{\mu \notin \{1,3\}} + \delta_{\mu \in \{1,3\} }) \Big)
		\\
		&\qquad =
 		(-1)^{\frac{\mu - 1}{2}} (-\Delta ) f\Big(\frac{\mu+1}{2}\Big) 
		\Big(
		\frac{\delta_{\mu,3}}{2} + \phi_{\frac{\mu-1}{2} } \phi_{\frac{\mu-5}{2}}
		\Big).
	\end{align*}
	Putting the two contributions together, we get:
	\begin{align*}
		(\ad_{\alpha_{1}^B})^3 \mathcal{O}_{-\mu}^{B, r}(u)
		&= \sum_{ l > \mu/2}  (-1)^{l+1} (-\Delta )^3 f(l)  \phi_{l} \phi_{-l + (\mu-3)}
		\\
		&+
		 (-1)^{\frac{\mu - 1}{2}}
		\left(
 		(-\Delta )^2 f\Big(\frac{\mu+1}{2}\Big)\phi_{\frac{\mu-1}{2}} \phi_{\frac{\mu-5}{2}} 
 		+ 
 		(-\Delta ) f\Big(\frac{\mu+1}{2}\Big) 
		\Big(
		\frac{\delta_{\mu,3}}{2}
		+
		\phi_{\frac{\mu-1}{2} } \phi_{\frac{\mu-5}{2}}
		\Big)
		\right).
	\end{align*}
	In general
	\begin{align*}
		& \ad_{\alpha_{1}^B}\sum_{ l > \mu/2}  (-1)^{l+1} (-\Delta )^{t-1} f(l) \phi_{l}\phi_{-l + (\mu-(t-1))}   
		\\
		&=  \sum_{ l > \mu/2}  (-1)^{l+1} (-\Delta )^t f(l)  \phi_{l} \phi_{-l + (\mu-t)}
		 + 
 		(-1)^{\frac{\mu - 1}{2}} (-\Delta )^{t-1} f\Big(\frac{\mu+1}{2}\Big)\phi_{\frac{\mu-1}{2}} \phi_{\frac{\mu-1}{2} - (t-1)}
	\end{align*}
	and the commutation with the other terms is of the type
	\begin{equation*}
		\sum_{k>-1}(-1)^{k+1} \bigl[
			\phi_k\phi_{-k-1}, \phi_{\frac{\mu - 1}{2} - A}\phi_{\frac{\mu - 1}{2} - B}
		\bigr]
		=
		\phi_{\frac{\mu - 1}{2} - (A+1)}\phi_{\frac{\mu - 1}{2} - B} 
		+
		\phi_{\frac{\mu - 1}{2} - A}\phi_{\frac{\mu - 1}{2} - (B+1)} ,
	\end{equation*}
	where we should keep in mind that $ \phi_k \phi_k = \frac{\delta_{k,0}}{2}$. Hence in general we obtain
	\begin{equation}\label{AdjointActionIterated}
		\begin{split}
		(\ad_{\alpha_{1}^B})^t \mathcal{O}_{-\mu}^{B, r}(u)&= \sum_{ l > \mu/2}  (-1)^{l+1} (-\Delta )^t f(l) 			\phi_{l} \phi_{-l + (\mu-t)} 
		\\
		+ &
		(-1)^{\frac{\mu - 1}{2}}
		\sum_{v=0}^{t-1}
		(-\Delta )^{v} f\Big(\frac{\mu+1}{2}\Big) 
		\Big[
		\sum_{j=0}^{\lfloor t/2-1 \rfloor} c^t_{v,j} 
		\phi_{\frac{\mu-1}{2} - j}\phi_{\frac{\mu-1}{2} - (t-1) + j}  + \frac{\delta_{t,\mu}}{2} c^{t-1}_{v,\frac{t-3}{2}}
		\Big],
		\end{split}
	\end{equation}
	where $c^t_{v,j} = C_v\big( (t-1)-j-v,j \big) $ with the $C_m(n,k)$ elements of Catalan's trapezoids \cite{Reu14}. These are a generalisation of Catalan's triangle (which is $m=1$), and $ C_m(n,k)$ gives the number of paths in the two-dimensional square lattice, that start at the origin, have $n$ steps to the right, $k$ upwards, and do not cross $ y = x+ m-1$. Explicitly, they are given by
	\begin{equation*}
		C_m(n,k) = \begin{cases} \binom{n+k}{k} & k < m \\ \binom {n+k}{k} - \binom{n+k}{k-m} & m \leq k < n+m \\ 0 & \text{else} \end{cases}.
	\end{equation*}
	In our situation, the lattice coordinates are the arguments of the $ \phi$'s. They start at $ (l, \mu-l)$, in every commutation move either down or to the left, and cannot cross the diagonal (as $ \phi_k^2 $ is central). Explicitly, $ c^0_{v, -1} = C_v(-v,-1)= \binom{-v-1}{-1} = \delta_{v,0} $ is the only surviving term on the second line for $t=1$.

	\begin{lemma}
		For $ 0 \leq j \leq \lfloor t/2 \rfloor -1$, we have
		\begin{equation}
			\sum_{v=0}^{t-1} c^t_{v,j} (-\Delta )^v f\Bigl( \frac{\mu+1}{2} \Bigr) = (-1)^{j+1} (-\Delta )^t f \Bigl( \frac{\mu-1}{2} -j \Bigr).
		\end{equation}
	\end{lemma}
	\begin{proof}
		In this range, we may write $ c^t_{v,j} = \binom{t-1-v}{j} - \binom{t-1-v}{j-v}$, where the second term is zero if $ j <v$.
		Hence, using the Chu-Vandermonde identity,
		\begin{align*}
			\sum_{v=0}^{t-1} c^t_{v,j} (-\Delta )^v f\Bigl( \frac{\mu+1}{2} \Bigr) &= \sum_{v=0}^{t-1} \biggl( \binom{t-1-v}{j} - 			\binom{t-1-v}{j-v}\biggr) \sum_{k=0}^v (-1)^k  \binom{v}{k} f\Bigl( \frac{\mu+1}{2} +k\Bigr)
			\\
			&= \sum_{k=0}^{t-1} (-1)^k \sum_{v=k}^{t-1} \biggl( \binom{t-1-v}{j} - \binom{t-1-v}{j-v}\biggr)   \binom{v}{k} f\Bigl( \frac{\mu+1}{2} +k\Bigr)
			\\
			&= \sum_{k=0}^{t-1} (-1)^k \biggl( \binom{t}{j+k+1} - \binom{t}{j-k}\biggr) f\Bigl( \frac{\mu+1}{2} +k\Bigr)
			\\
			&= \sum_{m=j+1}^t \! \binom{t}{m} (-1)^{m+j+1} f\Bigl(\frac{\mu-1}{2}\!+ \! m\!- \! j\Bigr) + \! \sum_{m=0}^j \!\binom{t}{m} (-1)^{m+j+1} f\Bigl( \frac{\mu+1}{2} \! +\!  j \! - \! m\Bigr)
			\\
			&= \sum_{m=0}^t \binom{t}{m} (-1)^{m+j+1} f\Bigl(\frac{\mu-1}{2}\!+ \! m\!- \! j\Bigr) =(-1)^{j+1} (-\Delta )^t f \Bigl( \frac{\mu-1}{2} -j \Bigr)\,,
		\end{align*}
		where from the fourth to the fifth line, we used that $ f(\mu -l) =f(l)$.
	\end{proof}

	Substituting this into \cref{AdjointActionIterated}, we get
	\begin{align*}
		(\ad_{\alpha_{1}^B})^t \mathcal{O}_{-\mu}^{B, r}(u)&= \sum_{ l > \mu/2}  (-1)^{l+1} (-\Delta )^t f(l) \phi_{l} \phi_{-l + (\mu-t)}
		\\ 
 		&+  \sum_{j=0}^{\lfloor t/2-1 \rfloor} (-1)^{\frac{\mu+1}{2} + j} (-\Delta )^t f \Bigl( \frac{\mu-1}{2} -j \Bigr) 
		\phi_{\frac{\mu-1}{2} - j}\phi_{\frac{\mu-1}{2} - (t-1) + j}  + \frac{\delta_{t,\mu}}{2} (-\Delta )^{\mu-1} f ( 1 )
		\\
 		&= \sum_{ l = \frac{\mu+1}{2} - \lfloor \frac{t}{2} \rfloor }^\infty  (-1)^{l+1} (-\Delta )^t f(l) \phi_{l} \phi_{-l + (\mu-t)} + \frac{\delta_{t,\mu}}{2} (-\Delta )^{\mu-1} f ( 1 ).
	\end{align*}
	Summing over $t$ proves the proposition.
\end{proof}

\subsection{The \texorpdfstring{$(0,2)$}{(0,2)}-free energies}

The $(0,2)$-free energies is given in \cref{cor:free:energy:genus:zero}, from comparison with the non-spin case. In this appendix, we give an alternative proof via the operator formalism.

\smallskip

In this case, the difference between connected and non-connected Hurwitz numbers is relevant. The general inclusion-exclusion formula reads
\begin{equation}
	\langle \mc{O}_1 \cdots \mc{O}_n \rangle = \sum_{M \vdash [n]} \prod_{i=1}^{|M|} \langle \mc{O}_{M_i} \rangle^{\circ}, 
	\qquad \qquad 
	\langle \mc{O}_1 \cdots \mc{O}_n \rangle^{\circ} = \sum_{M \vdash [n]} (-1)^{|M| - 1} (|M| - 1)! \prod_{i=1}^{|M|} \langle \mc{O}_{M_i} \rangle
\end{equation}
where the correlators $\langle \cdot \rangle$ are possibly disconnected and the correlators are $\langle \cdot \rangle^{\circ} $ connected. For $n=2$ we simply have
\begin{equation}
	\langle{\mathcal{O}_1 \mathcal{O}_2\rangle}^{\circ}
	=
	\langle{\mathcal{O}_1 \mathcal{O}_2\rangle}
	-
	\langle{\mathcal{O}_1\rangle}\langle{ \mathcal{O}_2\rangle},
	\qquad
	\langle{ \mathcal{O}_i\rangle} = \langle{ \mathcal{O}_i\rangle}^{\circ}.
\end{equation}
Therefore, the relation between the two counting in the $(0,2)$-case can be expressed in terms of free energies and vacuum expectation values as
\begin{equation}
\begin{split}
	F_{0,2}^{r,\theta}(e^{x_1},e^{x_2}) 
	& =
	\sum_{\substack{\mu_1,\mu_2 > 0 \\ \textup{odd}}}
		h^{r,\theta}_{0;(\mu_1,\mu_2)} e^{\mu_1 x_1+\mu_2x_2}
	\\
	& =
	\sum_{\substack{\mu_1,\mu_2 > 0 \\ \textup{odd}}}
		\frac{2}{\mu_1\mu_2} [u^{\mu_1+\mu_2}] \bigg( \corr{\Big( e^{ \ad_{\alpha^B_{1}}} \mathcal{O}_{-\mu_1}^{B, r}(u)\Big)\Big( e^{ \ad_{\alpha^B_{1}}} \mathcal{O}_{-\mu_2}^{B, r}(u)\Big) }  
	\\
	& \qquad \qquad \qquad \qquad \qquad
	-\corr{e^{ \ad_{\alpha^B_{1}}} \mathcal{O}_{-\mu_1}^{B, r}(u) } \corr{e^{ \ad_{\alpha^B_{1}}} \mathcal{O}_{-\mu_2}^{B, r}(u) } \bigg) e^{\mu_1 x_1+\mu_2x_2}.
\end{split}
\end{equation}
In this case, Wick's theorem states
\begin{equation}
	\big\< \phi_j\phi_k\phi_l\phi_m \big\>
	=
	\< \phi_j\phi_k \big\> \big\< \phi_l\phi_m \big\> - \big\< \phi_j \phi_l \big\> \big\< \phi_k \phi_m \big\> + \big\< \phi_j \phi_m \big\> \big\< \phi_k \phi_l \big\> .
\end{equation}
The first of these three terms drops out against the purely disconnected term, along with any contributions from identity components. By \cref{lem:VEV:Fock:basis}, we can write
\begin{align}
	\big\< \phi_j\phi_k\phi_l\phi_m \big\>^\circ 
	&= - \big\< \phi_j \phi_l \big\> \big\< \phi_k \phi_m \big\> + \big\< \phi_j \phi_m \big\> \big\< \phi_k \phi_l \big\> \notag 
	\\
	&=(-\delta_{j+l} \delta_{k+m} + \delta_{j+m} \delta_{k+l} ) (-1)^{l+m} \Big(\delta_{l >0} + \frac{\delta_l}{2}\Big) \Big( \delta_{m>0}+ \frac{\delta_m}{2}\Big) \notag 
	\\
	&=\delta_{j+k+l+m} ( \delta_{j+m} - \delta_{j+l} )  (-1)^{l+m} \Big(\delta_{l >0} + \frac{\delta_l}{2}\Big) \Big( \delta_{m>0}+ \frac{\delta_m}{2}\Big).\label{ConnWickFourterm}
\end{align}
So
\begin{equation}
\begin{split}
	F_{0,2}^{r,\theta}(e^{x_1},e^{x_2}) 
	&=
	\!\! \sum_{\substack{\mu_1,\mu_2>0\\ \textup{odd}}} \!\!\! \frac{2[u^{\mu_1+\mu_2}]}{\mu_1\mu_2}
	\!\!\!\sum_{t_1,t_2=0}^\infty \sum_{ l_1 = \frac{\mu_1+1}{2} - \lfloor \frac{t_1}{2} \rfloor }^\infty
	\!\!\!\!\!\! (-1)^{l_1} \left( \frac{(-\Delta )^{t_1}}{t_1!} f_{\mu_1}\right)\!\!(l_1)
	\!\!\! \sum_{ l_2 = \frac{\mu_2+1}{2} - \lfloor \frac{t_2}{2} \rfloor }^\infty
	\!\!\!\!\!\! (-1)^{l_2} \left( \frac{(-\Delta )^{t_2}}{t_2!} f_{\mu_2} \right)\!\!(l_2) \\
	&\hspace{5cm}  \big\< \phi_{l_1} \phi_{-l_1 + (\mu_1-t_1)} \phi_{l_2} \phi_{-l_2 + (\mu_2-t_2)} \big\>^\circ e^{\mu_1 x_1+\mu_2x_2}\,.
\end{split}
\end{equation}
The first factor from \cref{ConnWickFourterm} becomes $ \delta_{\mu_1+\mu_2-t_1-t_2}$, so we can rewrite $ t_2 = \mu_1 + \mu_2 -t_1$. Furthermore, in the second factor, $ \delta_{j+l}$ becomes $ \delta_{l_1 + l_2}$, which is identically zero in the summation range. Hence, we can impose the other $\delta$, which is $ \delta_{l_1 - l_2 +t_1 -\mu_1}$, i.e. we can rewrite $ l_2 = l_1 +t_1 - \mu_1$. We also rename $ t_1 \to t$ and $ l_1 \to l$ to lighten notation.
\begin{equation}
\begin{split}
	F_{0,2}^{r,\theta}(e^{x_1},e^{x_2}) 
	&=\!\! \sum_{\substack{\mu_1,\mu_2>0\\ \textup{odd}}} \!\!\! \frac{2[u^{\mu_1+\mu_2}]}{\mu_1\mu_2}  \sum_{t=0}^{\mu_1+\mu_2} \sum_{ l = \frac{\mu_1+1}{2} - \lfloor \frac{t}{2} \rfloor }^\infty \!\!\!\!
	\left( \frac{(-\Delta )^{t}}{t!} f_{\mu_1}\right)\!\!(l)
	\left( \frac{(-\Delta )^{\mu_1+\mu_2-t}}{(\mu_1+\mu_2-t)!} f_{\mu_2} \right)\!\!(l+t-\mu_1) \\
	& \qquad \qquad \qquad
	\Big(\delta_{l + t - \mu_1 >0} + \frac{\delta_{l+t-\mu_1}}{2} \Big) \Big(\delta_{l < 0} + \frac{\delta_l}{2} \Big) e^{\mu_1 x_1+\mu_2x_2}\,.
\end{split}
\end{equation}
As in the case of the $(0,1)$-free energies, the number of $\Delta$'s equals the power of $ u$, which equals the sum of degrees of the polynomials the $\Delta$'s act on. Therefore, again only the leading terms are relevant, and we infer that $r = 2s$ divides $t$. Hence, we write $ t= 2sj$.
\begin{equation}
	F_{0,2}^{r,\theta}(e^{x_1},e^{x_2}) 
	=\!\!\!\! \sum_{\substack{\mu_1,\mu_2>0\\ \textup{odd}\\ 2s \mid \mu_1 +\mu_2}} \!\! \frac{2}{\mu_1\mu_2}  \!\! \sum_{j=0}^{\frac{\mu_1+\mu_2}{2s}} \sum_{ l = \frac{\mu_1+1}{2} - sj}^\infty \!\! \frac{\mu_1^j}{j!} \frac{\mu_2^{\frac{\mu_1+\mu_2}{2s}-j}}{(\frac{\mu_1+\mu_2}{2s}-j)!} \Big(\delta_{l + 2sj - \mu_1 >0} + \frac{\delta_{l+2sj-\mu_1}}{2} \Big) \Big(\delta_{l < 0} + \frac{\delta_l}{2} \Big) e^{\mu_1 x_1+\mu_2x_2}.
\end{equation}
In this formula, the only $l$-dependence of the terms is in the Kronecker $ \delta$'s. Hence we can calculate that part of the sum independently:
\begin{equation}\label{eqn:l:sum:deltas}
	\sum_{ l = \frac{\mu_1+1}{2} - sj}^\infty \Big(\delta_{l + 2sj - \mu_1 >0} + \frac{\delta_{l+2sj-\mu_1}}{2} \Big) \Big(\delta_{l < 0} + \frac{\delta_l}{2} \Big).
\end{equation}
Assuming $ l \leq 0$ by the second factor, we get that $ l \geq 2l \geq \mu_1+1 -2sj$, so the first factor equals $1$. Hence \cref{eqn:l:sum:deltas} reduces to
\begin{equation}
	\sum_{ l = \frac{\mu_1+1}{2} - sj}^\infty \Big(\delta_{l < 0} + \frac{\delta_l}{2} \Big) = \max \big\{ 0,sj-\tfrac{\mu_1+1}{2} \big\} + \frac{\delta_{2sj-\mu_1-1\geq0}}{2}.
\end{equation}
Plugging this back in, we get
\begin{equation}
\begin{split}
	F_{0,2}^{r,\theta}(e^{x_1},e^{x_2}) 
	&= \!\! \sum_{\substack{\mu_1,\mu_2>0\\ \textup{odd}\\ 2s \mid \mu_1 +\mu_2}} \sum_{j=\frac{\mu_1+1}{2s}}^{\frac{\mu_1+\mu_2}{2s}} \frac{\mu_1^{j-1}}{j!} \frac{\mu_2^{\frac{\mu_1+\mu_2}{2s}-j-1}}{(\frac{\mu_1+\mu_2}{2s}-j)!} (2sj-\mu_1) e^{\mu_1 x_1+\mu_2x_2}
	\\
	&= \!\!\!\! \sum_{\substack{\mu_1,\mu_2>0\\ \textup{odd}\\ 2s \mid \mu_1 +\mu_2}} \sum_{j=\frac{\mu_1+1}{2s}}^{\frac{\mu_1+\mu_2}{2s}} \bigg(2s \frac{\mu_1^{j-1}}{(j-1)!} \frac{\mu_2^{\frac{\mu_1+\mu_2}{2s}-j-1}}{(\frac{\mu_1+\mu_2}{2s}-j)!} -\frac{\mu_1^j}{j!} \frac{\mu_2^{\frac{\mu_1+\mu_2}{2s}-j-1}}{(\frac{\mu_1+\mu_2}{2s}-j)!}  \bigg) e^{\mu_1x_1+\mu_2x_2}\\
	&= \!\!\!\! \sum_{\substack{\mu_1,\mu_2>0\\ \textup{odd}\\ 2s \mid \mu_1 +\mu_2}} \frac{1}{(\frac{\mu_1+\mu_1}{2s} )!}\sum_{j=\frac{\mu_1+1}{2s}}^{\frac{\mu_1+\mu_2}{2s}} \bigg( \binom{\frac{\mu_1+\mu_2}{2s}\! -\! 1}{j\! -\! 1} \Big( \mu_1^j \mu_2^{\frac{\mu_1+\mu_2}{2s}-j-1} + \mu_1^{j-1}\mu_2^{\frac{\mu_1+\mu_2}{2s}-j}\Big) 
	\\[-.3cm]
	& \hspace{4cm} -  \binom{\frac{\mu_1+\mu_2}{2s}}{j}  \mu_1^j \mu_2^{\frac{\mu_1+\mu_2}{2s}-j-1} \bigg) e^{\mu_1 x_1+\mu_2x_2}\\
	&= \!\!\!\! \sum_{\substack{\mu_1,\mu_2>0\\ \textup{odd}\\ 2s \mid \mu_1 +\mu_2}} \frac{1}{(\frac{\mu_1+\mu_1}{2s} )!}\sum_{j=\frac{\mu_1+1}{2s}}^{\frac{\mu_1+\mu_2}{2s}} \bigg( 	\binom{\frac{\mu_1+\mu_2}{2s}-1}{j-1}  \mu_1^{j-1}\mu_2^{\frac{\mu_1+\mu_2}{2s}-j}
	\\[-.3cm]
	& \hspace{4cm} - \binom{\frac{\mu_1+\mu_2}{2s}-1}{j}  \mu_1^j \mu_2^{\frac{\mu_1+\mu_2}{2s}-j-1} \bigg) e^{\mu_1 x_1+\mu_2x_2}
	\\
	&=\!\!\!\! \sum_{\substack{\mu_1,\mu_2>0\\ \textup{odd}\\ 2s \mid \mu_1 +\mu_2}} \frac{2s}{\mu_1+\mu_1} \frac{\mu_1^{\lfloor \frac{\mu_1}{2s} \rfloor}}{\lfloor \frac{\mu_1}{2s} \rfloor !} \frac{\mu_2^{\lfloor \frac{\mu_2}{2s}\rfloor}}{\lfloor \frac{\mu_2}{2s} \rfloor!} e^{\mu_1 x_1+\mu_2x_2}.
\end{split}
\end{equation}
Comparing this to \cite[lemma~5.1]{KLPS19}, we see that this is the odd/antisymmetrised part of the non-spin case, i.e.
\begin{equation}
	F_{0,2}^{r,\theta}(e^{x_1},e^{x_2}) = \frac{1}{4} \big( F_{0,2}^r (e^{x_1},e^{x_2}) - F_{0,2}^r(-e^{x_1},e^{x_2}) - F_{0,2}^r(e^{x_1},-e^{x_2}) + F_{0,2}^r(-e^{x_1},-e^{x_2}) \big).
\end{equation}
Note that changing the sign of $e^x$ induces a sign change of $y=z$ as well, via the spectral curve equation $e^x = y e^{-y^{2s}}$. By \cite[theorem~5.2]{KLPS19}, we then get the following.

\begin{proposition}\label{prop:omega02}
	Conjecture~\ref{conj:spin:HNs:TR} holds for $(g,n) = (0,2)$:
	\begin{equation}
		d_1d_2F_{0,2}^{r,\theta} (e^{x_1},e^{x_2}) 
		-
		\frac{1}{2} \bigg( \frac{1}{(e^{x_1}-e^{x_2})^2} + \frac{1}{(e^{x_1} + e^{x_2})^2} \bigg) de^{x_1} de^{x_2}
		=
		\frac{1}{2} \bigg( \frac{1}{(z_1-z_2)^2} + \frac{1}{(z_1+z_2)^2} \bigg) dz_1 dz_2 \, .
	\end{equation}
\end{proposition}

\section{Numerics}
\label{app:numerics}

We append some numerics, computed via topological recursion. The numbers for $g=0$ and for $n=1$ agree with the neutral fermion computations of \cref{cor:free:energy:genus:zero} and \cref{prop:one:part:single:spin:HNs}.

\begin{table}[h]
  {
  \renewcommand{\arraystretch}{1.3}
  \hfill
  \begin{tabular}[t]{c|c|c}
    \toprule
    $g$ & $\mu$ & $h_{g;\mu}^{2,\theta}$ \\
    \midrule
    \multirow{5}{*}{$0$} & $1$ & $1$ \\
    & $3$ & $\tfrac{1}{3}$ \\
    & $5$ & $\tfrac{1}{2}$ \\
    & $7$ & $\tfrac{7}{6}$ \\
    & $9$ & $\tfrac{27}{8}$ \\
    \midrule
    \multirow{5}{*}{$1$} & $1$ & $\tfrac{1}{6}$ \\
    & $3$ & $1$ \\
    & $5$ & $\tfrac{25}{4}$ \\
    & $7$ & $\tfrac{343}{9}$ \\
    & $9$ & $\tfrac{2645}{13}$ \\
    \midrule
    \multirow{5}{*}{$2$} & $1$ & $\tfrac{1}{72}$ \\
    & $3$ & $\tfrac{13}{8}$ \\
    & $5$ & $\tfrac{5975}{144}$ \\
    & $7$ & $\tfrac{1409387}{2160}$ \\
    & $9$ & $\tfrac{2556603}{320}$ \\
    \bottomrule
  \end{tabular}
  \hfill
  \begin{tabular}[t]{c|c|c}
    \toprule
    $g$ & $\mu$ & $h_{g;\mu}^{4,\theta}$ \\
    \midrule
    \multirow{5}{*}{$0$} & $1$ & $1$ \\
    & $5$ & $\tfrac{1}{5}$ \\
    & $9$ & $\tfrac{1}{2}$ \\
    & $13$ & $\tfrac{13}{6}$ \\
    & $17$ & $\tfrac{289}{24}$ \\
    \midrule
    \multirow{5}{*}{$1$} & $3$ & $\tfrac{7}{6}$ \\
    & $7$ & $\tfrac{35}{2}$ \\
    & $11$ & $\tfrac{2783}{12}$ \\
    & $15$ & $\tfrac{11625}{4}$ \\
    & $19$ & $\tfrac{1694173}{48}$ \\
    \midrule
    \multirow{5}{*}{$2$} & $1$ & $\tfrac{1}{20}$ \\
    & $5$ & $\tfrac{451}{8}$ \\
    & $9$ & $\tfrac{84987}{20}$ \\
    & $13$ & $\tfrac{12793131}{80}$ \\
    & $17$ & $\tfrac{416853311}{96}$ \\
    \bottomrule
  \end{tabular}
  \hfill
  \begin{tabular}[t]{c|c|c}
    \toprule
    $g$ & $\mu$ & $h_{g;\mu}^{6,\theta}$ \\
    \midrule
    \multirow{5}{*}{$0$} & $1$ & $1$ \\
    & $7$ & $\tfrac{1}{7}$ \\
    & $13$ & $\tfrac{1}{2}$ \\
    & $19$ & $\tfrac{19}{6}$ \\
    & $25$ & $\tfrac{625}{24}$ \\
    \midrule
    \multirow{5}{*}{$1$} & $5$ & $4$ \\
    & $11$ & $\tfrac{187}{2}$ \\
    & $17$ & $\tfrac{3757}{2}$ \\
    & $23$ & $\tfrac{425845}{12}$ \\
    & $29$ & $\tfrac{7780091}{12}$ \\
    \midrule
    \multirow{5}{*}{$2$} & $3$ & $\tfrac{49}{12}$ \\
    & $9$ & $\tfrac{11109}{4}$ \\
    & $15$ & $\tfrac{2134515}{8}$ \\
    & $21$ & $14054082$ \\
    & $27$ & $\tfrac{88146516681}{160}$ \\
    \bottomrule
  \end{tabular}
  \hfill
  }
  \bigskip
  \caption{Some spin single Hurwitz numbers for $\ell(\mu) = 1$.}
  \label{table:spin:single:HNs:one:part}
\end{table}

\begin{table}
  {
  \renewcommand{\arraystretch}{1.3}
  \hfill
  \begin{tabular}[t]{c|c|c}
    \toprule
    $g$ & $(\mu_1,\mu_2)$ & $h_{g;\mu}^{2,\theta}$ \\
    \midrule
    \multirow{6}{*}{$0$}
    & $(1,1)$ & $1$ \\
    & $(3,1)$ & $\tfrac{9}{4}$ \\
    & $(5,1)$ & $\tfrac{125}{18}$ \\
    & $(3,3)$ & $\tfrac{27}{4}$ \\
    & $(5,3)$ & $\tfrac{375}{16}$ \\
    & $(5,5)$ & $\tfrac{3125}{36}$ \\
    \midrule
    \multirow{6}{*}{$1$}
    & $(1,1)$ & $\tfrac{5}{6}$ \\
    & $(3,1)$ & $\tfrac{17}{2}$ \\
    & $(5,1)$ & $\tfrac{925}{12}$ \\
    & $(3,3)$ & $\tfrac{99}{2}$ \\
    & $(5,3)$ & $\tfrac{1425}{4}$ \\
    & $(5,5)$ & $\tfrac{53125}{24}$ \\
    \bottomrule
  \end{tabular}
  \hfill
  \begin{tabular}[t]{c|c|c}
    \toprule
    $g$ & $(\mu_1,\mu_2)$ & $h_{g;\mu}^{4,\theta}$ \\
    \midrule
    \multirow{6}{*}{$0$}
    & $(3,1)$ & $3$ \\
    & $(7,1)$ & $\tfrac{49}{4}$ \\
    & $(5,3)$ & $\tfrac{75}{4}$ \\
    & $(9,3)$ & $\tfrac{243}{2}$ \\
    & $(7,5)$ & $\tfrac{1225}{12}$ \\
    & $(9,7)$ & $\tfrac{11907}{16}$ \\
    \midrule
    \multirow{6}{*}{$1$}
    & $(1,1)$ & $\tfrac{3}{2}$ \\
    & $(5,1)$ & $\tfrac{115}{2}$ \\
    & $(3,3)$ & $40$ \\
    & $(7,3)$ & $784$ \\
    & $(5,5)$ & $\tfrac{4025}{6}$ \\
    & $(7,7)$ & $\tfrac{30184}{3}$ \\
    %
    \bottomrule
  \end{tabular}
  \hfill
  \begin{tabular}[t]{c|c|c}
    \toprule
    $g$ & $(\mu_1,\mu_2,\mu_3)$ & $h_{g;\mu}^{2,\theta}$ \\
    \midrule
    \multirow{10}{*}{$0$}
    & $(1,1,1)$ & $4$ \\
    & $(3,1,1)$ & $12$ \\
    & $(5,1,1)$ & $50$ \\
    & $(3,3,1)$ & $36$ \\
    & $(3,3,3)$ & $108$ \\
    & $(5,3,1)$ & $150$ \\
    & $(5,3,3)$ & $450$ \\
    & $(5,5,1)$ & $625$ \\
    & $(5,5,3)$ & $1875$ \\
    & $(5,5,5)$ & $\tfrac{15625}{2}$ \\
    \bottomrule
  \end{tabular}
  \hfill
  }
  \bigskip
  \caption{Some spin single Hurwitz numbers for $\ell(\mu) = 2$ and $\ell(\mu) = 3$.}
  \label{table:spin:single:HNs:two:three:parts}
\end{table}

\bigskip
\printbibliography

\end{document}